\documentclass[sigconf, nonacm]{acmart}

\newcommand\vldbdoi{XX.XX/XXX.XX}
\newcommand\vldbpages{XXX-XXX}

\newcommand\vldbvolume{14}
\newcommand\vldbissue{9}
\newcommand\vldbavailabilityurl{http://vldb.org/pvldb/format_vol14.html}

\newcommand\vldbyear{2021}
\newcommand\vldbauthors{Raghavendra Addanki, Sainyam Galhotra, Barna Saha}
\newcommand\vldbtitle{\shorttitle} 
\usepackage{mathtools}
\usepackage{alltt}
\usepackage{url}
\usepackage{subfigure}
\usepackage{paralist}
\usepackage{listings}
\usepackage{color}

\usepackage{booktabs}
\usepackage{natbib}
\usepackage{enumitem}
\setlist[enumerate]{itemsep=-1mm}
\usepackage[utf8]{inputenc} 
\usepackage[T1]{fontenc}    
\usepackage{amsmath}
\usepackage{wrapfig}
\usepackage{xcolor}
\usepackage{xspace}
\usepackage[noend]{algpseudocode}
\usepackage{algorithm}
\usepackage{color}
\usepackage{amsfonts}
\usepackage{bbm}
\usepackage{amsthm}
\usepackage{tikz}
\usepackage{hyperref}       
\usepackage{url}            
\usepackage{booktabs}       
\usepackage{amsfonts}       
\usepackage{nicefrac}       
\usepackage{microtype}      

\DeclareMathOperator*{\oracle}{\mathcal{O}}
\DeclareMathOperator*{\E}{{\bf {E}}}

\DeclareMathOperator{\OPT}{OPT}

\newcommand{\children}{\textit{children}}

\newcommand{\vmax}{v_{\text{max}}}

\newcommand{\umax}{u_{\text{max}}}
\newcommand{\smax}{s_{\text{max}}}

\newcommand{\Count}{\texttt{Count}}
\newcommand{\FCount}{\texttt{FCount}}
\usepackage[utf8]{inputenc}

\newcommand{\MCount}{\texttt{MCount}}
\newcommand{\ACount}{\texttt{ACount}}

\newcommand{\identifycore}{\textsc{Identify-Core}}
\newcommand{\Assign}{\textsc{Assign}}
\newcommand{\ApproxFarthest}{\textsc{Approx-Farthest}\xspace}
\newcommand{\Yes}{\texttt{Yes}}
\newcommand{\No}{\texttt{No}}

\newtheorem{theorem}{Theorem}[section]
\newtheorem{lemma}[theorem]{Lemma}
\newtheorem{definition}[theorem]{Definition}

\newtheorem{observation}[theorem]{Observation}

\newtheorem{claim}[theorem]{Claim}
\newtheorem{example}[theorem]{Example}
\newtheorem{problem}[theorem]{Problem}

\usepackage{import}

\usepackage{tablefootnote}

\definecolor{mygreen}{rgb}{0,0.6,0}
\definecolor{mygray}{rgb}{0.5,0.5,0.5}
\definecolor{mymauve}{rgb}{0.58,0,0.82}

\newcommand{\reva}[1]{{#1}}
\newcommand{\revb}[1]{{#1}}
\newcommand{\revc}[1]{{#1}}

\newcommand{\changes}[1]{{#1}}

\newcommand{\cities}{\texttt{cities}}
\newcommand{\caltech}{\texttt{caltech}}
\newcommand{\amazon}{\texttt{amazon}}
\newcommand{\monument}{\texttt{monuments}}
\newcommand{\monuments}{\texttt{monuments}}

\begin{document}

\title{How to Design Robust Algorithms using Noisy Comparison Oracle}
\author{Raghavendra Addanki}
\affiliation{%
  \institution{UMass Amherst}
  \streetaddress{}
  \city{}
  \country{}
}
\email{raddanki@cs.umass.edu}

\author{Sainyam Galhotra}
\affiliation{%
  \institution{UMass Amherst}
  \streetaddress{}
  \city{}
  \state{}
  \postcode{}
}
\email{sainyam@cs.umass.edu}

\author{Barna Saha}
\affiliation{%
  \institution{UC Berkeley}
  \city{}
  \country{}
}
\email{barnas@berkeley.edu}

\begin{abstract}

Metric based comparison operations such as finding maximum, nearest and farthest neighbor are fundamental to studying various clustering techniques such as $k$-center clustering and agglomerative hierarchical clustering. These techniques crucially rely on accurate estimation of pairwise distance between records. However,  computing exact features of the records, and their pairwise distances is often challenging, and sometimes not possible. We circumvent this challenge by leveraging weak supervision in the form of a comparison oracle that compares the relative distance between the queried points such as `Is point $u$ closer to $v$ or  $w$ closer to $x$?'. 

However, it is possible that some queries are easier to answer than others using a comparison oracle. We capture this by introducing two different noise models called adversarial and probabilistic noise. In this paper, we study various problems that include finding maximum, nearest/farthest neighbor search under these noise models. Building upon the techniques we develop for these comparison operations, we give robust algorithms for $k$-center clustering and agglomerative hierarchical clustering. We prove that our algorithms achieve good approximation guarantees with a high probability and analyze their query complexity. We  evaluate the effectiveness and efficiency  of our techniques empirically on various real-world datasets.
\end{abstract}



\maketitle

\vspace*{-1ex}
\begingroup\small\noindent\raggedright\textbf{PVLDB Reference Format:}\\
\vldbauthors. \vldbtitle. PVLDB, \vldbvolume(\vldbissue): \vldbpages, \vldbyear.
\href{https://doi.org/\vldbdoi}{doi:\vldbdoi}
\endgroup
\begingroup
\renewcommand\thefootnote{}\footnote{\noindent
This work is licensed under the Creative Commons BY-NC-ND 4.0 International License. Visit \url{https://creativecommons.org/licenses/by-nc-nd/4.0/} to view a copy of this license. For any use beyond those covered by this license, obtain permission by emailing \href{mailto:info@vldb.org}{info@vldb.org}. Copyright is held by the owner/author(s). Publication rights licensed to the VLDB Endowment. 
\raggedright Proceedings of the VLDB Endowment, Vol. \vldbvolume, No. \vldbissue\ %
ISSN 2150-8097. \href{https://doi.org/\vldbdoi}{doi:\vldbdoi}
}\addtocounter{footnote}{-1}\endgroup

\ifdefempty{\vldbavailabilityurl}{}{
}
\vspace*{-1ex}

\section{Introduction\label{sec:intro}}
Many real world applications such as data summarization, social network analysis, facility location crucially rely on metric based comparative operations such as finding maximum, nearest neighbor search or ranking. As an example, data summarization aims to identify a small representative subset of the data where each representative is a summary of similar records in the dataset. Popular clustering algorithms such as $k$-center clustering  and hierarchical clustering are often used for data summarization~\cite{kleindessner2019fair, girdhar2012efficient}. In this paper, we study fundamental metric based operations such as finding maximum, nearest neighbor search, and use the developed techniques to study clustering algorithms such as $k$-center clustering and agglomerative hierarchical clustering.

Clustering is often regarded as a challenging task especially due to the absence of domain knowledge, and the final set of clusters identified can be  highly inaccurate and noisy~\cite{ben2018}. It is often hard to compute the exact features of points and thus pairwise distance computation from these feature vectors could be highly noisy. This will render the clusters computed based on objectives such as $k$-center unreliable.
\begin{center}
\vspace*{-2ex}
\begin{figure}
    \centering
    \includegraphics[width=0.77\columnwidth]{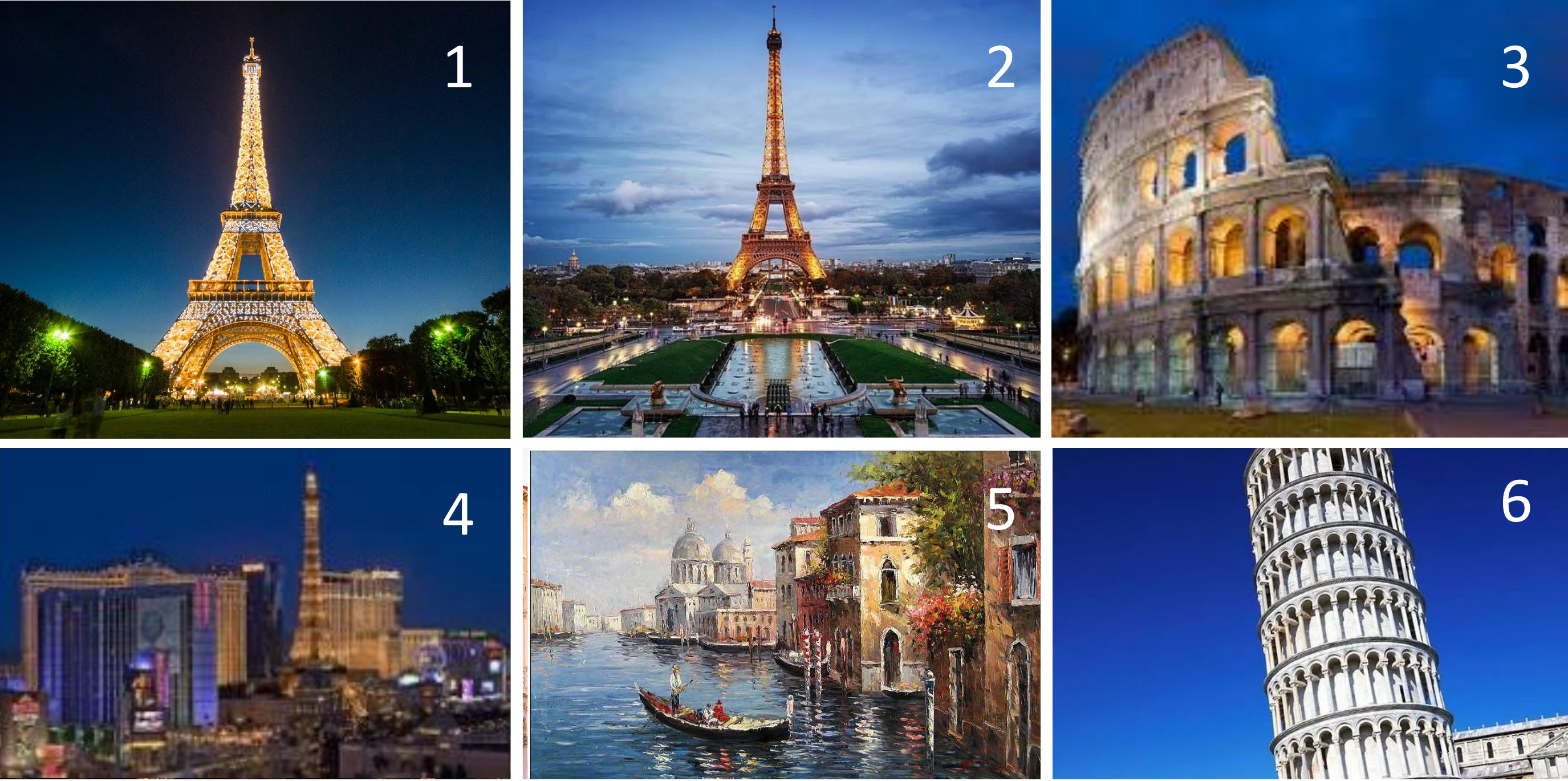}
    \vspace*{-2.5ex}
    \caption{Data summarization example}
    \vspace*{-3ex}
    \label{fig:example}
\end{figure}
\end{center}
To address these challenges, there has been a recent interest to leverage supervision from crowd workers (abstracted as an oracle) which  provides significant improvement in accuracy but at an added cost incurred by human intervention~\cite{verroios2015entity,galhotra2018robust,vinayak2016crowdsourced}. For clustering, the existing literature on oracle based techniques mostly use \emph{optimal cluster queries}, that ask questions of the form `do the points  u and v belong to the same optimal cluster?'\cite{ashtiani2016clustering,vinayak2016crowdsourced,mazumdar2017clustering,emamjomeh2018adaptive}. The goal is to minimize the number of queries aka query complexity while ensuring high accuracy of clustering output. This model is relevant for applications where the oracle (human expert or a crowd worker) is aware of the optimal clusters such as in entity resolution~\cite{verroios2015entity,galhotra2018robust}. However, in most applications, the clustering output depends highly on the required number of clusters and the presence of other records. Without a holistic view of the entire dataset, answering optimal queries may not be feasible for any realistic oracle. Let us consider an example data summarization task that highlights some of the challenges.
\begin{example}
Consider a  data summarization task over a collection of images (shown in Figure~\ref{fig:example}). The goal is to identify $k$  images (say $k=3$) that summarize the \reva{different locations in the dataset}.  The images $1,2$ refer to the Eiffel tower in Paris, $3$ is the Colosseum in Rome, $4$ is the replica of Eiffel tower at Las Vegas, USA, $5$ is Venice and $6$ is the Leaning tower of Pisa. The ground truth output in this case would be $\{\{1,2\},\{3,5,6\},\{4\}\}$. We calculated pairwise similarity between images using the visual features generated from Google Vision API~\cite{googlevision}. The pair $(1,4)$ exhibits the highest similarity of $0.87$, while all other pairs have similarity lower than $0.85$. Distance between a pair of images $u$ and $v$, denoted as $d(u,v)$, is defined as $(1-$similarity between $u$ and $v$).   We ran a user experiment by querying crowd workers to answer simple Yes/No questions to help summarize the data (Please refer to Section~\ref{sec:userstudy} for more details).
\end{example}
\noindent In this example, we make the following observations.
\setlength{\textfloatsep}{10pt}


\sloppy
\begin{description}[leftmargin=0pt]
\vspace{0mm}
\item[$\bullet$ Automated clustering techniques generate noisy clusters.] Consider the greedy approach for $k$-center clustering~\cite{gonzalez1985clustering} which sequentially identifies the farthest record as a new cluster center. In this example, records $1$ and $4$ are placed in the same cluster by the greedy $k$-center clustering, thereby leading to poor performance. \reva{In general, automated techniques are known to generate erroneous similarity values between records due to missing information or even presence of noise~\cite{wang2012crowder,vesdapunt2014crowdsourcing,firmanier}. } Even Google's landmark detection API~\cite{googlevision} did not identify the location of images $4$ and $5$.
\item[$\bullet$ Answering pairwise optimal cluster query is infeasible.] Answering whether \reva{$1$ and $3$} belong to the same optimal cluster when presented in isolation is impossible unless the crowd worker is aware of other records present in the dataset, \reva{and the granularity of the optimum clusters}. \reva{Using the pair-wise Yes/No answers obtained from the crowd workers for the ${6\choose 2}$ pairs in this example, the identified clusters achieved $0.40$ F-score for $k=3$. Please refer to Section~\ref{sec:userstudy} for additional details. }
\item[$\bullet$ Comparing relative distance between the locations is easy.] Answering \emph{relative distance} queries of the form \reva{ `Is $1$ closer to $3$, or is $5$ closer to $6$?' does not require any extra knowledge about other records in the dataset. For the $6$ images in the example, we queried relative distance queries and the final clusters constructed for $k=3$ achieved an F-score of $1$. }
\end{description}

\reva{In summary, we observe that humans have an innate understanding of the domain knowledge and can answer \emph{relative distance queries} between records easily. }
Motivated by {the aforementioned} observations, we consider a \emph{quadruplet} comparison oracle that  compares the relative distance between two pairs of points $(u_1,u_2)$ and $(v_1,v_2)$ and outputs the pair with smaller distance between them breaking ties arbitrarily.  
Such oracle models
have been studied extensively in the literature ~\cite{ilvento2019metric,emamjomeh2018adaptive,chatziafratis2018hierarchical,ghoshdastidar2019foundations,ukkonen2017crowdsourced,tamuz2011adaptively,hopkins2020noise}. Even though quadruplet queries are easier than binary optimal queries, some oracle queries maybe harder than the rest.  In a comparison query, if there is a significant gap between the two distances being compared, then such queries are easier to answer~\cite{DKMR15,braverman2008noisy}. However, when the two distances are close, the chances of an error could increase. For example,  `Is location in image $1$ closer to $3$, or $2$ is closer to $6$?' maybe difficult to answer. 

To capture noise in quadruplet comparison oracle answers, we consider two noise models. In the first noise model, when the pairwise distances are comparable, the oracle can return the pair of points that are farther instead of closer. Moreover, we assume that the oracle has access to all previous queries and can answer queries by acting adversarially. More formally, there is a parameter $\mu > 0$ such that if $\frac{\max{d(u_1, u_2), d(v_1, v_2)}}{\min{d(u_1, u_2), d(v_1, v_2)}} \leq (1+\mu)$, then adversarial error may occur, otherwise the answers are correct. We call this the "Adversarial Noise Model". 
In the second noise model called "Probabilistic Noise Model", given a pair of distances, we assume that the oracle answers correctly with a probability of $1-p$ for some fixed constant $p < \frac{1}{2}$. We consider a \textit{persistent} probabilistic noise model, where our oracle answers are \emph{persistent}, i.e., query responses remain unchanged even upon repeating the same query multiple times. Such noise models have been studied extensively~\cite{mazumdar2017clustering, prelec2017solution,ghoshdastidar2019foundations, braverman2008noisy,bressan2019correlation, galhotra2018robust} since the error due to oracles often does not change with repetition, and in some cases increases upon repeated querying~\cite{mazumdar2017clustering,prelec2017solution,galhotra2018robust}. This is in contrast to the noise models studied in ~\cite{emamjomeh2018adaptive} where response to every query is independently noisy. Persistent query models are more difficult to handle than independent query models where repeating each query is sufficient to generate the correct answer by majority voting. 
 
\subsection{Our Contributions}
\sloppy
We present algorithms for finding the maximum, nearest and farthest neighbors, $k$-center clustering and hierarchical clustering objectives under the adversarial and probabilistic noise model using comparison oracle. \revc{We show that our techniques have provable approximation guarantees for both the noise models, are efficient and  obtain good query complexity}. We empirically evaluate the robustness and efficiency of our techniques on real world datasets.

\noindent (i) \textbf{Maximum, Farthest and Nearest Neighbor}: Finding maximum has received significant attention under both adversarial and probabilistic models \cite{ajtai2009sorting,feige1994computing,geissmann2018optimal,DKMR15,braverman2008noisy,geissmann2017sorting,klein2011tolerant,geissmann2020optimal}. In this paper, we provide the following results.
\begin{description}[leftmargin=0pt]
\item[$\bullet$ Maximum under adversarial model.] We present an algorithm that returns a value within $(1+\mu)^3$ of the maximum among a set of $n$ values $V$ with probability $1-\delta$\footnote{$\delta$ is the confidence parameter and is standard in the literature of randomized algorithms.} \revc{using $O(n \log^2(1/\delta))$ oracle queries and running time (Theorem~\ref{thm:max_adv})}.
\item[$\bullet$ Maximum under probabilistic model.] \revc{We present an algorithm that requires $O(n \log^2(n/\delta))$ queries to identify $O(\log^2 (n/\delta))$th rank value with probability $1-\delta$ (Theorem~\ref{thm:max_prob})}. In other words, in $O(n \log^2(n))$ time we can identify $O(\log^2 (n))$th value in the sorted order with probability $1-\frac{1}{n^c}$ for any constant $c$. 
\end{description}
To contrast our results with the state of the art, Ajtai et al.~\cite{ajtai2009sorting} study a slightly different additive adversarial error model where the answer of a maximum query is correct if the compared values differ by $\theta$ (for some $\theta > 0$) and otherwise the oracle answers adversarially. Under this setting, they give an additive $3\theta$-approximation with $O(n)$ queries. Although, our model cannot be directly compared with theirs, we note that our model is scale invariant, and thus, provides a much stronger bound when distances are small. As a consequence, our algorithm can be used under additive adversarial model as well, and obtaining the same approximation guarantees (Theorem~\ref{thm:farthest_nearest_prob}). 

For the probabilistic model, after a long series of works \cite{braverman2008noisy,geissmann2017sorting,klein2011tolerant,geissmann2020optimal}, only recently an algorithm is proposed with query complexity $O(n\log{n})$ that returns an $O(\log{n})$th rank value with probability $1-\frac{1}{n}$ \cite{geissmann2018optimal}. Previously, the best query complexity was $O(n^{3/2})$ \cite{geissmann2020optimal}. While our bounds are slightly worse than \cite{geissmann2018optimal}, our algorithm is significantly simpler.

Rest of the work in finding maximum allow repetition of queries and assume the answers are independent~\cite{feige1994computing,DKMR15}. As discussed earlier, persistent errors are much more difficult to handle than independent errors. In \cite{feige1994computing}, the authors present an algorithm that finds maximum using $O(n\log 1/\delta)$ queries and succeeds with probability $1-\delta$. Therefore, even under persistent errors, we obtain guarantees close to the existing ones which assume independent error. The algorithms of \cite{feige1994computing,DKMR15} do not extend to our model.
\begin{description}[leftmargin=0pt]
\item[$\bullet$ Nearest Neighbor.] Nearest neighbor queries can be cast as ``finding minimum'' among a set of distances.  We can obtain bounds similar to finding maximum for the nearest neighbor queries. \revc{In the adversarial model, we obtain an $(1+\mu)^3$-approximation, and in the probabilistic model, we are guaranteed to return an element with rank $O(\log^2({n}/\delta))$ with probability $1-\delta$ using $O(n \log^2(1/\delta))$ and $O(n\log^2({n}/\delta))$ oracle queries respectively.}
\end{description}

Prior techniques have studied nearest neighbor search under noisy distance queries~\cite{mason2019learning}, where the oracle returns a noisy estimate of a distance between queried points, and repetitions are allowed. Neither the algorithm of \cite{mason2019learning}, nor other techniques developed for maximum~\cite{ajtai2009sorting,feige1994computing} and top-$k$~\cite{DKMR15} extend for nearest neighbor under our noise models.

\begin{description}[leftmargin=0pt]
\item[$\bullet$ Farthest Neighbor.] Similarly, the farthest neighbor query can be cast as finding maximum among a set of distances, and the results for computing maximum extends to this setting. However, computing the farthest neighbor is one of the basic primitives for more complex tasks like {\em $k$-center clustering}, and for that, the existing bounds under the probabilistic model that return an $O(\log{n})$th rank element is insufficient. Since distances on a metric space satisfies triangle inequality, we exploit it to get a constant approximation to the farthest query under the probabilistic model and a mild distribution assumption (Theorem~\ref{thm:farthest_nearest_prob}).
\end{description}

\noindent (ii) \textbf{$k$-center Clustering}:
$k$-center clustering is one of the fundamental models of clustering and is very well-studied \cite{williamson2011design,vazirani2013approximation}.
\begin{description}[leftmargin=0pt]
\item[$\bullet$ $k$-center under adversarial model] \revc{We design algorithm that returns a clustering that is a $2+\mu$ approximation for small values of $\mu$ with probability $1-\delta$ using $O(nk^2 + nk\log^2(k/\delta))$ queries (Theorem~\ref{thm:main_adv}).} In contrast, even when exact distances are known, $k$-center cannot be approximated better than a $2$-factor unless $P=NP$ \cite{vazirani2013approximation}. Therefore, we achieve near-optimal results.
\item[$\bullet$ $k$-center under probabilistic noise model.] \revc{For probabilistic noise, when optimal $k$-center clusters are of size at least $\Omega(\sqrt{n})$, our algorithm returns a clustering that achieves constant approximation with probability $1-\delta$ using $O(nk \log^2(n/\delta))$ queries (Theorem~\ref{thm:kcenterprob}).}
\end{description}
To the best of our knowledge, even though $k$-center clustering is an extremely popular and basic clustering paradigm, it hasn't been studied under the comparison oracle model, and we provide the first results in this domain.

\noindent (iii) \textbf{Single Linkage and Complete Linkage-- Agglomerative Hierarchical Clustering} : \revc{Under adversarial noise, we show a clustering technique that loses only a multiplicative factor of  $(1+\mu)^3$ in each merge operation and has an overall query complexity of $O(n^2)$.} Prior work~\cite{ghoshdastidar2019foundations} considers comparison oracle queries to perform  average linkage in which the unobserved pairwise similarities are generated according to a normal distribution. These techniques do not extend to our noise models.

\subsection{Other Related Work}
For finding the maximum among a given set of values, it is known that techniques based on tournament obtain optimal guarantees and are widely used~\cite{DKMR15}. For the problem of finding nearest neighbor, techniques based on locality sensitive hashing generally work well in practice~\cite{arora2018hd}. Clustering points using $k$-center objective is NP-hard and there are many well known heuristics and approximation algorithms~\cite{williamson2011design} with the classic greedy algorithm achieving an approximation ratio of $2$.  All these techniques are not applicable when pairwise distances are unknown.
 As distances between points cannot always be accurately estimated, many recent techniques leverage supervision in the form of an oracle.  Most oracle based clustering frameworks consider `optimal cluster' queries~\cite{mazumdar2017clustering,huleihel2019same,mazumdar2017query,choudhury2019top,kasper} to identify ground truth clusters. Recent techniques for distance based clustering objectives, such as $k$-means~\cite{ashtiani2016clustering,chien2018query,kim2017relaxed,kim2017semi} and $k$-median~\cite{ailon2018approximate} use optimal cluster queries in addition to distance information for obtaining better approximation guarantees. As `optimal cluster' queries can be costly or sometimes infeasible, there has been recent interest in leveraging distance based comparison oracles for other problems similar to our quadruplet oracles~\cite{emamjomeh2018adaptive, ghoshdastidar2019foundations}. 

Distance based comparison oracles have been used to study a wide range of problems and we list a few of them -- learning fairness metrics~\cite{ilvento2019metric}, top-down hierarchical clustering with a different objective~\cite{emamjomeh2018adaptive,chatziafratis2018hierarchical,ghoshdastidar2019foundations}, correlation clustering~\cite{ukkonen2017crowdsourced} and classification~\cite{tamuz2011adaptively,hopkins2020noise}, identify maximum~\cite{guo2012so, venetis2012max}, top-$k$ elements~\cite{klein2011tolerant, polychronopoulos2013human,ciceri2015crowdsourcing, DKMR15,kou2017crowdsourced,dushkin2018top},  information retrieval~\cite{kazemi2018comparison}, skyline computation~\cite{verdugo2020skyline}. To the best of our knowledge, there is no work that considers \textit{quadruplet} comparison oracle queries to perform $k$-center clustering and single/complete linkage based hierarchical clustering. 

Closely related to finding maximum, sorting has also been well studied under various comparison oracle based noise models~\cite{braverman2008noisy, braverman2016parallel}. The work of~\cite{DKMR15} considers a different probabilistic noise model with error varying as a function of difference in the values but they assume that each query is independent and therefore repetition can help boost the probability of success. Using a quadruplet oracle, \cite{ghoshdastidar2019foundations} studies the problem of recovering a hierarchical clustering under a planted noise model and is not applicable for single linkage.

\section{Preliminaries}\label{sec:prelim}
\sloppy

Let $V=\{v_1,v_2,\ldots, v_n\}$ be a collection of $n$ records such that each record maybe associated with a value $val(v_i), \forall i\in[1,n]$. We assume that there exists a total ordering over the values of elements in $V$. For simplicity we denote the value of record $v_i$ as $v_i$ instead of $val(v_i)$ whenever it is clear from the context.

Given this setting, we consider a comparison oracle that compares the values of any pair of records $(v_i,v_j)$ and outputs $\Yes$ if $v_i\leq v_j$ and $\No$ otherwise.

\begin{definition}[Comparison Oracle]
An oracle  is a function $\oracle: V\times V\rightarrow \{\Yes,\No\}$. Each oracle query considers two values as input and outputs $\oracle(v_1,v_2)=\Yes$ if $v_1\leq v_2$ and \No{} otherwise.
\end{definition}

\revc{Note that a comparison oracle is defined for any pair of values. }Given this oracle setting, we define the problem of identifying the maximum over the records $V$.

\begin{problem}[Maximum] Given a collection of $n$ records $V=\{v_1,\ldots,v_n\}$ and access to a comparison oracle $\oracle$, identify the $\arg \max_{v_i\in V} v_i$ with minimum number of queries to the oracle.\label{prob:max}
\end{problem}

As a natural extension, we can also study the problem of identifying the record corresponding to the smallest value in $V$.

\subsection{Quadruplet Oracle Comparison Query}

In applications that consider distance based comparison of records like nearest neighbor identification, the records $V=\{v_1,\ldots,v_n\}$ are generally considered to be present in  a high-dimensional metric space along with a distance  $d:V\times V\rightarrow \mathbb{R}^+$  defined over pairs of records. We assume that the embedding of records in latent space is not known, but there exists an underlying ground truth~\cite{arora2018hd}. Prior techniques mostly assume complete knowledge of accurate distance metric and are not applicable in our setting. In order to capture the setting where we can compare distances between pair of records, we define quadruplet oracle below.

\begin{definition}[ Quadruplet Oracle]
An oracle  is a function $\oracle: V\times V\times V\times V\rightarrow \{\Yes,\No\}$. Each oracle query considers two pairs of records as input and outputs $\oracle(v_1,v_2,v_3,v_4)=\Yes$ if $d(v_1,v_2)\leq d(v_3,v_4)$ and $\No{}$ otherwise.
\end{definition}

The quadruplet oracle is similar to the comparison oracle discussed before with  a difference that the two values being compared are associated with pair of records as opposed to individual records. 
Given this oracle setting, we define the problem of identifying the farthest record over $V$ with respect to a query point $q$ as follows.

\begin{problem}[Farthest point] Given a collection of $n$ records $V=\{v_1,\ldots,v_n\}$, a query record $q$ and access to a quadruplet oracle $\oracle$, identify $\arg \max_{v_i\in V\setminus \{q\}} d(q,v_i)$.\label{prob:farthest}
\end{problem}

Similarly, the nearest neighbor query returns a point that satisfies $\arg \min_{u_i\in V\setminus \{q\}} d(q,u_i)$.
 Now, we formally define the k-center clustering problem.

\begin{problem}[k-center clustering] Given a collection of $n$ records $V=\{v_1,\ldots,v_n\}$ and access to a comparison oracle $\oracle$, identify $k$ centers (say $S\subseteq V$) and a mapping of records to corresponding centers, $\pi: V\rightarrow S$, such that the maximum distance of any record from its center, i.e., $\max_{v_i \in V} d(v_i,\pi(v_i))$ is minimized.\label{prob:kcenter}
\end{problem}

We assume that the points $v_i\in V$ exist in a metric space and the distance between any pair of points is not known. We denote the \textit{unknown} distance between any pair of points $(v_i, v_j)$ where $v_i, v_j \in V$ as $d(v_i,v_j)$ and use $k$ to denote the number of clusters. Optimal clusters are denoted as $C^*$  with $C^*(v_i)\subseteq V$ denoting the set of points belonging to the optimal cluster containing $v_i$. Similarly, $C(v_i)\subseteq V$ refers to the nodes belonging to the cluster containing $v_i$ for any clustering given by $C(\cdot)$. 

In addition to the k-center clustering,  we study single linkage and complete linkage--agglomerative clustering techniques where the distance metric over the records is not known apriori. These techniques initialize each record $v_i$ in a separate singleton cluster and sequentially merge the pair of clusters having the least distance between them. In case of single linkage, the distance between two clusters $C_1$ and $C_2$ is characterized by the closest pair of records defined as:
$$d_{SL}(C_1,C_2) = \min_{v_i\in C_1,v_j\in C_2}d(v_i,v_j)$$
In complete linkage, the distance between a pair of clusters $C_1$ and $C_2$ is calculated by identifying the farthest pair of records,  
$d_{CL}(C_1,C_2) = \max_{v_i\in C_1,v_j\in C_2}d(v_i,v_j).$
\subsection{Noise Models}
The oracle models discussed in Problem~\ref{prob:max},~\ref{prob:farthest} and \ref{prob:kcenter} assume that the oracle answers every comparison query correctly. In real world applications, however, the answers can be wrong which can lead to noisy results. To formalize the notion of noise, we consider two different models. First, adversarial noise model considers a setting where a comparison query can be adversarially wrong if the two values being compared are within a multiplicative factor of $(1+\mu)$ for some constant $\mu >0$. 
$$
\oracle(v_1,v_2) =
\begin{cases}
    \Yes{}, \text{ if }  v_1< \frac{1}{(1+\mu)} v_2 \\
    \No{}, \text{ if } v_1 > (1+\mu )v_2 \\
    \text{adversarially incorrect if }\frac{1}{(1+\mu)}  \leq \frac{v_1}{v_2}\leq (1+\mu)
\end{cases}
$$

The parameter $\mu$ corresponds to the degree of error. For example, $\mu=0$ implies a perfect oracle. The model extends to quadruplet oracle as follows.
 $$
\oracle(v_1,v_2,v_3,v_4) =
\begin{cases}
    \Yes{}, \text{ if }  d(v_1,v_2)< \frac{1}{(1+\mu)} d(v_3,v_4) \\
   \No{}, \text{ if } d(v_1,v_2) > (1+\mu )d(v_3,v_4) \\
    \text{adversarially incorrect if}
    \frac{1}{(1+\mu)}  \leq \frac{d(v_1,v_2)}{d(v_3,v_4)}\leq (1+\mu)
\end{cases}
$$

The second model considers a probabilistic noise model where each comparison query is incorrect independently with a probability $p < \frac{1}{2}$ and asking the same query multiple times yields the same response.  
\revc{We discuss ways to estimate $\mu$ and $p$ from real data in Section~\ref{sec:Experiments}.}
\section{Finding \textsc{maximum}}\label{sec:finding_max}
In this section, we present robust algorithms to identify the record corresponding to the maximum value in $V$ under the adversarial noise model and the probabilistic noise model. Later we extend the algorithms to find the farthest and the nearest neighbor.
\revc{We note that our algorithms for the adversarial model are parameter free (do not depend on $\mu$) and the algorithms for the probabilistic model can use $p=0.5$ as a worst case estimate of the noise.}
\subsection{Adversarial Noise}
Consider a trivial approach that  maintains a running maximum value while sequentially processing the records, i.e., if a larger value is encountered, the current maximum value is updated to the larger value. This approach requires $n-1$ comparisons. However, in the presence of adversarial noise, our output can have a significantly lower value compared to the correct maximum. In general, if $v_{max}$ is the true maximum of $V$, then the above approach can return an approximate maximum whose value could be as low as ${v_{max}}/{(1+\mu)^{n-1}}$. To see this, assume $v_1=1$, and $v_i=(1+\mu-\epsilon)^i$ where $\epsilon > 0$ is very close to $0$. It is possible that while comparing $v_i$ and $v_{i+1}$, the oracle returns $v_i$ as the larger element. If this mistake is repeated for every $i$, then, $v_1$ will be declared as the maximum element whereas the correct answer is $v_n\approx v_1(1+\mu)^{n-1}$.

To improve upon this naive strategy, we introduce a natural \textit{keeping score} based idea where given a set $S \subseteq V$ of records, we maintain $\Count(v, S)$ that is equal to the number of values smaller than $v$ in $S$. 
    $$ \Count(v, S) = \sum_{x \in S \setminus \{ v \}} 1\{ \oracle(v, x) == \No{} \}$$

It is easy to observe that when the oracle makes no mistakes, $\Count(\smax, S) = |S|-1$ and obtains the highest score, where $\smax$ is the maximum value in $S$. Using this observation, in Algorithm~\ref{alg:maxCount}, we output the value with the highest $\Count$ score.

Given a set of records $V$, we show in Lemma~\ref{lem:count} that
$\textsc{Count-Max}(V)$ obtained using Algorithm~\ref{alg:maxCount} always returns a good approximation of the maximum value in $V$.

\begin{lemma}\label{lem:count}
Given a set of values $V$ with maximum value $\vmax$, $\textsc{Count-Max}(V)$ returns a value $\umax$ where $\umax \geq \vmax/(1+\mu)^2$ using $O(|V|^2)$ oracle queries.
\end{lemma}

\revc{
\noindent Using Example~\ref{ex:maximum}, when $\mu = 1$, we demonstrate that $(1+\mu)^2 = 4$ approximation ratio is achieved by Algorithm~\ref{alg:maxCount}.
\begin{example}
Let $S$ denote a set of four records $u,v,w$ and $t$ with ground truth values $51$, $101$, $102$ and $202$, respectively. While identifying the maximum value under adversarial noise with $\mu=1$, the oracle must return a correct answer to $\oracle(u,t)$ and all other oracle query answers can be incorrect adversarially. If the oracle answers all other queries incorrectly, we have, \Count~values of $t, w, u, v$ are $1, 1, 2$, and $2$ respectively. Therefore, $u$ and $v$ are equally likely, and when Algorithm~\ref{alg:maxCount} returns $u$, we have a $202/51\approx 3.96$ approximation.\label{ex:maximum}
\end{example}
}

\noindent From Lemma~\ref{lem:count}, we have that $O(n^2)$ oracle queries where $|V| = n$, are required to get $(1+\mu)^2$ approximation. In order to improve the query complexity, we use a \textit{tournament} to obtain the maximum value. The idea of using a \emph{tournament} for finding maximum has been studied in the past~\cite{DKMR15,feige1994computing}.

Algorithm~\ref{alg:maxTournament} presents pseudo code of the approach that takes values $V$ as input and outputs an approximate maximum value. It constructs a balanced $\lambda$-ary tree $\mathcal{T}$ containing $n$ leaf nodes such that a random permutation of the values $V$ is assigned to the leaves of $\mathcal T$. In a \textit{tournament}, the internal nodes of $\mathcal{T}$ are processed bottom-up such that at every internal node $w$, we assign the value that is largest among the children of $w$. To identify the largest value, we calculate $\arg \max_{v\in \children(w)}\Count(v,\children(w))$ at the internal node $w$, where $\Count(v, X)$ refers to the number of elements in $X$ that are considered smaller than $v$. Finally, we return the value at the root of $\mathcal T$ as our output. 
In Lemma~\ref{lem:tournament}, we show that Algorithm~\ref{alg:maxTournament} returns a value that is a $(1+\mu)^{2\log_\lambda n}$ multiplicative approximation of the maximum value.

\begin{algorithm}
\begin{algorithmic}[1]
\State \textbf{Input} : A set of values $S$
\State \textbf{Output} : An approximate maximum value of $S$
\For {$v \in S$}
\State Calculate $\Count(v, S)$ 
\EndFor
\State $\umax \leftarrow \text{arg max}_{v \in S} \Count(v, S)$
\State \Return $\umax$
\end{algorithmic}
\caption{\textsc{Count-Max}(S) : finds the max. value by counting}\label{alg:maxCount}
\end{algorithm}
\setlength{\textfloatsep}{3pt}

\begin{lemma}
Suppose $v_{max}$ is the maximum value among the set of records $V$. Algorithm~\ref{alg:maxTournament} outputs a value $u_{max}$ such that $u_{max}\geq \frac{v_{max}}{(1+\mu)^{2 \log_\lambda n}}$ using  $O(n\lambda)$ oracle queries.
\label{lem:tournament}
\end{lemma}

According to Lemma~\ref{lem:tournament}, Algorithm~\ref{alg:maxTournament}  identifies a constant approximation when $\lambda=\Theta(n)$, $\mu$ is a fixed constant and has a query complexity of $\Theta(n^2)$. By reducing the degree of the tournament tree from $\lambda$ to $2$, we can achieve $\Theta(n)$ query complexity, but with a worse approximation ratio of $(1+\mu)^{\log n}$. 

Now, we describe our main algorithm (Algorithm~\ref{alg:max_adv}) that uses the the following observation to improve the overall query complexity.

\begin{observation}
At an internal node $w\in \mathcal T$, the identified maximum is incorrect only if there exists $x\in \children(w)$ that is very close to the true maximum (say $w_{max}$), i.e. $\frac{w_{\max}}{(1+\mu)}\leq x\leq (1+\mu)w_{\max}$.
\end{observation}
Based on the above observation, our Algorithm~\textsc{Max-Adv} uses two steps to identify a good approximation of $\vmax$. Consider the case when there are a lot of values that are close to $\vmax$. In Algorithm~\textsc{Max-Adv}, we use a subset $\widetilde{V} \subseteq V$ of size $\sqrt{n} t$ (for a suitable choice of parameter $t$) obtained using uniform sampling with replacement. We show that using a sufficiently large subset $\widetilde{V}$, obtained by sampling, we ensure that at least one value that is closer to $\vmax$ is in $\widetilde{V}$, thereby giving a good approximation of $\vmax$.

\begin{algorithm}
\begin{algorithmic}[1]
\State \textbf{Input} : Set of values $V$, Degree $\lambda$
\State \textbf{Output} : An approximate maximum value $\umax$
\State Construct a balanced $\lambda$-ary tree $\mathcal T$ with $|V|$ nodes as leaves.
\State Let $\pi_V$ be a random permutation of $V$ assigned to leaves of $\mathcal T$ 
\For{$i=1$ to $\log_\lambda |V|$}
\For{internal node $w$ at level $\log_\lambda |V|-i$}
\State Let  $U$ denote the children of $w$.
\State Set the internal node $w$ to $\textsc{Count-Max}(U)$
\EndFor
\EndFor
\State $\umax \leftarrow $value at root of $\mathcal T$
\State \Return $\umax$
\end{algorithmic}
\caption{\textsc{Tournament} : finds the maximum value using a tournament tree}\label{alg:maxTournament}
\end{algorithm}
\setlength{\textfloatsep}{3pt}
In order to handle the case when there are only a few values closer to $\vmax$, we divide the entire data set into $l$ disjoint parts (for a suitable choice of $l$) and run the \textsc{Tournament} algorithm with degree $\lambda = 2$ on each of these parts separately (Algorithm~\ref{alg:max_noise}). As there are very few points close to $\vmax$, the probability of comparing any such value with $\vmax$ is small, and this ensures that in the partition containing $\vmax$, $\textsc{Tournament}$ returns $\vmax$. We collect the maximum values returned by Algorithm~\ref{alg:maxTournament} from all the partitions and include these values in $T$ in Algorithm~\textsc{Max-Adv}. We repeat this procedure $t$ times and set $ l = \sqrt{n}, t = 2\log(2/\delta)$ to achieve the desired success probability $1-\delta$. We combine the outputs of both the steps, i.e., $\widetilde{V}$ and $T$ and output the maximum among them using $\textsc{Count-Max}$. This ensures that we get a good approximation as we use the best of both the approaches.

\setlength{\textfloatsep}{3pt}
\begin{algorithm}[!ht]
\begin{algorithmic}[1]
\State \textbf{Input} : Set of values $V$, number of partitions $l$
\State \textbf{Output} : A set of maximum values from each partition
\State Randomly partition $V$ into $l$ equal parts $V_1, V_2, \cdots V_{l}$ 
\For{$i = 1$ to $l$}
\State $p_i\leftarrow \textsc{Tournament}(V_i,2)$
\State $T\leftarrow T\cup\{p_i\}$
\EndFor
\State \Return $T$
\end{algorithmic}
\caption{\textsc{Tournament-Partition} \label{alg:max_noise}}
\end{algorithm}

\begin{algorithm}[h]
\begin{algorithmic}[1]
\State \textbf{Input} : Set of values $V$, number of iterations $t$, partitions $l$
\State \textbf{Output} : An approximate maximum value $\umax$
\State $i\leftarrow 1,  T \leftarrow \phi$ 
\State Let $\widetilde{V}$ denote a sample of size $\sqrt{n}t$ selected uniformly at random (with replacement) from $V$.
\For{$i\leq t$}
\State $T_i \leftarrow \textsc{Tournament-Partition}(V, l)$
\State $T\leftarrow T \cup T_i$
\EndFor
\State  $\umax  \leftarrow \textsc{Count-Max}(\widetilde{V} \cup T)$
\State \Return $\umax$
\end{algorithmic}
\caption{\textsc{Max-Adv} : {Maximum with Adversarial Noise} \label{alg:max_adv}}
\end{algorithm}
\setlength{\textfloatsep}{3pt}

\noindent \textbf{Theoretical Guarantees}. In order to prove approximation guarantee of Algorithm~\ref{alg:max_adv}, we first argue that the sample $\widetilde{V}$ contains a good approximation of the maximum value $\vmax$ with a high probability. Let $C$ denote the set of values that are very close to $\vmax$.
$\text{Suppose }C=\{u : {\vmax}/{(1+\mu)}\leq u\leq \vmax \}. $
In Lemma~\ref{lem:hierquality}, we first show that $\widetilde{V}$ contains a value $v_j \in \widetilde{V}$ such that $v_j \geq {\vmax}/{(1+\mu)}$,  whenever the size of $C$ is large, i.e.,  $|C| > {\sqrt{n}}/{2}$. Otherwise, we show that we can recover $\vmax$ correctly with probability $1-\delta/2$ whenever $|C| \leq {\sqrt{n}}/{2}$.

\begin{lemma}

\begin{enumerate}
 \item If $|C| >  {\sqrt{n}}/{2}$, then there exists a value $ v_j \in \widetilde{V}$ satisfying $v_j \geq {\vmax}/{(1+\mu)}$ with  probability of $1-\delta/2$.
 \item Suppose $|C| \leq  {\sqrt{n}}/{2}$. Then, $T$ contains  $\vmax$ with probability at least $1-\delta/2$.
\end{enumerate}

\label{lem:hierquality}
\end{lemma}

Now, we briefly provide a sketch of the proof of Lemma~\ref{lem:hierquality}. Consider the first step, where we use a uniformly random sample $\widetilde{V}$ of $\sqrt{n}t$ points from $V$ (obtained with replacement).  When $|C| \geq {\sqrt{n}}/{2}$, probability that $\widetilde{V}$ contains a value from $C$ is given by $1-(1-|C|/n)^{|\widetilde{V}|} 
=1-\left(1-\frac{1}{2\sqrt{n}}\right)^{2\sqrt{n}\log(2/\delta)} \approx 1-\delta/2$.

In the second step, Algorithm~\ref{alg:max_adv} uses a modified tournament tree that partitions the set $V$ into $l = \sqrt{n}$ parts of size ${n}/{l} = \sqrt{n}$ each and identifies a maximum $p_i$ from each partition $V_i$ using Algorithm~\ref{alg:maxTournament}. We have that the expected number of elements from $C$ in a partition $V_i$ containing $\vmax$ is ${|C|}/{l} = {\sqrt{n}}/{(2\sqrt{n})} = {1}/{2}$. Thus by the Markov's inequality, the probability that $V_i$ contains a value from $C$ is $\leq {1}/{2}$. With $1/2$ probability, $\vmax$ will never be compared with any point from $C$ in the partition $V_i$. To increase the success probability, we run this procedure $t$ times and obtain all the outputs. Among the $t$ runs of Algorithm~\ref{alg:maxTournament}, we argue that $\vmax$ is never compared with any value of $C$ in at least one of the iterations with a probability at least $1-(1-1/2)^{2\log(2/\delta) } \geq 1-{\delta}/{2}$.

 {In Lemma~\ref{lem:count}}, we show that using $\textsc{Count-Max}$ we get a $(1+\mu)^2$ multiplicative approximation. Combining it with Lemma~\ref{lem:hierquality}, we have that $\umax$ returned by Algorithm~\ref{alg:max_adv} satisfies $\umax \geq \frac{\vmax}{(1+\mu)^3}$ with probability $1-\delta$. For query complexity, Algorithm~\ref{alg:max_noise} identifies $\sqrt{n}t$  samples denoted by $\widetilde{V}$. These identified values, along with $T$ are then processed by $\textsc{Count-Max}$ to identify the maximum $u_{max}$. This step requires $O(|\widetilde{V} \cup T |^2 ) = O(n\log^2(1/\delta))$ oracle queries. 

\begin{theorem}\label{thm:max_adv}
Given a set of values $V$, Algorithm~\ref{alg:max_adv} returns a $(1+\mu)^3$ approximation of maximum value with probability $1-\delta$ using $O(n \log^2(1/\delta))$ oracle queries.
\end{theorem}
\subsection{Probabilistic Noise}\label{sec:max_prob}
We cannot directly extend the algorithms for the adversarial noise model to probabilistic noise. Specifically, the theoretical guarantees of Lemma~\ref{lem:tournament} do not apply when the noise is probabilistic. In this section, we develop several new ideas to handle probabilistic noise.

Let $\textrm{rank}(u, V)$ denote the index of $u$ in the non-increasing sorted order of values in $V$. So, $v_{max}$ will have rank $1$ and so on. Our main idea is to use an early stopping approach that uses a sample $S \subseteq V$ of $O(\log(n/\delta))$ values selected randomly and for every value $u$ that is not in $S$, we calculate $\Count(u, S)$ and discard $u$ using a chosen threshold for the $\Count$ scores. We argue that by doing so, it helps us eliminate the values that are far away from the maximum in the sorted ranking. This process is continued $\Theta(\log n)$ times to identify the maximum value. We present the pseudo code in the Appendix and prove the following approximation guarantee.

\begin{theorem}\label{thm:max_prob}
There is an algorithm that returns $\umax \in V$ such that $\textrm{rank} (\umax, V) = O(\log^2(n/\delta))$ with probability $1-\delta$ and requires $O(n \log^2(n/\delta))$ oracle queries.
\end{theorem}

The algorithm to identify the minimum value is same as that of maximum with a modification where $\Count$ scores consider the case of $\Yes$ (instead of $\No$): $\Count(v,S) = \sum_{x\in S\setminus \{v\}}\mathbf{1} \{\oracle(v,x)==\Yes{}\} $
\subsection{Extension to Farthest and Nearest Neighbor}\label{sec:farthest}

Given a set of records $V$, the farthest record from a query $u$ corresponds to the record $u'\in V$ such that $d(u,u')$ is maximum. This query is equivalent to finding maximum in the set of distance values given by $D(u) =\{ d(u,u') \mid \forall u'\in V \}$ containing $n$ values for which we already developed algorithms in Section~\ref{sec:finding_max}. Since the ground truth distance between any pair of records is not known, we require quadruplet oracle (instead of comparison oracle) to identify the maximum element in $D(u)$.  Similarly, the nearest neighbor of query record $u$ corresponds to finding the record with minimum distance value in $D(u)$.  Algorithms for finding maximum from previous sections, extend for these settings with similar guarantees. 

\revc{ \begin{example}
Figure~\ref{fig:1dexample} shows a worst-case example for the approximation guarantee to identify the farthest point from $s$ (with $\mu=1$). Similar to Example~
\ref{ex:maximum}, we have, \Count~values of $t, w, u, v$ are $1, 1, 2, 2$ respectively. Therefore, $u$ and $v$ are equally likely, and when Algorithm~\ref{alg:maxCount} outputs $u$, we have a $\approx 3.96$ approximation.\label{ex:farthest}
\end{example}
}
\begin{figure}
\centering
\begin{tikzpicture}
\fill[color=red]  (-2,0) circle (2pt);

\fill  (-0.25,0) circle (2pt);
\fill  (1.5,0) circle (2pt);

\fill  (1.8,0) circle (2pt);
\fill  (5.2,0) circle (2pt);

\draw (-2,0) -- (0,0);
\draw (0,0) -- (-2,0);

\draw (2,0) -- (0,0);
\draw (0,0) -- (2,0);

\draw (2,0) -- (3,0);
\draw (3,0) -- (2,0);

\draw (5,0) -- (3,0);
\draw (3,0) -- (5.2,0);

\draw (-2,0) node[anchor=north east] {$s$};

\draw (-1.25,0) node[anchor=south] {$51$};

\draw (-0.25,0) node[anchor=north] {$u$};

\draw (0.25,0) node[anchor=south] {$50$};

\draw (1.5,0) node[anchor=north east] {$v$};

\draw (1.7,0) node[anchor=south] {$1$};

\draw (1.8,0) node[anchor=north] {$w$};

\draw (4,0) node[anchor=south] {$100$};

\draw (5,0) node[anchor=north west] {$t$};

\end{tikzpicture}
\caption{\small \revc{Example for Lemma~\ref{lem:count} with $\mu=1$. }}
\label{fig:1dexample}
\end{figure}
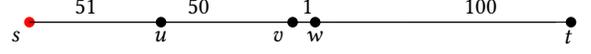

For probabilistic noise, the farthest identified in Section~\ref{sec:max_prob} is guaranteed to rank within the top-$O(\log^2 n)$ values of set $V$ (Theorem~\ref{thm:max_prob}). In this section, we show that it is possible to compute the farthest point within a small additive error under the probabilistic model, if the data set satisfies an additional property {discussed below}. \revc{For the simplicity of exposition,  we assume $p \leq 0.40$, though our algorithms work for any value of $p < 0.5$ (with different constants).} 

One of the challenges in  developing robust algorithms for farthest identification is that every relative distance comparison of records from $u$ ($\oracle(u,v_i,u,v_j)$ for some $v_i,v_j\in V$) may be answered incorrectly with constant error probability $p$ and the success probability cannot be boosted by repetition. We overcome this challenge by performing pairwise comparisons in a robust manner. Suppose the desired failure probability is $\delta$, we observe that if $\Theta(\log(1/\delta))$ records closest to the query $u$ are known (say $S$) and $\max_{x\in S}\{d(u,x)\}\leq \alpha$ for some $\alpha > 0$, then each pairwise comparison of the form $\oracle(u,v_i,u,v_j)$  can be replaced by Algorithm~\textsc{PairwiseComp} and use it to execute Algorithm~\ref{alg:max_adv}. Algorithm~\ref{alg:pairwise} takes the two records $v_i$ and $v_j$ as input along with $S$ and outputs $\Yes$ or $\No$ where $\Yes$ denotes that $v_i$ is closer to $u$. We calculate $\FCount(v_i,v_j) = \sum_{x\in S} \mathbf{1}\{ \oracle(v_i,x,v_j,x) == \Yes \}$ as a robust estimate where the oracle considers $v_i$ to be closer to $x$ than $v_j$. If $\FCount(v_i,v_j)$ is smaller than $0.3 |S| \leq (1-p)|S|/2$  then we output $\No$ and $\Yes$ otherwise. Therefore, every pairwise comparison query is replaced with $\Theta(\log(1/\delta))$ quadruplet queries using Algorithm~\ref{alg:pairwise}.

We argue that Algorithm~\ref{alg:pairwise} will output the correct answer with a high probability if $|d(u,v_j)-d(u,v_i)|\geq 2\alpha$ \revc{(See Fig~\ref{fig:probexample})}.  In Lemma~\ref{lem:pairwise}, we show that, if $d(u,v_j)> d(u,v_i)+2\alpha$, then,  $\FCount(v_i,v_j) \geq 0.3 |S|$  with probability $1-\delta$.
\begin{lemma}
Suppose $\max_{v_i \in S} d(u, v_i) \leq \alpha$ and $|S| \geq 6 \log(1/\delta)$. Consider two records $v_i$ and $v_j$ such that $d(u,v_i) < d(u,v_j)-2\alpha$ then  $\FCount(v_i,v_j)\geq 0.3 |S|$  with a probability of $1-\delta$.\label{lem:pairwise}
\end{lemma}
With the help of Algorithm~\ref{alg:pairwise}, relative distance query of any pair of records $v_i,v_j$ from $u$ can be answered correctly with a high probability provided $|d(u,v_i)-d(u,v_j)|\geq 2\alpha$. Therefore, the output of  Algorithm~\ref{alg:pairwise} is equivalent to an additive adversarial error model where any quadruplet query can be adversarially incorrect if the distance $|d(u,v_i)-d(u,v_j)|< 2\alpha$ and correct otherwise. In the Appendix, we show that Algorithm~\ref{alg:max_adv} can be extended to the additive adversarial error model, such that each comparison $\O(u,v_i,u,v_j)$ is replaced by \textsc{PairwiseComp} (Algorithm~\ref{alg:pairwise}).
 We give an approximation guarantee, that loses an additive $6 \alpha$  following a similar analysis of Theorem~\ref{thm:max_adv}. 

\begin{figure}
\centering
\begin{tikzpicture}[thick, scale=0.4, node/.style={circle,draw=white,scale=0.9, fill=white,font=\sffamily\Large\bfseries}]
 \filldraw[color=black, fill=blue!20,  line width=2pt](0,0) circle (3.5);
 \filldraw[color=black, fill=red!0,  line width=2pt](0,0) circle (2);
 \filldraw[color=gray, fill=red!10,  line width=2pt](0,0) circle (0.75);
 \coordinate [label=below:$S$] (A) at (0,-0.8);
 \coordinate [label=left:$v_j$] (A) at (-3,2);
\coordinate [label=left:$v_i$] (A) at (0,1.20);
\coordinate [label=below:$u$] (A) at (0,0);
  \draw[line width=1pt,black,->] (0,0)--(0.75, 0)node[anchor=north ] at (0.35,0.72){\textbf{$\alpha$}};
  \draw[line width=1pt,gray,->] (0,0)--(0,1.7);
  \draw[line width=1pt,black,<->] (2,0)--(3.5,0)node[anchor=north ] at (2.75,0.2){\textbf{$2\alpha$}};
      \draw[gray,->] (0,0)--(-2.9,2)node[anchor=north ] at (-3.39,2.58){};
      \node[draw,text width=4.7cm] at (11,0) {In this example, $\oracle(u,v_i,u,v_j)$ is answered correctly with a probability $1-p$. To boost the correctness probability, $\FCount$ uses the queries $\oracle(x,v_i,x,v_j)$,  $\forall x$ in the red region around $u$, denoted by $S$.};
\end{tikzpicture}
\caption{\revc{Algorithm~\ref{alg:pairwise} returns `\Yes'~as $d(u, v_i) < d(u, v_j) - 2\alpha$.}\label{fig:probexample}}
\end{figure}
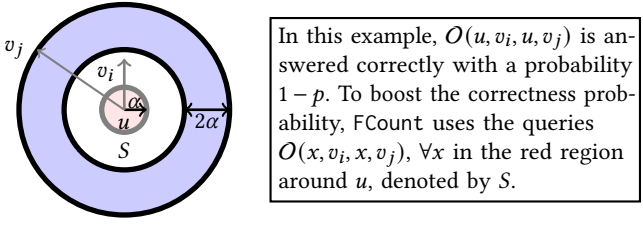

\begin{algorithm}
\begin{small}
\caption{\small \textsc{PairwiseComp} ($u, v_i, v_j, S$)\label{alg:pairwise}}
\label{alg:maxoptimized}
\begin{algorithmic}[1]
\State Calculate $\FCount(v_i, v_j) = \sum_{x\in S} \mathbf{1}\{\oracle(x,v_i,x,v_j) == \Yes \} $
\If{$\FCount(v_i, v_j) < 0.3 |S|$}
\State \Return \No
\Else ~\Return \Yes
\EndIf
\end{algorithmic}
\end{small}
\end{algorithm}

\begin{theorem}\label{thm:farthest_nearest_prob}
Given a query vertex $u$ and a set $S$ with $|S| = \Omega(\log(n/\delta))$ such that $\max_{v\in S}d(u,v)\leq \alpha$ then the farthest identified using Algorithm~\ref{alg:max_adv} (with \textsc{PairwiseComp}), denoted by $\umax$ is within $6 \alpha$ distance from the optimal farthest point, i.e., $d(u, \umax) \geq \max_{v \in V} d(u, v) - 6\alpha$ with a probability of $1-\delta$. Further the query complexity is $O(n\log^3(n/\delta))$.
\end{theorem}

\section{$k$-center clustering}\label{sec:kcenter}
In this section, we present algorithms for $k$-center clustering and prove constant approximation guarantees of our algorithm.
Our algorithm is an adaptation of the classical greedy algorithm for $k$-center~\cite{gonzalez1985clustering}. 
The greedy algorithm~\cite{gonzalez1985clustering} is initialized with an arbitrary point as the first cluster center and then iteratively identifies the next centers. In each iteration, it assigns all the points to the current set of clusters, by identifying the closest center for each point. Then, it finds the farthest point among the clusters and uses it as the new center. This technique requires $O(n k)$ distance  comparisons in the absence of noise and guarantees $2$-approximation of the optimal clustering objective. We provide the pseudocode for this approach in Algorithm~\ref{alg:greedykcenter_adv}. Using an argument similar to the one presented for the worst case example in Section~\ref{sec:finding_max}, we can show that if we use Algorithm~\ref{alg:greedykcenter_adv} where we replace every comparison with an oracle query, the generated clusters can be arbitrarily worse even for small error. In order to improve its robustness, we devise new algorithms to perform assignment of points to respective clusters and  farthest point identification. Missing Details from this section are discussed in Appendix~\ref{app:k-center} and ~\ref{app:k-center-prob}.

\begin{algorithm}
\begin{algorithmic}[1]
\State \textbf{Input} : Set of points $V$
\State \textbf{Output} : Clusters $\mathcal{C}$
\State $s_1 \leftarrow \text{arbitrary point from }V$, $S=\{s_1\}$, $C=\{\{V\}\}$.
\For{$i = 2$ to $k$}
\State $s_i\leftarrow \ApproxFarthest(S,C)$
\State $S\leftarrow S\cup\{s_i\}$
\State $\mathcal C \leftarrow  \Assign(S)$
\EndFor
\State \Return ${\mathcal{C}}$
\end{algorithmic}
\caption{{Greedy Algorithm} \label{alg:greedykcenter_adv}}
\end{algorithm}

\subsection{Adversarial Noise}
Now, we describe the two steps (\textsc{Approx-Farthest} and \textsc{Assign}) of the Greedy Algorithm that will complete the description of Algorithm~\ref{alg:greedykcenter_adv}. To do so, we build upon the results from previous section that give algorithms for obtaining maximum/farthest point.  

\noindent \textbf{\textsc{Approx-Farthest}}. 
Given a clustering $\mathcal C$, and a set of centers $S$, we construct the pairs $(v_i, s_j)$ where $v_i$ is assigned to cluster $C(s_j)$ centered at $s_j \in S$.  Using Algorithm~\ref{alg:max_adv}, we identify the point, center pair that have the maximum distance i.e. $\arg \max_{v_i\in V} d(v_i,s_j)$, which corresponds to the farthest point. For the parameters, we use $l = \sqrt{n}$, $t = \log(2k/\delta)$ and number of samples $\widetilde{V} = \sqrt{n}t$. 

\noindent  \textbf{\textsc{Assign}}. After identifying the farthest point, we reassign all the points to the centers (now including the farthest point as the new center) closest to them. We calculate a movement score called $\MCount$ for every point with respect to each center. 
$\MCount(u, s_j) = \sum_{s_k \in S\setminus \{s_j\}} 1 \{ \oracle((s_j, u), (s_k, u)) == \Yes \}$, for any record $u\in V$ and $s_j\in S$.
This step is similar to \textsc{Count-Max} Algorithm. We assign the point $u$ to the center with the highest \MCount~value.

\revc{\begin{example}
Suppose we run k-center algorithm with $k=2$ and $\mu=1$ on the points in Example~\ref{ex:farthest}. The optimal centers are $u$ and $t$ with radius $51$. On running our algorithm, suppose  $w$ is chosen as the first center and \textsc{Approx-Farthest} calculates \Count~values similar to Example~\ref{ex:maximum}.
We have, \Count~values of $s, t, u, v$ are $1, 2, 3, 0$ respectively. Therefore, our algorithm identifies $u$ as the second center, achieving $3$-approximation.
\end{example}}

\noindent \textbf{Theoretical Guarantees}.
We now prove the approximation guarantee obtained by Algorithm ~\ref{alg:greedykcenter_adv}.

In each iteration, we show that \textsc{Assign} reassigns each point to a center with distance approximately similar to the distance from the closest center. This is surprising given that we only use $\MCount$ scores for assignment. Similarly, we show that \textsc{Approx-Farthest} (Algorithm~\ref{alg:max_adv}) identifies a close approximation to the true farthest point. Concretely, we show that every point is assigned to a center which is a $(1+\mu)^2$ approximation; Algorithm~\ref{alg:max_adv} identifies farthest point $w$ which is a $(1+\mu)^{5}$ approximation. 

In every iteration of the Greedy algorithm, if we identify an $\alpha$-approximation of the farthest point, and a $\beta$-approximation when reassigning the points, then, we show  that the clusters output are a $2\alpha \beta^2$-approximation to the $k$-center objective. For complete details, please refer Appendix~\ref{app:k-center}. Combining all the claims, for a given error parameter $\mu$, we obtain:

{
\begin{theorem}
For $\mu < \frac{1}{18}$, 
Algorithm~\ref{alg:greedykcenter_adv} achieves a $(2+O(\mu))$-approximation for the $k$-center objective using $O(nk^2 + n k\cdot \log^2(k/\delta))$ oracle queries with probability $1-\delta$.
\label{thm:main_adv}
\end{theorem}}

\subsection{Probabilistic Noise}
For probabilistic noise, each query can be incorrect with  probability $p$ and therefore, Algorithm~\ref{alg:greedykcenter_adv} 
may lead to poor approximation guarantees. Here, we build upon the results from section~\ref{sec:farthest} and provide \ApproxFarthest~and \Assign~algorithms. We denote the size of minimum cluster among optimum clusters $C^*$ to be $m$, and total failure probability of our algorithms to be $\delta$. We assume $p \leq 0.40$, a constant strictly less than $\frac{1}{2} $.  Let $\gamma = 450$ be a large constant used in our algorithms which obtains the claimed guarantees.

\noindent \textbf{Overview}. Algorithm~\ref{alg:greedykcenter_prob} presents the pseudo-code of our algorithm that operates in two phases. In the first phase (lines 3-12), we sample each point with a probability ${\gamma \log(n/\delta)}/{m}$ to identify a small sample of  {$\approx \frac{\gamma n\log(n/\delta)}{m}$} points (denoted by $\widetilde{V}$) and use Algorithm~\ref{alg:greedykcenter_prob}  to identify $k$ centers iteratively. In this process, we also identify a \textit{core} for each cluster (denoted by $R$). Formally, \textit{core} is defined as a set of $\Theta({\log(n/\delta)})$ points that are very close to the center with high probability. The cores are then used in the second phase (line 15) for the assignment of remaining points. 
\begin{algorithm}
\begin{algorithmic}[1]
\State \textbf{Input} : Set of points $V$, smallest cluster size $m$.
\State \textbf{Output} : Clusters ${C}$
\State For every$ \ u \in V$, include $u$ in $\widetilde{V}$ with probability $\frac{\gamma \log(n/\delta)}{m}$
\State $s_1\leftarrow \text{select an arbitrary point from }\widetilde{V}$, $S\leftarrow \{s_1\}$
\State ${C(s_1)} \leftarrow  \widetilde{V}$
\State $R(s_1)\leftarrow \identifycore(C(s_1),s_1)$
\For{$i = 2$ to $k$}
\State $s_i\leftarrow \ApproxFarthest(S,C)$
\State ${C,R} \leftarrow  \Assign(S, s_i,R)$
\State $S\leftarrow S\cup\{s_i\}$
\EndFor
\State $C\leftarrow \textsc{Assign-Final}(S,R,V\setminus \widetilde{V})$\\
\Return ${{C}}$
\end{algorithmic}
\caption{{Greedy Clustering} \label{alg:greedykcenter_prob}}
\end{algorithm}
Now, we describe the main challenge in extending $\ApproxFarthest$ and $\Assign$ ideas of Algorithm~\ref{alg:greedykcenter_adv}. Given a cluster $C$ containing the center $s_i$, when we find the $\ApproxFarthest$, the ideas from Section~\ref{sec:max_prob} give a $O(\log^2 n)$ rank approximation. As shown in section~\ref{sec:farthest}, we can improve the approximation guarantee by considering a set of $\Theta(\log(n/\delta))$ points closest to $s_i$, denoted by $R(s_i)$ and call them \textit{core} of $s_i$. We argue that such an assumption of set $R$ is justified. For example, consider the case when clusters are of size $\Theta(n)$ and sampling $k\log(n/\delta)$ points gives us $\log(n/\delta)$ points from each optimum cluster; which means that there are $\log(n/\delta)$ points within a distance of $2\OPT$ from every sampled point where $\OPT$ refers to the optimum $k$-center objective.

\noindent \textbf{\Assign}. Consider a point $s_i$ such that we have to assign points to form the cluster $C(s_i)$ centered at $s_i$. We calculate an \textit{assignment} score (called \ACount~in line 4) for every point $u$ of a cluster $C(s_j) \setminus R(s_j)$ centered at $s_j$. \ACount~captures the total number of times $u$ is considered to belong to same cluster as that of $x$ for each $x$ in the core $R(s_j)$. Intuitively, points that belong to same cluster as that of $s_i$ are expected to have higher \ACount~score. Based on the scores, we move $u$ to $C(s_i)$ or keep it in $C(s_j)$.
\begin{algorithm}
\caption{$\Assign( S, s_i,  R)$}
\label{alg:assign_prob}
\begin{small}
\begin{algorithmic}[1]
\State $C(s_i)\leftarrow \{ s_i \}$ 
\For{$s_j \in S$}
\For{$u \in C(s_j) \setminus R(s_j)$}
\State $\ACount (u, s_i, s_j) = \sum_{v_k\in R(s_j)} \mathbf{1}\{\oracle(u,s_i,u,v_k) == \Yes \}$
\If {$\ACount(u, s_i, s_j) > 0.3 |R(s_j)|$}    
\State $C(s_i) \leftarrow C(s_i) \cup \{ u \}; C(s_j) \leftarrow C(s_j) \setminus \{u \}$ 
\EndIf
\EndFor
\EndFor
\State $R(s_i)\leftarrow \identifycore(C(s_i),s_i)$
\State \Return C, R
\end{algorithmic}
\end{small}
\end{algorithm}
\begin{algorithm}
\caption{$\identifycore(C(s_i),s_i)$}
\label{alg:core_prob}
\begin{algorithmic}[1]
\For{$u \in C(s_i)$}
\State \small{\Count(u)=}\small{$\sum_{x \in C(s_i)} 1\{ \oracle(s_i, x, s_i,u) == \No \}$}
\EndFor
\State $R(s_i)$ denote set of ${8\gamma \log(n/\delta)}/9$ points with the highest $\Count{}$ values.
\State \Return $R(s_i)$ 
\end{algorithmic}
\end{algorithm}

\noindent \textbf{\textsc{\identifycore}}.  After forming cluster $C(s_i)$, we identify the \emph{core} of $s_i$. For this, we calculate a score, denoted by \Count~and captures number of times it is closer to $s_i$ compared to other points in $C(S_i)$. Intuitively, we expect  points with high values of \Count~to belong to $C^*(s_i)$ i.e., optimum cluster containing $s_i$. Therefore we sort these \Count~scores and return the highest scored points.

\noindent \textbf{\ApproxFarthest}. 
For a set of clusters $\mathcal C$, and a set of centers $S$, we construct the pairs $(v_i, s_j)$ where $v_i$ is assigned to cluster $C(s_j)$ centered at $s_j \in S$ and each center $s_j\in S$ has a corresponding core $R(s_j)$.  The farthest point can be found by finding the maximum distance (point, center) pair among all the points considered. To do so, we use the ideas developed in section~\ref{sec:farthest}.

 We leverage \textsc{ClusterComp} (Algorithm~\ref{alg:cluster_comp}) to compare the distance of two points, say $v_i,v_j$ from their respective centers $s_i,s_j$. \textsc{ClusterComp} gives a robust answer to a pairwise comparison query to the oracle $\oracle(v_i, s_i, v_j, s_j)$ using the cores $R(s_i)$ and $R(s_j)$. \textsc{ClusterComp} can be used as a pairwise comparison subroutine in place of \textsc{PairwiseComp} for the algorithm in Section~\ref{sec:finding_max} to calculate the farthest point. For every $s_i \in S$, let $\widetilde{R}(s_i)$ denote an arbitrary set of $\sqrt{R(s_i)}$ points from $R(s_i)$. For a \textsc{ClusterComp} comparison query between the pairs $(v_i, s_i)$ and $(v_j, s_j)$, we use these subsets in Algorithm~\ref{alg:cluster_comp} to ensure that we only make $\Theta(\log(n/\delta))$ oracle queries for every comparison. However, when the query is between points of the same cluster, say $C(s_i)$, we use all the $\Theta(\log(n/\delta))$ points from $R(s_i)$. For the parameters used to find maximum using Algorithm~\ref{alg:max_adv}, we use $l = \sqrt{n}$, $t = \log(n/\delta)$.
 \revc{
 \begin{example}
 Suppose we run $k$-center Algorithm~\ref{alg:greedykcenter_prob} with $k=2$ and $m = 2$ on the points in Example~\ref{ex:farthest}. Let $w$ denote the first center chosen and Algorithm~\ref{alg:greedykcenter_prob} identifies the core $R(w)$ by calculating \Count~values. If $\oracle(u,w,s,w)$ and $\oracle(s,w,t,w)$ are answered incorrectly (with probability $p$), we obtain \Count~values of $v, s, u, t$ as $3, 2, 1, 0$ respectively; and $v$ is added to $R(w)$. We identify the second center $u$ by calculating $\FCount$ for $s,u$ and $t$ (See Fig.~\ref{fig:probexample}). After assigning (using \Assign), the clusters identified are $\{w, v\}, \{u, s, t\}$, achieving $3$-approximation.
 \end{example}
 }

\begin{algorithm}
\begin{small}
\caption{\small \textsc{ClusterComp} ($v_i, s_i, v_j, s_j$)}
\label{alg:cluster_comp}
\begin{algorithmic}[1]
\State comparisons $\leftarrow 0, \FCount(v_i, v_j) \leftarrow 0$
\If{$s_i = s_j$}
\State Let $\FCount(v_i, v_j) = \sum_{x \in R(s_i)} \mathbf{1} \{ \oracle(v_i, x, v_j, x) == \Yes \}$
\State comparisons $\leftarrow |R(s_i)|$
\Else~ Let $\FCount(v_i, v_j) = \sum_{x \in \widetilde{R}(s_i), y\in \widetilde{R}(s_j)} \mathbf{1} \{\oracle(v_i,x,v_j,y) == \Yes \} $
\State comparisons $\leftarrow |\widetilde{R}(s_i)| \cdot |\widetilde{R}(s_j)|$
\EndIf
\If{$\FCount(v_i, v_j) < 0.3 \cdot\text{comparisons }$}
\State \Return \No
\Else ~\Return \Yes
\EndIf
\end{algorithmic}
\end{small}
\end{algorithm}

\noindent \textbf{\textsc{Assign-Final}}.
After obtaining $k$ clusters on the set of sampled points $\widetilde{V}$, we assign the remaining points using $\ACount$ scores, similar to the one described in $\Assign$. For every point that is not sampled, we first assign it to $s_1 \in S$, and if $\ACount(u, s_2, s_1) \geq 0.3 {|R(s_1)|}$, we re-assign it to $s_2$, and continue this process iteratively. After assigning all the points, the clusters are returned as output.

\subsection*{Theoretical Guarantees} Our algorithm first constructs a sample $\widetilde{V} \subseteq V$ and runs the greedy algorithm on this sampled set of points. Our main idea to ensure that good approximation of the $k$-center objective lies in identifying a good \textit{core} around each center. Using a sampling probability of $\gamma\log(n/\delta)/m$ ensures that we have at least $\Theta(\log(n/\delta))$ points from each of the optimal clusters in our sampled set $\widetilde{V}$. By finding the closest points using $\Count$ scores, we identify $O(\log(n/\delta))$ points around every center that are in the optimal cluster. Essentially, this forms the core of each cluster. These cores are then used for robust pairwise comparison queries (similar to Section~\ref{sec:farthest}), in our $\ApproxFarthest$ and $\Assign$ subroutines. 
We give the following theorem, which guarantees a constant, i.e., $O(1)$ approximation with high probability.
\begin{theorem}\label{thm:kcenterprob}
Given $p \leq 0.4$, a failure probability $\delta$, and $m = \Omega({\log^3(n/\delta)}/{\delta})$. Then, Algorithm~\ref{alg:greedykcenter_prob} achieves a $O(1)$-approximation for the $k$-center objective using $O(nk\log(n/\delta)+\frac{n^2}{m^2}k \log^2(n/\delta))$ oracle queries with probability $1-O(\delta)$.
\end{theorem}

\section{Hierarchical Clustering}\label{sec:hierarchical}
In this section, we present robust algorithms for agglomerative hierarchical clustering using single linkage and complete linkage objectives. The naive algorithms initialize every record as a singleton cluster and merge the closest pair of clusters iteratively. For a set of clusters $\mathcal{C}=\{C_1,\ldots,C_t\}$, the distance between any pair of clusters $C_i$ and $C_j$, for single linkage clustering, is defined as the minimum distance between any pair of records in the clusters,
$d_{SL}(C_1,C_2) = \min_{v_1\in C_1,v_2\in C_2} d(v_1,v_2)$.
For complete linkage, cluster distance is defined as the  maximum distance between any pair of records.
 All algorithms discussed in this section can be easily extended for complete linkage, and therefore we study single linkage clustering. The main challenge in implementing single linkage clustering in the presence of \emph{adversarial} noise is identification of  minimum value in a list of at most ${n\choose 2}$ distance values. In each iteration, the closest pair of clusters can be identified by using Algorithm~\ref{alg:max_adv} (with $t=2\log(n/\delta)$) to calculate the minimum over the set containing pairwise distances.
For this algorithm, Lemma~\ref{lem:single_linkage} shows that the pair of clusters merged in any iteration are a constant approximation of the optimal merge operation at that iteration. The proof of this lemma follows from Theorem~\ref{thm:max_adv}.
\begin{lemma}
Given a collection of clusters $\mathcal{C}=\{C_1,\ldots,C_r\}$, our algorithm to calculate the closest pair (using Algorithm~\ref{alg:max_adv}) identifies $C_1$ and $C_2$ to merge according to single  linkage objective if $d_{SL}(C_2,C_2) \leq (1+\mu)^3 \min_{C_i,C_j\in \mathcal{C}} d(C_i,C_j)$ with $1-\delta/n$ probability and requires  $O(n^2\log^2(n/\delta))$ queries.\label{lem:single_linkage}
\end{lemma}

\begin{algorithm}
\begin{small}
\begin{algorithmic}[1]
\State \textbf{Input} : Set of points $V$
\State \textbf{Output} : Hierarchy $H$
\State $H \leftarrow \{\{v\}\mid v\in V\},\mathcal{C}\leftarrow \{\{v\}\mid v\in V\}$
\For{$C_i\in \mathcal{C}$}
\State $\widetilde{C_i} \leftarrow $\textsc{NearestNeighbor} of $C_i$ among $\mathcal{C}\setminus\{C_i\}$ using Sec~\ref{sec:farthest}
\EndFor
\While{$|\mathcal{C}|>1$}
\State Let $(C_j,\widetilde{C_j})$ be the closest pair among $(C_i,\widetilde C_i), \forall C_i\in \mathcal{C}$
\State $C'\leftarrow C_j\cup \widetilde{C_j}$
\State Update Adjacency list of $C'$ with respect to $\mathcal{C}$
\State Add $C'$ as parent of $\widetilde{C_j}$ and $C_j$ in $H$.
\State $\mathcal{C} \leftarrow \left(\mathcal{C}\setminus \{C_j,\widetilde{C_j}\}\right)\cup \{C'\}$
\State $\widetilde{C'}\leftarrow$ \textsc{NearestNeighbor} of $C'$ from its adjacency list
\EndWhile
\State \Return $H$
\end{algorithmic}
\caption{{Greedy Algorithm}\label{alg:greedy_hierarchical}}
\end{small}
\end{algorithm}
\noindent \textbf{Overview.} Agglomerative clustering techniques are known to be inefficient. Each iteration of merge operation compares at most ${n\choose 2}$ pairs of distance values and the algorithm operates $n$ times to construct the hierarchy. This yields an overall query complexity of $O(n^3)$. To improve their query complexity, \texttt{SLINK} algorithm~\cite{sibson1973slink} was proposed to construct the hierarchy in $O(n^2)$ comparisons. To implement this algorithm with a comparison oracle, for every cluster $C_i \in \mathcal{C}$, we maintain an adjacency list containing every cluster $C_j$ in $\mathcal{C}$ along with a pair of records with the distance equal to the distance between the clusters. For example, the entry for $C_j$ in the adjacency list of $C_i$ contains the pair of records $(v_i,v_j)$ such that $d(v_i,v_j)=  \min_{v_i\in C_i,v_j\in C_j} d(v_i,v_j)$.
Algorithm~\ref{alg:greedy_hierarchical} presents the pseudo code for single linkage clustering under the adversarial noise model.
The algorithm is initialized with singleton clusters where every record is a separate cluster. Then, we identify the closest cluster for every $C_i \in \mathcal{C}$, and denote it by $\widetilde{C_i}$. This step takes $n$ nearest neighbor queries, each requiring $O(n \log^2(n/\delta))$ oracle queries. In every subsequent iteration, we identify the closest pair of clusters (Using section~\ref{sec:farthest}), say $C_j$ and $\widetilde{C_j}$ from $\mathcal{C}$. 

After merging these clusters, the data structure is updated as follows. To update the adjacency list, we need the pair of records with minimum distance between the merged cluster $C'\equiv C_j\cup \widetilde{C_j}$ and every other cluster $C_k \in \mathcal{C}$. In the previous iteration of the algorithm, we already have the minimum distance record pair for $(C_j,C_k)$ and $(\widetilde{C_j},C_k)$. Therefore a single query between these two pairs of records is sufficient to identify the minimum distance edge between $C'$ and $C_k$ (formally: 
$ d_{SL}(C_j\cup \widetilde{C_j},C_k)
     = \min \{d_{SL}(C_j,C_k),d_{SL}(\widetilde{C_j},C_k)\}\label{eq:mincomparison}$).
The nearest neighbor of the merged cluster is identified by running minimum calculation over its adjacency list. In Algorithm~\ref{alg:greedy_hierarchical}, as we identify closest pair of clusters, each iteration requires $O(n\log^2(n/\delta))$ queries. As our Algorithm terminates in at most $n$ iterations, it has an overall query complexity of $O(n^2\log^2(n/\delta))$.  In Theorem~\ref{thm:optimizedsingle}, we given an approximation guarantee for every merge operation of Algorithm~\ref{alg:greedy_hierarchical}.

\begin{theorem}
In any iteration, suppose the distance between a cluster $C_j\in \mathcal{C}$ and its identified nearest neighbor $\widetilde{C_j}$ is $\alpha$-approximation of its distance from the optimal nearest neighbor, then the distance between pair of clusters merged by  Algorithm~\ref{alg:greedy_hierarchical} is $\alpha (1+\mu)^3$ approximation of the optimal distance between the closest pair of clusters in $\mathcal{C}$ with a probability of $1-\delta$ using $O(n \log^2(n/\delta))$ oracle queries. \label{thm:optimizedsingle}
\end{theorem}

\noindent \textbf{Probabilistic Noise model.} The above discussed algorithms do not extend to the probabilistic noise due to constant probability of error for each query. However, when we are given a priori,  a partitioning of $V$ into clusters of size $>\log n$ such that the maximum distance between any pair of records in every cluster is smaller than $\alpha$ (a constant), Algorithm~\ref{alg:greedy_hierarchical} can be used to construct the hierarchy correctly. For this case, the algorithm to identify the closest and farthest pair of clusters is same as the one discussed in Section~\ref{sec:farthest}.

\revb{
\noindent Note that agglomerative clustering algorithms are known to require $\Omega(n^2)$ queries, which can be infeasible for million scale datasets. However, blocking based techniques present efficient heuristics to prune out low similarity pairs~\cite{papadakis2016comparative}. Devising provable algorithms with better time complexity is outside the scope of this work.
}
\section{Experiments}\label{sec:Experiments}
In this section, we evaluate the effectiveness of our techniques on various real world datasets and answer the following questions. \changes{\textbf{Q1:} Is quadruplet oracle practically feasible? How do the different types of queries compare in terms of quality and time taken by annotators? \textbf{Q2:}} Are proposed techniques robust to different levels of noise in oracle answers? \changes{ \textbf{Q3:}} How does the query complexity and solution quality of proposed techniques compare with optimum for varied levels of noise?  
 \subsection{Experimental Setup}

\noindent \textbf{Datasets.} We consider the following real-world datasets.\\
\textbf{(1)} \texttt{cities} dataset~\cite{uscities} comprises of $~36$K cities of the United States. The different features of the cities include state, county, zip code, population, time zone, latitude and longitude.\\
\changes{\textbf{(2)} \texttt{caltech} dataset comprises  11.4K images from 20 categories. The ground truth distance between records is calculated using the hierarchical categorization as described in~\cite{griffin2007caltech}.\\
\textbf{(3)} \amazon~dataset contains $7$K images and textual descriptions collected from \texttt{amazon.com}~\cite{he2016ups}. For obtaining the ground truth distances we use Amazon's hierarchical catalog.\\
\textbf{(4)} \monuments~dataset comprises of $100$ images belonging to $10$ tourist locations around the world.\\
}
\textbf{(5)} \texttt{dblp} contains 1.8M titles of computer science papers from different areas~\cite{zhang2018taxogen}. From these titles, noun phrases were extracted and a dictionary of all the phrases was constructed. Euclidean distance in word2vec embedding space is considered as the ground truth distance between concepts.

\noindent \textbf{Baselines.} We compare our techniques with  the optimal solution (whenever possible) and the following baselines.
(a) \texttt{Tour2} constructs a binary tournament tree over the entire dataset to compare the values and the root node corresponds to the identified maximum/minimum value (Algorithm~\ref{alg:maxTournament} with $\lambda=2$). This approach is an adaptation of the finding maximum algorithm in~\cite{DKMR15} with a difference that each query is not repeated multiple times to increase success probability. We also use them to identify the farthest and nearest point in the greedy $k$-center Algorithm~\ref{alg:greedykcenter_adv} and closest pair of clusters in  hierarchical clustering.  

(b) \texttt{Samp} considers a sample of $\sqrt{n}$ records and identifies the farthest/nearest by performing quadratic number of comparisons over the sampled points using \textsc{Count-Max}. For $k$-center, \texttt{Samp} considers a 
sample of $k\log n$ points to identify $k$ centers over these samples using the greedy algorithm. 
It then assigns all the remaining points to the identified centers by querying each record with every pair of center.

Calculating optimal clustering objective for $k$-center is NP-hard even in the presence of accurate pairwise distance~\cite{williamson2011design}. So, we compare the solution quality with respect to the greedy algorithm on the ground truth distances, denoted by \texttt{TDist}. For farthest, nearest neighbor and hierarchical clustering, \texttt{TDist} denotes the optimal technique that has access to ground truth distance between records.

Our algorithm is labelled \texttt{Far} for farthest identification, \texttt{NN} for nearest neighbor, \texttt{kC} for $k$-center and \texttt{HC} for hierarchical clustering with subscript $a$ denoting the adversarial model and $p$ denoting the probabilistic noise model.
All algorithms are implemented in C++ and run on a server with 64GB RAM. The reported results are averaged over 100 randomly chosen iterations. Unless specified, we set  $t=1$ in Algorithm~\ref{alg:max_adv} and $\gamma=2$ in Algorithm~\ref{alg:greedykcenter_prob}.

\noindent \textbf{Evaluation Metric.} For finding maximum and nearest neighbors, we compare different techniques by evaluating the true distance of the returned solution from the queried points. 
For $k$-center, we use the  objective value, i.e., maximum radius  of the returned clusters as the evaluation metric and compare against the true greedy algorithm (\texttt{TDist}) and other baselines. \changes{For datasets where ground truth clusters are known (\amazon, \caltech~and \monument), we use F-score over intra-cluster pairs for comparing it with the baselines~\cite{galhotra2018robust}.} For hierarchical clustering, we compute the pairs of clusters merged in every iteration and compare the average true distance between these clusters.
In addition to the quality of returned solution, we compare the query complexity and running time of the proposed techniques with the baselines described above.

\changes{\noindent \textbf{Noise Estimation.} For \cities, \amazon, \caltech,\ and \monuments\ datasets, we ran a user study on Amazon Mechanical Turk to estimate the noise in oracle answers over a small sample of the dataset, often referred to as the validation set. Using crowd responses, we trained a classifier (random forest~\cite{scikit} obtained the best results) using active learning to act as the quadruplet oracle, and reduce the number of queries to the crowd. Our active learning algorithm~\cite{modal} uses a batch of $20$ queries and we stop it when the classifier accuracy on the validation set does not improve by more than $0.01$~\cite{gokhale2014corleone}. \changes{To efficiently construct a small set of candidates for active learning and pruning low similarity pairs for \texttt{dblp}, we employ token based blocking~\cite{papadakis2016comparative} for the datasets.} For the synthetic oracle, we simulate quadruplet oracle with different values of the noise parameters.
}

\changes{
\subsection{User study}\label{sec:userstudy}
In this section, we evaluate the users ability to answer quadruplet queries and compare it with other types of queries.

\noindent \textbf{Setup}. We ran a user study on Amazon Mechanical Turk platform for four datasets \cities, \amazon, \caltech~and \monument. 
We consider the ground truth distance between record pairs and discretize them into buckets, and assign a pair of records to a bucket if the distance falls within its range. For every pair of buckets, we query a random subset of $\log n$ quadruplet oracle queries (where $n$ is size of dataset). Each query is answered by three different crowd workers and a majority vote is taken as the answer to the query.

}
 \begin{figure}
 \centering
 \includegraphics[width=0.49\columnwidth ]{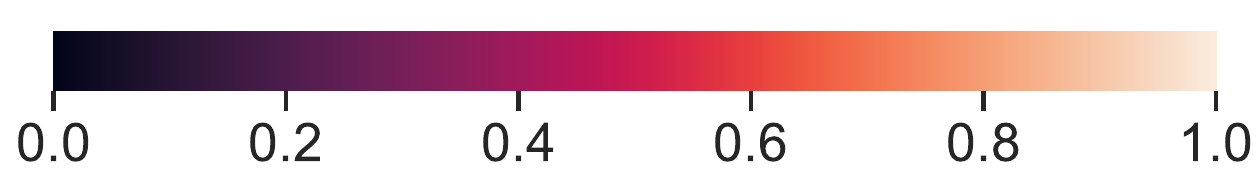}\\
 \vspace{-2mm}
  \subfigure[width=0.46\columnwidth][\texttt{caltech}]{\includegraphics[width=0.46\columnwidth ]{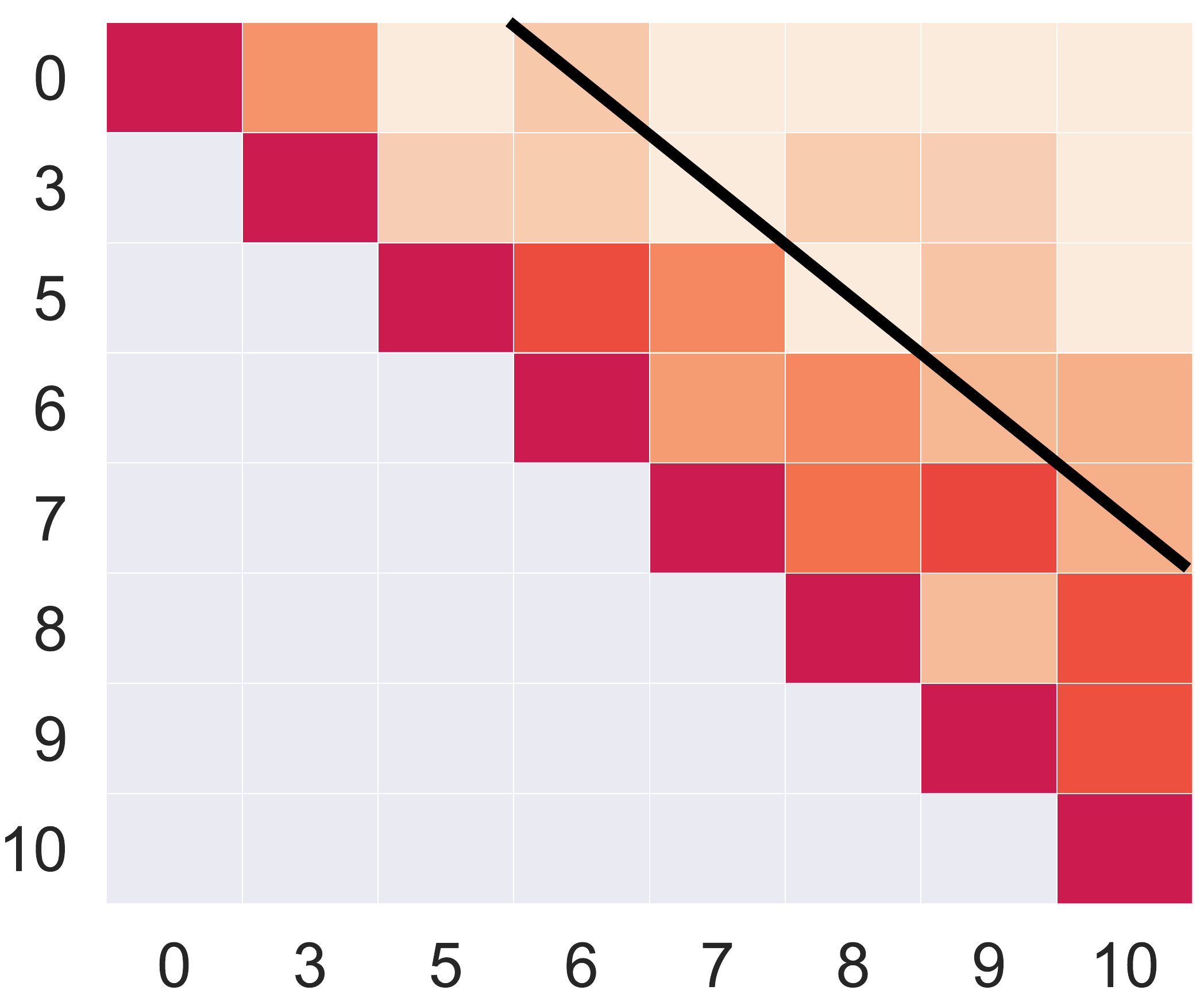}\label{fig:caltechheatmap}}
  ~\subfigure[width=0.46\columnwidth][\texttt{amazon}]{\includegraphics[width=0.46\columnwidth ]{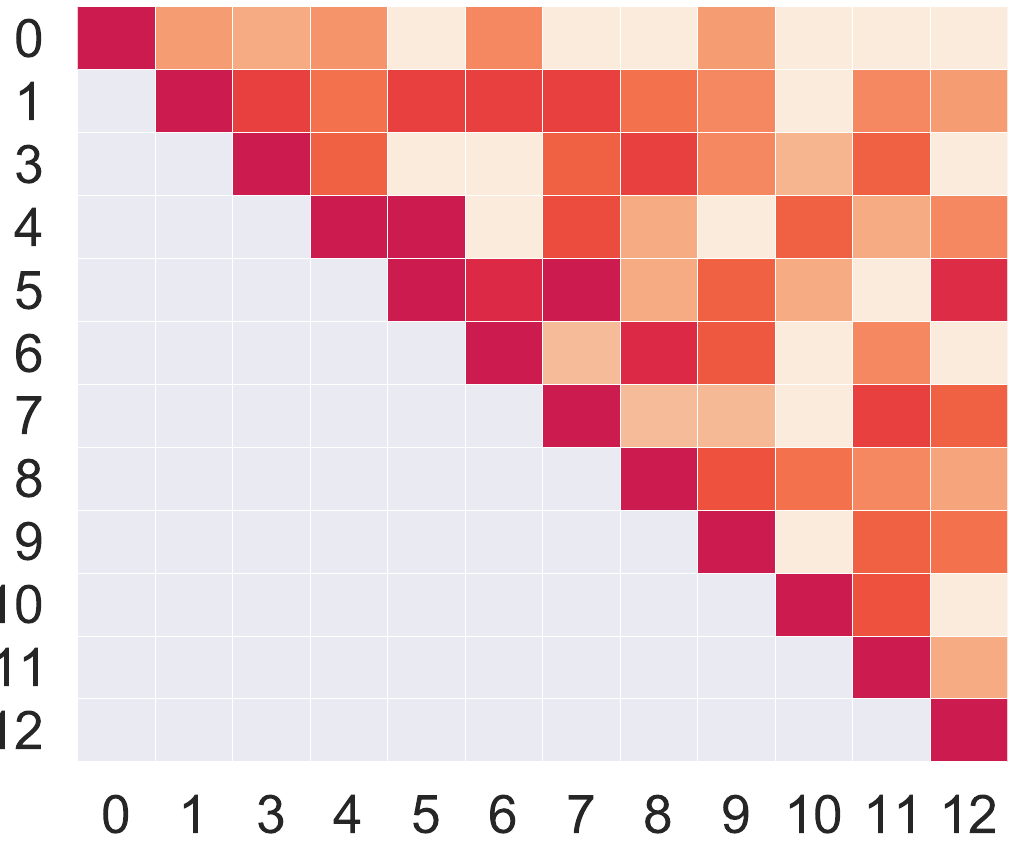}\label{fig:amazonheatmap}}
  \vspace{-3mm}
   \caption{\changes{Accuracy values (denoted by the color of a cell) for different distance ranges observed during our user study. The diagonal entries refer to the quadruplets with similar distance between the corresponding pairs and the distance increases as we go further away from the diagonal.}}\label{fig:heatmap}
 \vspace{-1mm}
 \end{figure}

\changes{
\subsubsection{Qualitative Analysis of Oracle}
In Figure~\ref{fig:heatmap}, for every pair of buckets, using a heat map, we plot the accuracy of answers obtained from the crowd workers for quadruplet queries. For all datasets, average accuracy of quadruplet queries is more than $0.83$ and the accuracy is minimum whenever both pairs of records belong to the same bucket (as low as $0.5$).  However, we observe varied behavior across datasets as the distance between considered pairs increases.

For the \caltech\ dataset, we observe that when the ratio of the distances is more than $1.45$ (indicated by a black line in the Figure~\ref{fig:caltechheatmap})
, there is no noise (or close to zero noise) observed in the query responses. As we observe a sharp decline in noise as the distance between the pairs increases, it suggests that adversarial noise is satisfied for this dataset. We observe a similar pattern for the \cities~and \monuments~ datasets. For the \amazon~dataset, we observe that there is substantial noise across all distance ranges (See Figure~\ref{fig:amazonheatmap}) rather than a sharp decline, suggesting that the probabilistic model is satisfied.

\subsubsection{Comparison with pairwise querying mechanisms} To evaluate the benefit of quadruplet queries, we compare the quality of quadruplet comparison oracle answers with the following pairwise oracle query models. 
(a) Optimal cluster query: This query asks questions of type `do $u$ and $v$ refer to same/similar type?'. (b) Distance query: How similar are the records $x$ and $y$? In this query, the annotator scores the similarity of the pair within $1$ to $10$. 

\noindent We make the following observations. (i) Optimal cluster queries are answered correctly only if the ground truth clusters refer to different entities (each cluster referring to a distinct entity). Crowd workers tend to answer `No' if the pair of records refer to different entities. Therefore, we observe high precision (more than $0.90$) but low recall ($0.50$ on \amazon~and $0.30$ on \caltech~for $k=10$) of the returned labels.
(ii) We observed very high variance in the distance estimation query responses. For all record pairs with identical entities, the users returned distance estimates that were within $20\%$ of the correct distances. In all other cases, we observe the estimates to have errors of upto $50\%$.
\changes{We provide more detailed comparison on the quality of clusters identified by pairwise query responses along with quadruplet queries in the next section.}
\begin{figure}
    \centering
    \subfigure[\changes{Farthest,higher is better}]{\includegraphics[width=0.48\columnwidth]{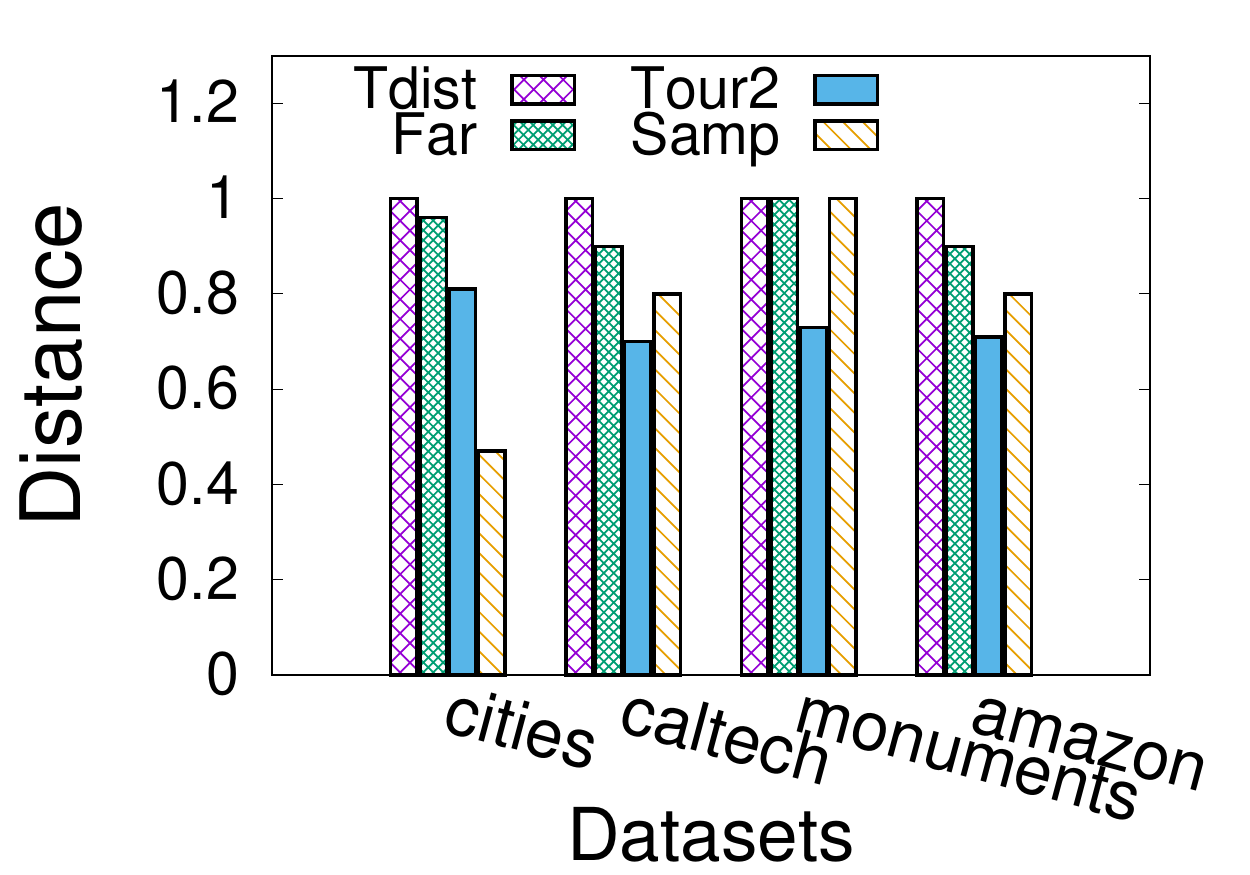}} 
    \subfigure[\changes{Nearest Neighbor (NN), lower is better}]{\includegraphics[width=0.48\columnwidth]{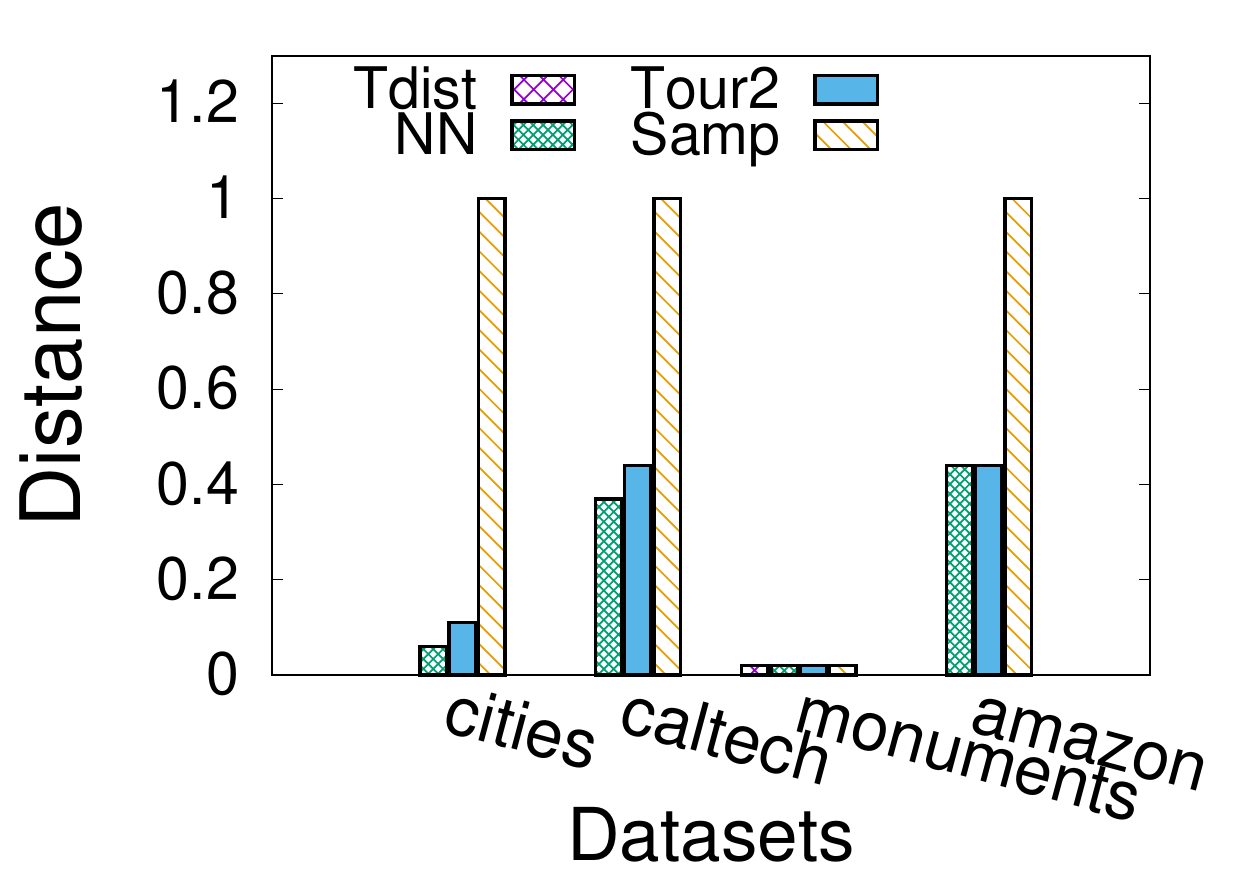}} 
    \vspace*{-4ex}
    \caption{\changes{Comparison of farthest and NN techniques for crowdsourced oracle queries.}}
    \vspace*{-1ex}
    \label{fig:realdata}
\end{figure}

}
\changes{\subsection{Crowd Oracle: Solution Quality \& Query Complexity }
In this section, we compare the quality of our proposed techniques for the datasets on which we performed the  user study. Following the findings of Section~\ref{sec:userstudy}, we use probabilistic model based algorithm for \amazon\ (with $p=0.50$) and  adversarial noise model based algorithm for \caltech, \monument\ and \cities.
\begin{figure*}
\includegraphics[width=\textwidth]{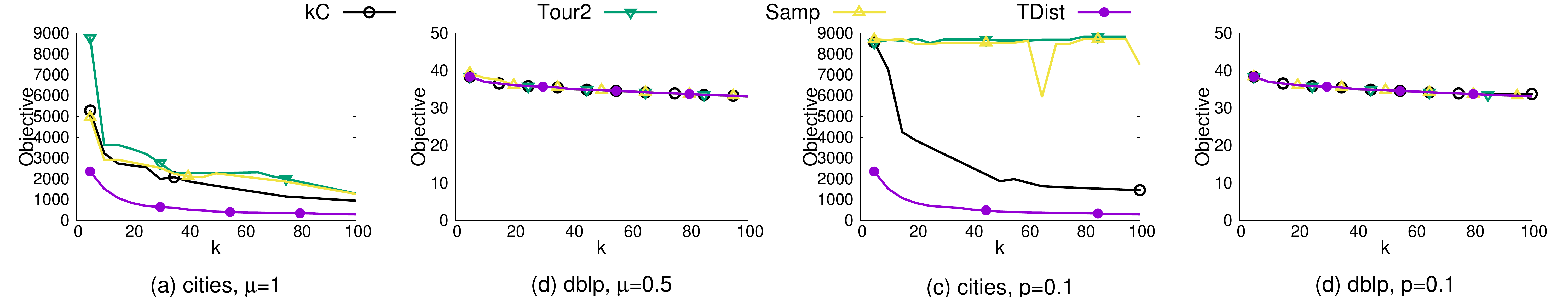}
\vspace*{-5ex}
\caption{$k$-center clustering objective comparison for adversarial and probabilistic noise model. \label{fig:kcenteradv}}
\vspace*{-3ex}
\end{figure*}

\noindent \textbf{Finding Max and Farthest/Nearest Neighbor.} 
\changes{Figure~\ref{fig:realdata} compares the quality of farthest and nearest neighbor (NN) identified by proposed techniques along with other baselines. The values are normalized according to the maximum value to present all datasets on the same scale. Across all datasets, the point identified by \texttt{Far} and \texttt{NN} is closest to the optimal value, \texttt{TDist}. In contrast, the farthest returned by \texttt{Tour2} is better than that of \texttt{Samp} for \cities~dataset but not for \caltech, \monument~and \amazon. 
We found that this difference in quality across datasets is due to varied distance distribution between pairs. The \cities~dataset has a skewed distribution of distance between record pairs, leading to a unique optimal solution to the farthest/NN problem. Due to this reason, the set of records sampled by \texttt{Samp} does not contain any record that is a  good approximation of the optimal farthest. However, ground truth distances between record pairs in \amazon, \monument~and \caltech~are less skewed with more than $\log n$ records satisfying the optimal farthest point for all queries. Therefore, \texttt{Samp} performs better than \texttt{Tour2} on these datasets.} We observe \texttt{Samp} performs worse for NN because our sample does not always contain the closest point.

\noindent \textbf{$k$-center Clustering.}
\changes{We evaluate the F-score\footnote{\changes{Optimal clusters are identified from the original source of the datasets (\amazon~and \caltech) and manually for \monument.}} of the clusters generated by our techniques along with baselines and techniques for pairwise optimal query mechanism (denoted as \texttt{Oq})\footnote{
\changes{We report the results on the sample of queries asked to the crowd as opposed to training a classifier because the classifier generates noisier results and has poorer F-score than the quality of labels generated by crowdsourcing}}. Table~\ref{tab:kcenter} presents the summary of our results for different values of $k$. Across all datasets, our technique achieves more than $0.90$ F-score. On the other hand, \texttt{Tour2} and \texttt{Samp} do not identify the ground truth clusters correctly, leading to low F-score. Similarly, \texttt{Oq} achieves poor recall (and hence low F-score) as it labels many record pairs to belong to separate clusters. For example, a \textit{frog} and a \textit{butterfly} belong to the same optimal cluster for \caltech\ (k=10) but the two records are assigned to different clusters by \texttt{Oq}.
}

\begin{figure}
    \centering
    \subfigure[\changes{Single Linkage}]{\includegraphics[width=0.49\columnwidth]{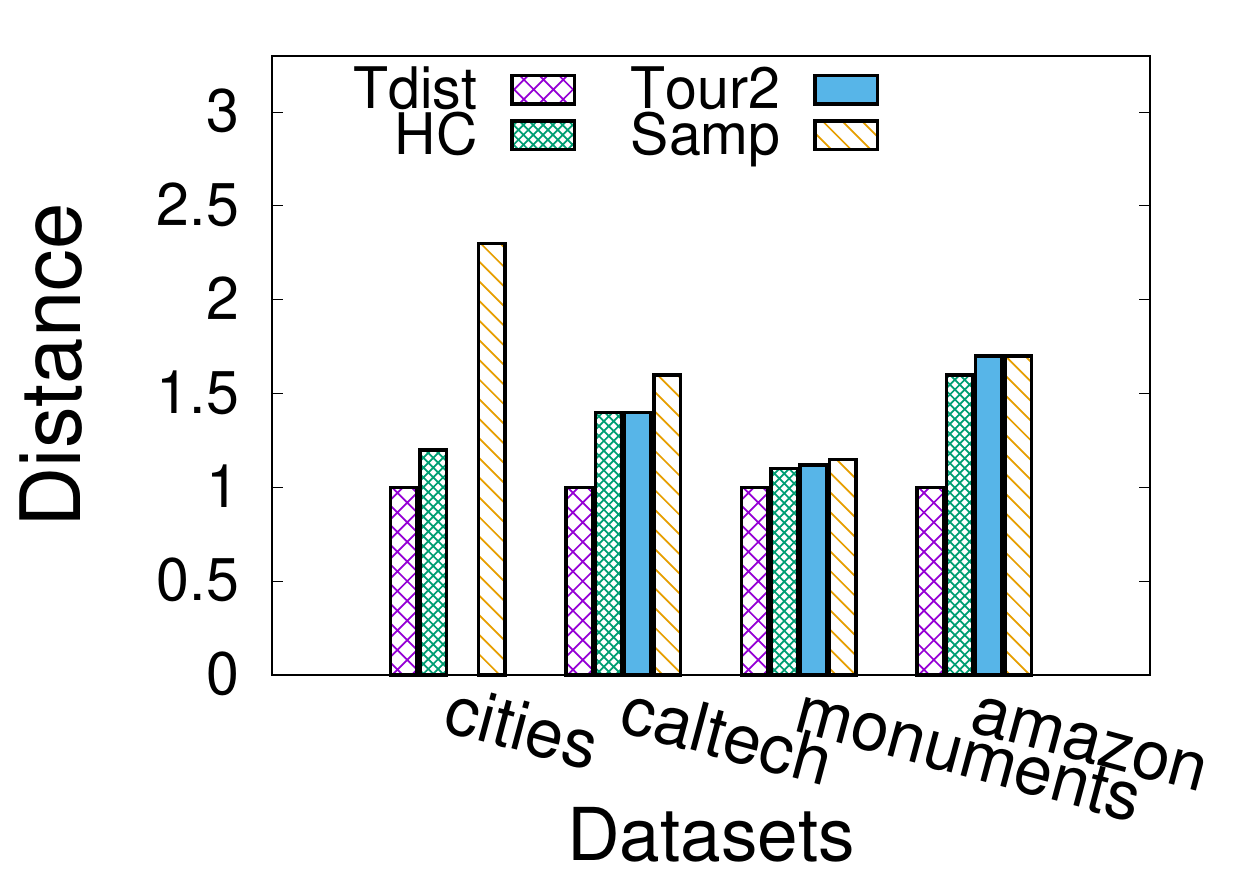}} 
    \subfigure[\changes{Complete Linkage}]{\includegraphics[width=0.49\columnwidth]{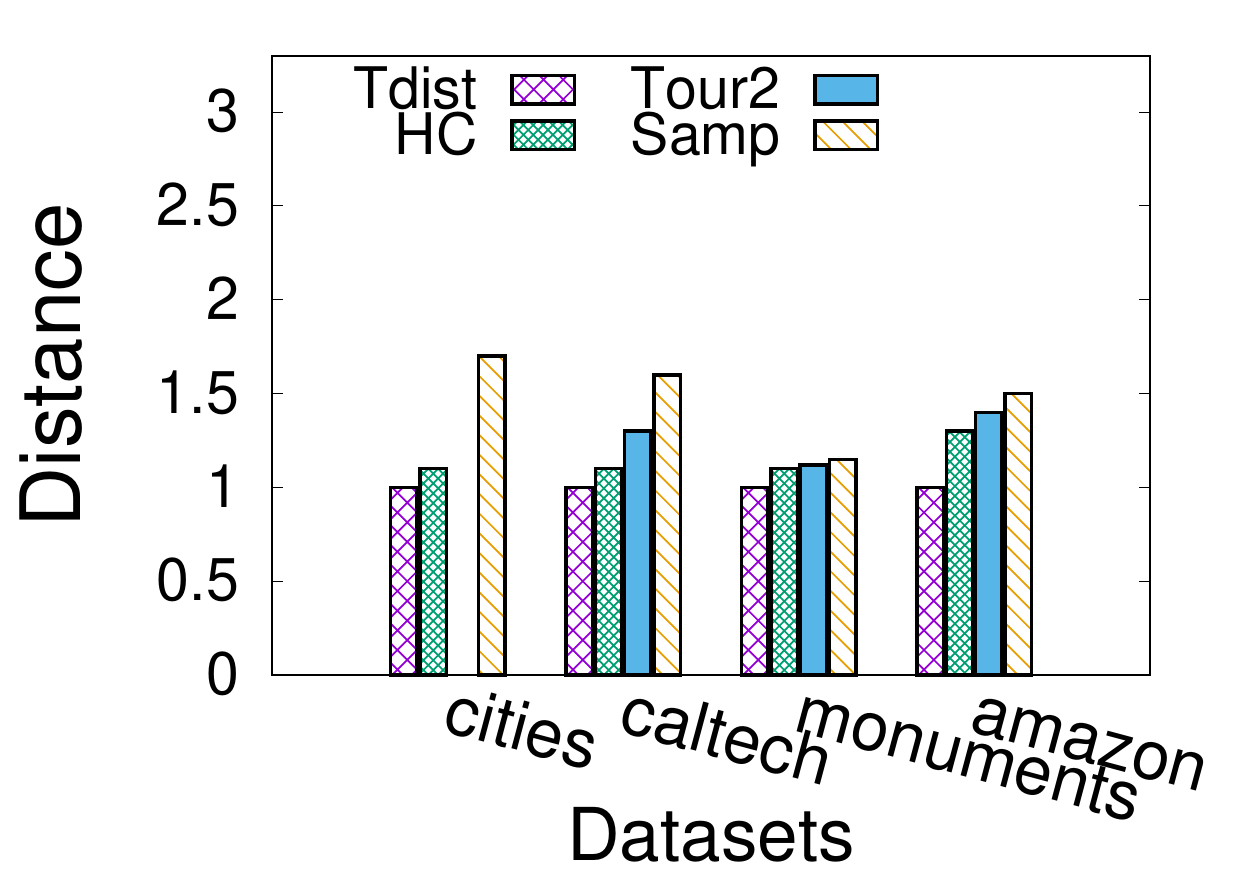}} 
    \vspace*{-4ex}
    \caption{\changes{Comparison of Hierarchical clustering techniques with crowdsourced oracle.}}
    \vspace*{-2ex}
    \label{fig:realhc}
\end{figure}

\noindent \textbf{Hierarchical Clustering.}
Figure~\ref{fig:realhc} compares the average distance of the merged clusters across different iterations of the agglomerative clustering algorithm. \texttt{Tour2} has $O(n^3)$ complexity and does not run for \cities~dataset in less than $48$ hrs. The objective value of different techniques are normalized by the optimal value with \texttt{Tdist} denoting $1$. For all datasets, \texttt{HC}  performs better than \texttt{Samp} and \texttt{Tour2}. Among datasets, the quality of hierarchies generated for \monuments~is similar for all techniques due to low noise.

\begin{table}
\footnotesize
    \centering
    \begin{tabular}{|c|c|c|c|c|}
    \hline
        Technique & \texttt{kC}  & \texttt{Tour2} & \texttt{Samp}&\texttt{
        Oq}* \\
        \hline
    \caltech\ ($k=10$) & $\mathbf{1}$ & $0.88$& $0.91$ & $0.45$\\
    \caltech\ ($k=15$) & $\mathbf{1}$ & $0.89$ & $0.88$ & $0.49$\\
    \caltech\ ($k=20$) & $\mathbf{0.99}$ & $0.93$&$0.87$ & $0.58$ \\
    \hline
    \monument\ ($k=5$) &$\mathbf{1}$ & $0.95$ &$0.97$&0.77\\
    \hline 
   \amazon\ ($k=7$)     & $\mathbf{0.96}$ & $0.74$ & $0.57$&$0.48$ \\
   \amazon\ ($k=14$)     & $\mathbf{0.92}$ & $0.66$ & $0.54$ & $0.72$   \\
    \hline
    \end{tabular}
    \caption{ \changes{F-score comparison of k-center clustering. \texttt{Oq} is marked with $*$ as it was computed on a sample of $150$ pairwise queries to the crowd$\footnotesize^3$. All other techniques were run on the complete dataset using a classifier.}}
    \label{tab:kcenter}
    \vspace{-2mm}
\end{table}
\noindent \textbf{Query Complexity.}
To ensure scalability,  we trained active learning based classifier for all the aforementioned experiments. In total, \amazon, \cities, and \caltech\ required $540$ (cost: $\$32.40$), $220$ (cost: $\$13.20$) and $280$ (cost: $\$16.80$) queries to the crowd respectively. 

\subsection{Simulated Oracle: Solution Quality \& Query Complexity \label{sec:quality}}

In this section, we compare the robustness of the techniques where the query response is simulated synthetically for given $\mu$ and $p$.}

\noindent \textbf{Finding Max and Farthest/Nearest Neighbor.} 
In Figure~\ref{fig:farthestadv}, $\mu=0$ denotes the setting where the oracle answers all queries correctly. In this case, \texttt{Far} and \texttt{Tour2} identify the optimal solution but \texttt{Samp} does not identify the optimal solution for \texttt{cities}. 
In both datasets, \texttt{Far} identifies the correct farthest point for $\mu < 1$. Even with an increase in noise ($\mu$), we observe that the farthest is always at a distance within $4$ times the optimal distance (See Fig~\ref{fig:farthestadv}).
We observe that the quality of farthest identified by \texttt{Tour2} is close to that of \texttt{Far} for smaller $\mu$ because the optimal farthest point $\vmax$ has only a few points in the confusion region $C$ (See Section~\ref{sec:finding_max}) that contains the points that are close to $\vmax$. For e.g., less than $10\%$ are present in $C$ when $\mu=1$ for \texttt{cities} dataset, i.e., less than $10\%$ points return erroneous answer when compared with $\vmax$.

\begin{figure}
    \centering
     \vspace*{-1ex}
      \subfigure[{\cities--Adversarial} \label{fig:farthestadv}]{\includegraphics[width=0.48\columnwidth]{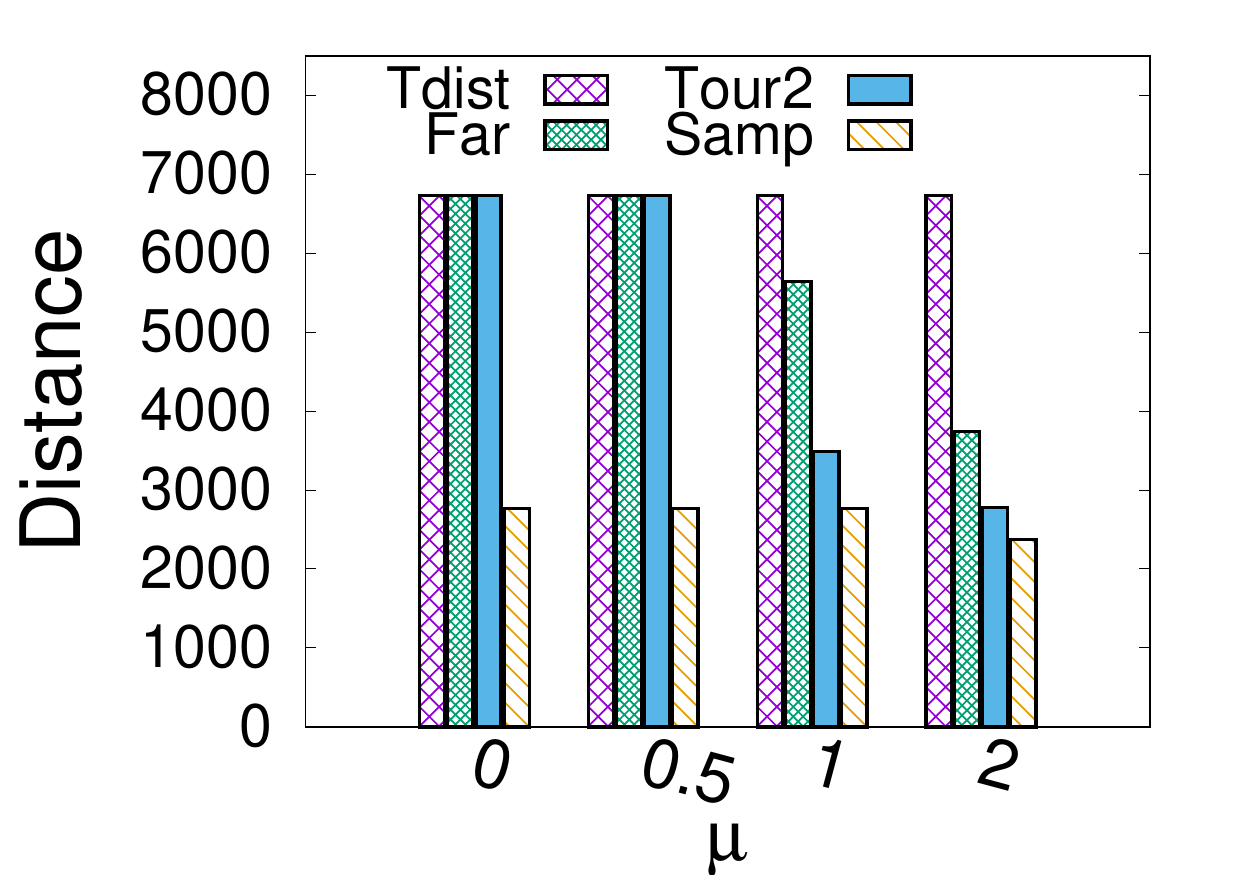}} 
    \subfigure[\texttt{cities}--Probabilistic\label{fig:farthestprob}]{\includegraphics[width=0.48\columnwidth]{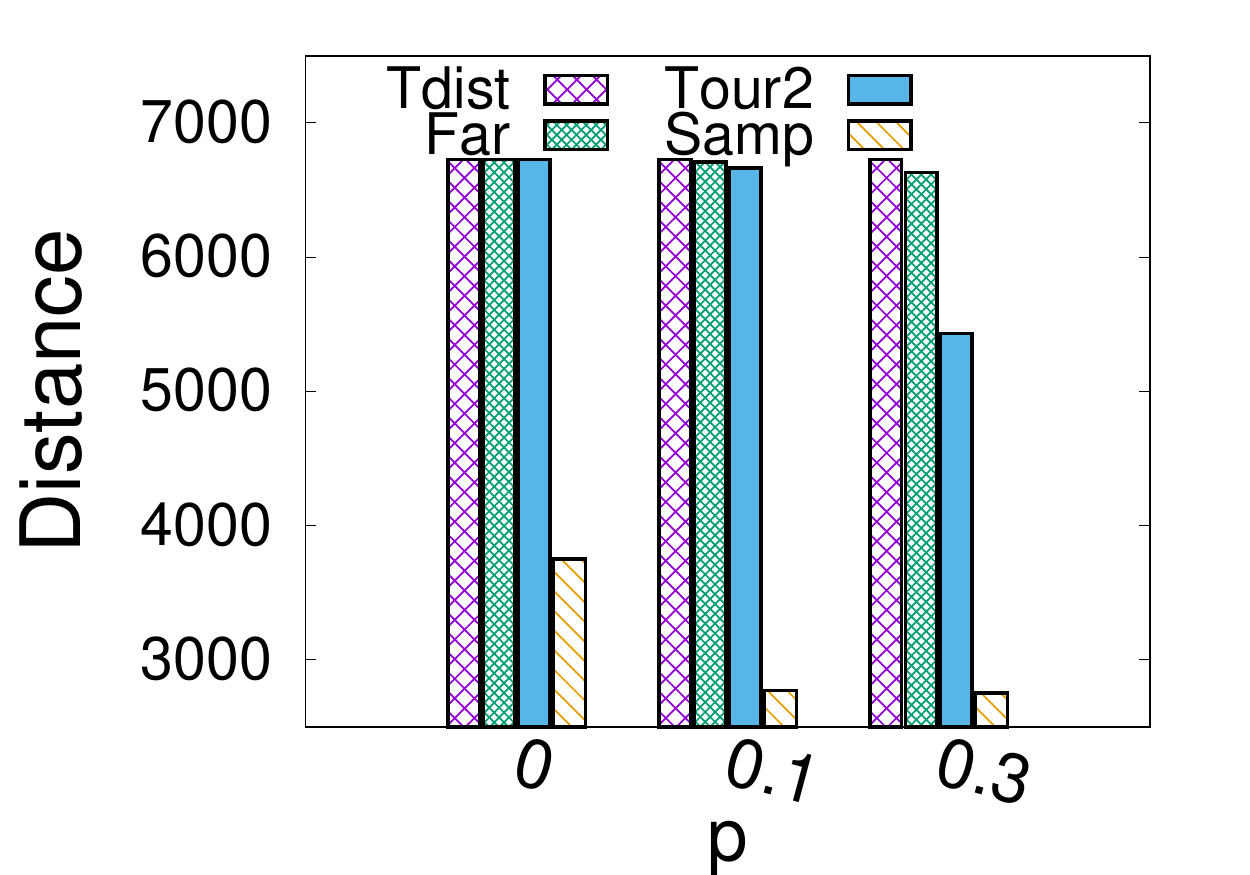}} 
    \vspace*{-4ex}
    \caption{Comparison of farthest identification techniques for adversarial and probabilistic noise models.}
    \vspace*{-1ex}
    
\end{figure}

\noindent In Figure~\ref{fig:farthestprob}, we compare the true distance of the identified farthest points for the case of probabilistic noise with error probability $p$. We observe that \texttt{Far}$_p$ identifies points with distance values very close to the farthest distance \texttt{Tdist}, across all data sets and error values. This shows that \texttt{Far} performs significantly better than the theoretical approximation presented in Section~\ref{sec:finding_max}.
On the other hand, the solution returned by \texttt{Samp} is more than $4\times$ smaller than the value returned by  \texttt{Far}$_p$ for an error probability of $0.3$. \texttt{Tour2} has a similar performance as that of \texttt{Far}$_p$ for $p \leq 0.1$,  but we observe a decline in solution quality for higher noise ($p$) values.

In Figures~\ref{fig:nearestadv},~\ref{fig:nearestprob}, we compare the true distance of the identified nearest neighbor with different baselines. 

\texttt{NN} shows superior performance as compared to \texttt{Tour2} across all error values. This justifies the lack of robustness of \texttt{Tour2} as discussed in Section~\ref{sec:finding_max}.  The solution quality of \texttt{NN} does not worsen with increase in error.   We omit \texttt{Samp} from the plots because the returned points had very poor performance (as bad as 700 even in the absence of error). We observed similar behavior for other datasets. In terms of query complexity, \texttt{NN} requires around $53 \times 10^3$ queries for \texttt{cities} dataset and the number of queries grow linearly with the dataset size. Among baselines, \texttt{Tour2} uses  $37 \times 10^3$ queries and \texttt{Samp} uses $18\times 10^3$.

\begin{figure}
    \centering
    \vspace{-3mm}
    \subfigure[\texttt{cities}--Adversarial\label{fig:nearestadv}]{\includegraphics[width=0.48\columnwidth]{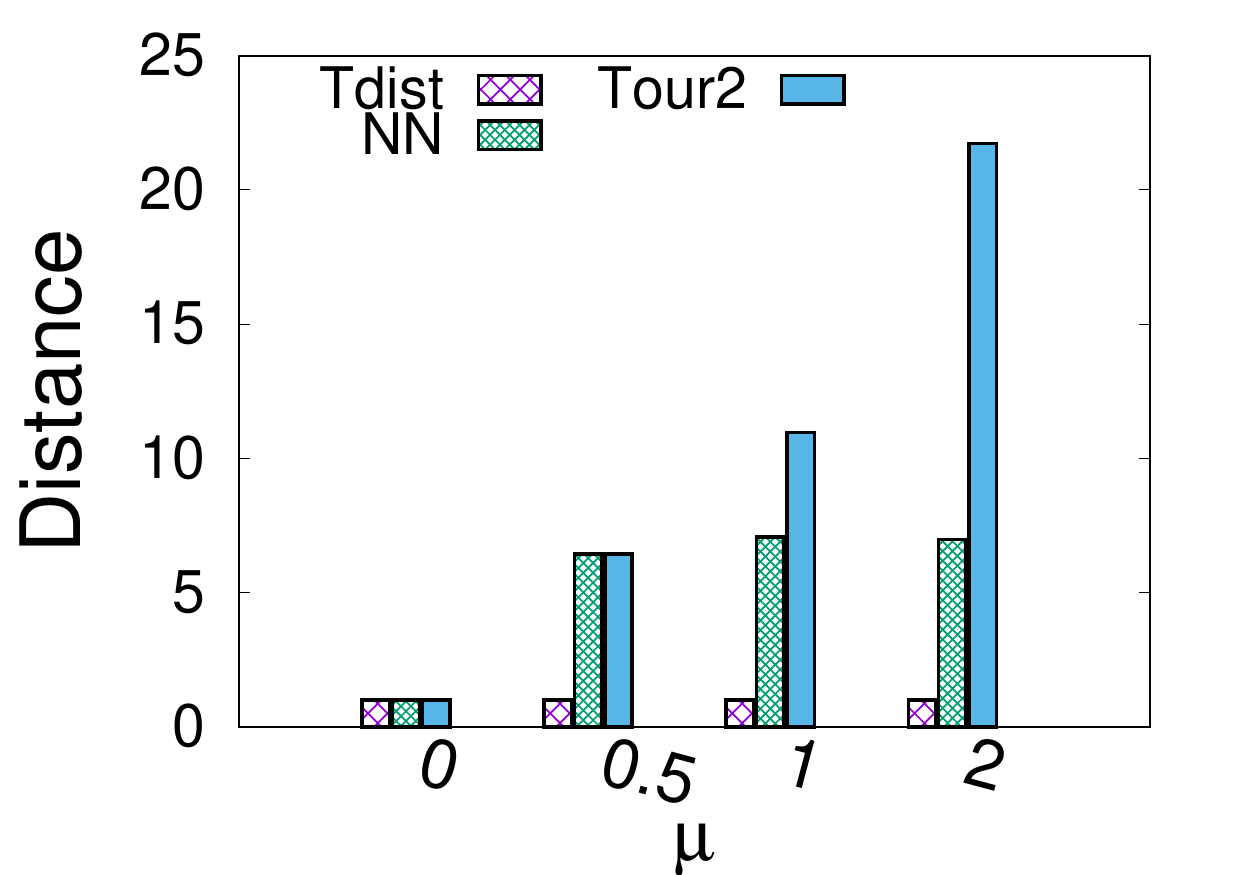}} 
    \subfigure[\texttt{cities}--Probabilistic\label{fig:nearestprob}]{\includegraphics[width=0.48\columnwidth]{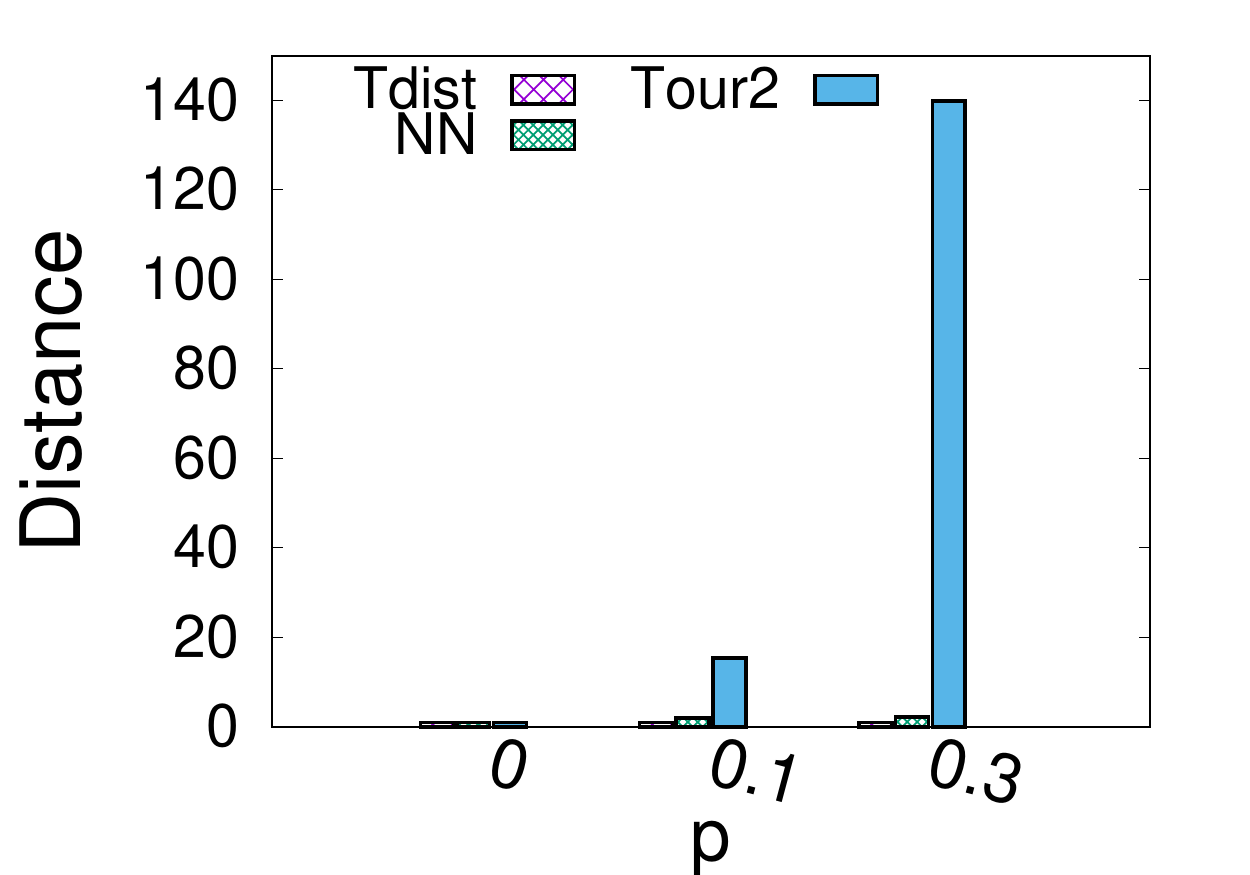}} 
    \vspace*{-4ex}
    \caption{{ Comparison of nearest neighbor techniques for adversarial and probabilistic noise model (lower is better).}}
     \vspace*{-1ex}
    \label{fig:nearest}
\end{figure}

\noindent {\it ``In conclusion, we observe that our techniques achieve the best quality across all data sets and error values, while \texttt{Tour2} performs similar to \texttt{Far} for low error, and its quality degrades with increasing error."}

\noindent \textbf{$k$-center Clustering.}
Figure~\ref{fig:kcenteradv} compares the $k$-center objective of the returned clusters for varying $k$ in the adversarial \changes{and probabilistic} noise model. \texttt{Tdist} denotes the best possible clustering objective, which is guaranteed to be a 2-approximation of the optimal objective. 
The set of clusters returned by \texttt{kC} are consistently very close to \texttt{TDist} across all datasets, validating the theory. For higher values of $k$, \texttt{kC} approaches closer to \texttt{TDist}, thereby improving the approximation guarantees. 
The quality of clusters identified by \texttt{kC} are similar to that of \texttt{Tour2} and \texttt{Far} for adversarial noise (Figure~\ref{fig:kcenteradv}a,b) but considerably better for probabilistic noise (Figure~\ref{fig:kcenteradv}c,d).

\noindent \textbf{Running time.} Table~\ref{tab:efficiency} compares the \revb{running time and the number of required quadruplet comparisons} for various problems under adversarial noise model with $\mu=1$ for the largest \texttt{dblp} dataset. 
\texttt{Far} and \texttt{NN}  requires less than $6$ seconds for both adversarial and probabilistic error models.  Our $k$-center clustering technique requires less than 450 min to identify $50$ centers for \texttt{dblp} dataset across different noise models; the running time grows linearly with $k$. 
While the running time of our algorithms are slightly higher than \texttt{Tour2} for farthest, nearest and $k$-center, \texttt{Tour2} did not finish in $48$ hrs due to $O(n^3)$ running time for single and complete linkage hierarchical clustering.  We observe similar performance for the probabilistic noise model. 
\revb{Note that even though the number of comparisons are in millions, this dataset requires only $740$ queries to the crowd workers to train the classifier.}

\begin{table}
\footnotesize
    \centering
    \begin{tabular}{|c|c|c|c|c|c|c|}
    \hline
 Problem & \multicolumn{2}{c|}{Our Approach} & \multicolumn{2}{c|}{\texttt{Tour2}} & \multicolumn{2}{c|}{\texttt{Samp}}\\\hline
 &Time&\# Comp&Time&\# Comp&Time& \# Comp\\\hline
     Farthest    & 0.1&2.2M& 0.06 &2M & 0.07 & 1M \\\hline
      Nearest   &0.075&2M & 0.07&2M & 0.61& 1M  \\\hline
      kC (k=50)   & 450&120M & 375.3&95M & 477& 105M\\\hline
       Single  Linkage  & 1813&990M & \multicolumn{2}{c|}{DNF} &  1760&940M \\\hline
      Complete Linkage & 1950&940M &\multicolumn{2}{c|}{DNF} & 1940&920M \\\hline
    \end{tabular}
    \caption{\revb{Running time (in minutes) and number of quadruplet comparisons (denoted by \# Comp, in millions) of different techniques for \texttt{dblp} dataset under the adversarial noise model with $\mu=1$. DNF denotes `did not finish'.}}
    \vspace*{-5ex}
    \label{tab:efficiency}
\end{table}   

\section{Conclusion}
In this paper, we show how algorithms for various basic tasks such as finding maximum, nearest neighbor, $k$-center clustering, and agglomerative hierarchical clustering can be designed using distance based comparison oracle in presence of noise. We believe our techniques can be useful for other clustering tasks such as $k$-means and $k$-median, and we leave those as future work.

\bibliographystyle{plain}
\bibliography{references}

\begin{thebibliography}{10}

\bibitem{googlevision}
Google vision api \url{https://cloud.google.com/vision}.

\bibitem{uscities}
United states cities database.
\newblock \url{https://simplemaps.com/data/us-cities}.

\bibitem{ailon2018approximate}
Nir Ailon, Anup Bhattacharya, Ragesh Jaiswal, and Amit Kumar.
\newblock Approximate clustering with same-cluster queries.
\newblock In {\em 9th Innovations in Theoretical Computer Science Conference
  (ITCS 2018)}, volume~94, page~40. Schloss Dagstuhl--Leibniz-Zentrum fuer
  Informatik, 2018.

\bibitem{ajtai2009sorting}
Mikl{\'o}s Ajtai, Vitaly Feldman, Avinatan Hassidim, and Jelani Nelson.
\newblock Sorting and selection with imprecise comparisons.
\newblock In {\em International Colloquium on Automata, Languages, and
  Programming}, pages 37--48. Springer, 2009.

\bibitem{arora2018hd}
Akhil Arora, Sakshi Sinha, Piyush Kumar, and Arnab Bhattacharya.
\newblock Hd-index: pushing the scalability-accuracy boundary for approximate
  knn search in high-dimensional spaces.
\newblock {\em Proceedings of the VLDB Endowment}, 11(8):906--919, 2018.

\bibitem{ashtiani2016clustering}
Hassan Ashtiani, Shrinu Kushagra, and Shai Ben-David.
\newblock Clustering with same-cluster queries.
\newblock In {\em Advances in neural information processing systems}, pages
  3216--3224, 2016.

\bibitem{ben2018}
Shai Ben-David.
\newblock Clustering-what both theoreticians and practitioners are doing wrong.
\newblock In {\em Thirty-Second AAAI Conference on Artificial Intelligence},
  2018.

\bibitem{braverman2016parallel}
Mark Braverman, Jieming Mao, and S~Matthew Weinberg.
\newblock Parallel algorithms for select and partition with noisy comparisons.
\newblock In {\em Proceedings of the forty-eighth annual ACM symposium on
  Theory of Computing}, pages 851--862, 2016.

\bibitem{braverman2008noisy}
Mark Braverman and Elchanan Mossel.
\newblock Noisy sorting without resampling.
\newblock In {\em Proceedings of the nineteenth annual ACM-SIAM symposium on
  Discrete algorithms}, pages 268--276. Society for Industrial and Applied
  Mathematics, 2008.

\bibitem{bressan2019correlation}
Marco Bressan, Nicol{\`o} Cesa-Bianchi, Andrea Paudice, and Fabio Vitale.
\newblock Correlation clustering with adaptive similarity queries.
\newblock In {\em Advances in Neural Information Processing Systems}, pages
  12510--12519, 2019.

\bibitem{chatziafratis2018hierarchical}
Vaggos Chatziafratis, Rad Niazadeh, and Moses Charikar.
\newblock Hierarchical clustering with structural constraints.
\newblock {\em arXiv preprint arXiv:1805.09476}, 2018.

\bibitem{chien2018query}
I~Chien, Chao Pan, and Olgica Milenkovic.
\newblock Query k-means clustering and the double dixie cup problem.
\newblock In {\em Advances in Neural Information Processing Systems}, pages
  6649--6658, 2018.

\bibitem{choudhury2019top}
Tuhinangshu Choudhury, Dhruti Shah, and Nikhil Karamchandani.
\newblock Top-m clustering with a noisy oracle.
\newblock In {\em 2019 National Conference on Communications (NCC)}, pages
  1--6. IEEE, 2019.

\bibitem{ciceri2015crowdsourcing}
Eleonora Ciceri, Piero Fraternali, Davide Martinenghi, and Marco Tagliasacchi.
\newblock Crowdsourcing for top-k query processing over uncertain data.
\newblock {\em IEEE Transactions on Knowledge and Data Engineering},
  28(1):41--53, 2015.

\bibitem{DKMR15}
Susan Davidson, Sanjeev Khanna, Tova Milo, and Sudeepa Roy.
\newblock Top-k and clustering with noisy comparisons.
\newblock {\em ACM Trans. Database Syst.}, 39(4), December 2015.

\bibitem{dushkin2018top}
Eyal Dushkin and Tova Milo.
\newblock Top-k sorting under partial order information.
\newblock In {\em Proceedings of the 2018 International Conference on
  Management of Data}, pages 1007--1019, 2018.

\bibitem{emamjomeh2018adaptive}
Ehsan Emamjomeh-Zadeh and David Kempe.
\newblock Adaptive hierarchical clustering using ordinal queries.
\newblock In {\em Proceedings of the Twenty-Ninth Annual ACM-SIAM Symposium on
  Discrete Algorithms}, pages 415--429. SIAM, 2018.

\bibitem{feige1994computing}
Uriel Feige, Prabhakar Raghavan, David Peleg, and Eli Upfal.
\newblock Computing with noisy information.
\newblock {\em SIAM Journal on Computing}, 23(5):1001--1018, 1994.

\bibitem{firmanier}
Donatella Firmani, Barna Saha, and Divesh Srivastava.
\newblock Online entity resolution using an oracle.
\newblock {\em PVLDB}, 9(5):384–395, 2016.

\bibitem{galhotra2018robust}
Sainyam Galhotra, Donatella Firmani, Barna Saha, and Divesh Srivastava.
\newblock Robust entity resolution using random graphs.
\newblock In {\em Proceedings of the 2018 International Conference on
  Management of Data}, pages 3--18, 2018.

\bibitem{geissmann2017sorting}
Barbara Geissmann, Stefano Leucci, Chih-Hung Liu, and Paolo Penna.
\newblock Sorting with recurrent comparison errors.
\newblock In {\em 28th International Symposium on Algorithms and Computation
  (ISAAC 2017)}. Schloss Dagstuhl-Leibniz-Zentrum fuer Informatik, 2017.

\bibitem{geissmann2018optimal}
Barbara Geissmann, Stefano Leucci, Chih-Hung Liu, and Paolo Penna.
\newblock Optimal sorting with persistent comparison errors.
\newblock In {\em 27th Annual European Symposium on Algorithms (ESA 2019)},
  volume 144, page~49. Schloss Dagstuhl-Leibniz-Zentrum f{\"u}r Informatik,
  2019.

\bibitem{geissmann2020optimal}
Barbara Geissmann, Stefano Leucci, Chih-Hung Liu, and Paolo Penna.
\newblock Optimal dislocation with persistent errors in subquadratic time.
\newblock {\em Theory of Computing Systems}, 64(3):508--521, 2020.

\bibitem{ghoshdastidar2019foundations}
Debarghya Ghoshdastidar, Micha{\"e}l Perrot, and Ulrike von Luxburg.
\newblock Foundations of comparison-based hierarchical clustering.
\newblock In {\em Advances in Neural Information Processing Systems}, pages
  7454--7464, 2019.

\bibitem{girdhar2012efficient}
Yogesh Girdhar and Gregory Dudek.
\newblock Efficient on-line data summarization using extremum summaries.
\newblock In {\em 2012 IEEE International Conference on Robotics and
  Automation}, pages 3490--3496. IEEE, 2012.

\bibitem{gokhale2014corleone}
Chaitanya Gokhale, Sanjib Das, AnHai Doan, Jeffrey~F Naughton, Narasimhan
  Rampalli, Jude Shavlik, and Xiaojin Zhu.
\newblock Corleone: Hands-off crowdsourcing for entity matching.
\newblock In {\em Proceedings of the 2014 ACM SIGMOD international conference
  on Management of data}, pages 601--612, 2014.

\bibitem{gonzalez1985clustering}
Teofilo~F Gonzalez.
\newblock Clustering to minimize the maximum intercluster distance.
\newblock {\em Theoretical Computer Science}, 38:293--306, 1985.

\bibitem{kasper}
Kasper Green~Larsen, Michael Mitzenmacher, and Charalampos Tsourakakis.
\newblock Clustering with a faulty oracle.
\newblock In {\em Proceedings of The Web Conference 2020}, WWW ’20, page
  2831–2834, New York, NY, USA, 2020. Association for Computing Machinery.

\bibitem{griffin2007caltech}
Gregory Griffin, Alex Holub, and Pietro Perona.
\newblock Caltech-256 object category dataset.
\newblock 2007.

\bibitem{guo2012so}
Stephen Guo, Aditya Parameswaran, and Hector Garcia-Molina.
\newblock So who won? dynamic max discovery with the crowd.
\newblock In {\em Proceedings of the 2012 ACM SIGMOD International Conference
  on Management of Data}, pages 385--396, 2012.

\bibitem{he2016ups}
Ruining He and Julian McAuley.
\newblock Ups and downs: Modeling the visual evolution of fashion trends with
  one-class collaborative filtering.
\newblock In {\em proceedings of the 25th international conference on world
  wide web}, pages 507--517, 2016.

\bibitem{hopkins2020noise}
Max Hopkins, Daniel Kane, Shachar Lovett, and Gaurav Mahajan.
\newblock Noise-tolerant, reliable active classification with comparison
  queries.
\newblock {\em arXiv preprint arXiv:2001.05497}, 2020.

\bibitem{huleihel2019same}
Wasim Huleihel, Arya Mazumdar, Muriel M{\'e}dard, and Soumyabrata Pal.
\newblock Same-cluster querying for overlapping clusters.
\newblock In {\em Advances in Neural Information Processing Systems}, pages
  10485--10495, 2019.

\bibitem{ilvento2019metric}
Christina Ilvento.
\newblock Metric learning for individual fairness.
\newblock {\em arXiv preprint arXiv:1906.00250}, 2019.

\bibitem{kazemi2018comparison}
Ehsan Kazemi, Lin Chen, Sanjoy Dasgupta, and Amin Karbasi.
\newblock Comparison based learning from weak oracles.
\newblock {\em arXiv preprint arXiv:1802.06942}, 2018.

\bibitem{kim2017relaxed}
Taewan Kim and Joydeep Ghosh.
\newblock Relaxed oracles for semi-supervised clustering.
\newblock {\em arXiv preprint arXiv:1711.07433}, 2017.

\bibitem{kim2017semi}
Taewan Kim and Joydeep Ghosh.
\newblock Semi-supervised active clustering with weak oracles.
\newblock {\em arXiv preprint arXiv:1709.03202}, 2017.

\bibitem{klein2011tolerant}
Rolf Klein, Rainer Penninger, Christian Sohler, and David~P Woodruff.
\newblock Tolerant algorithms.
\newblock In {\em European Symposium on Algorithms}, pages 736--747. Springer,
  2011.

\bibitem{kleindessner2019fair}
Matth{\"a}us Kleindessner, Pranjal Awasthi, and Jamie Morgenstern.
\newblock Fair k-center clustering for data summarization.
\newblock In {\em International Conference on Machine Learning}, pages
  3448--3457, 2019.

\bibitem{kou2017crowdsourced}
Ngai~Meng Kou, Yan Li, Hao Wang, Leong~Hou U, and Zhiguo Gong.
\newblock Crowdsourced top-k queries by confidence-aware pairwise judgments.
\newblock In {\em Proceedings of the 2017 ACM International Conference on
  Management of Data}, pages 1415--1430, 2017.

\bibitem{mason2019learning}
Blake Mason, Ardhendu Tripathy, and Robert Nowak.
\newblock Learning nearest neighbor graphs from noisy distance samples.
\newblock In {\em Advances in Neural Information Processing Systems}, pages
  9586--9596, 2019.

\bibitem{mazumdar2017clustering}
Arya Mazumdar and Barna Saha.
\newblock Clustering with noisy queries.
\newblock In {\em Advances in Neural Information Processing Systems}, pages
  5788--5799, 2017.

\bibitem{mazumdar2017query}
Arya Mazumdar and Barna Saha.
\newblock Query complexity of clustering with side information.
\newblock In {\em Advances in Neural Information Processing Systems}, pages
  4682--4693, 2017.

\bibitem{papadakis2016comparative}
George Papadakis, Jonathan Svirsky, Avigdor Gal, and Themis Palpanas.
\newblock Comparative analysis of approximate blocking techniques for entity
  resolution.
\newblock {\em Proceedings of the VLDB Endowment}, 9(9):684--695, 2016.

\bibitem{polychronopoulos2013human}
Vassilis Polychronopoulos, Luca De~Alfaro, James Davis, Hector Garcia-Molina,
  and Neoklis Polyzotis.
\newblock Human-powered top-k lists.
\newblock In {\em WebDB}, pages 25--30, 2013.

\bibitem{prelec2017solution}
Dra{\v{z}}en Prelec, H~Sebastian Seung, and John McCoy.
\newblock A solution to the single-question crowd wisdom problem.
\newblock {\em Nature}, 541(7638):532--535, 2017.

\bibitem{sibson1973slink}
Robin Sibson.
\newblock Slink: an optimally efficient algorithm for the single-link cluster
  method.
\newblock {\em The computer journal}, 16(1):30--34, 1973.

\bibitem{tamuz2011adaptively}
Omer Tamuz, Ce~Liu, Serge Belongie, Ohad Shamir, and Adam~Tauman Kalai.
\newblock Adaptively learning the crowd kernel.
\newblock In {\em Proceedings of the 28th International Conference on
  International Conference on Machine Learning}, pages 673--680, 2011.

\bibitem{ukkonen2017crowdsourced}
Antti Ukkonen.
\newblock Crowdsourced correlation clustering with relative distance
  comparisons.
\newblock In {\em 2017 IEEE International Conference on Data Mining (ICDM)},
  pages 1117--1122. IEEE, 2017.

\bibitem{modal}
\url{https://modal-python.readthedocs.io/en/latest/}.
\newblock modal library.

\bibitem{scikit}
\url{https://scikit-learn.org/stable/}.
\newblock Scikit-learn.

\bibitem{vazirani2013approximation}
Vijay~V Vazirani.
\newblock {\em Approximation algorithms}.
\newblock Springer Science \& Business Media, 2013.

\bibitem{venetis2012max}
Petros Venetis, Hector Garcia-Molina, Kerui Huang, and Neoklis Polyzotis.
\newblock Max algorithms in crowdsourcing environments.
\newblock In {\em Proceedings of the 21st international conference on World
  Wide Web}, pages 989--998, 2012.

\bibitem{verdugo2020skyline}
Victor Verdugo.
\newblock Skyline computation with noisy comparisons.
\newblock In {\em Combinatorial Algorithms: 31st International Workshop, IWOCA
  2020, Bordeaux, France, June 8--10, 2020, Proceedings}, page 289. Springer.

\bibitem{verroios2015entity}
Vasilis Verroios and Hector Garcia-Molina.
\newblock Entity resolution with crowd errors.
\newblock In {\em 2015 IEEE 31st International Conference on Data Engineering},
  pages 219--230. IEEE, 2015.

\bibitem{vesdapunt2014crowdsourcing}
Norases Vesdapunt, Kedar Bellare, and Nilesh Dalvi.
\newblock Crowdsourcing algorithms for entity resolution.
\newblock {\em Proceedings of the VLDB Endowment}, 7(12):1071--1082, 2014.

\bibitem{vinayak2016crowdsourced}
Ramya~Korlakai Vinayak and Babak Hassibi.
\newblock Crowdsourced clustering: Querying edges vs triangles.
\newblock In {\em Advances in Neural Information Processing Systems}, pages
  1316--1324, 2016.

\bibitem{wang2012crowder}
Jiannan Wang, Tim Kraska, Michael~J Franklin, and Jianhua Feng.
\newblock Crowder: Crowdsourcing entity resolution.
\newblock {\em Proceedings of the VLDB Endowment}, 5(11), 2012.

\bibitem{williamson2011design}
David~P Williamson and David~B Shmoys.
\newblock {\em The design of approximation algorithms}.
\newblock Cambridge university press, 2011.

\bibitem{zhang2018taxogen}
Chao Zhang, Fangbo Tao, Xiusi Chen, Jiaming Shen, Meng Jiang, Brian Sadler,
  Michelle Vanni, and Jiawei Han.
\newblock Taxogen: Unsupervised topic taxonomy construction by adaptive term
  embedding and clustering.
\newblock In {\em Proceedings of the 24th ACM SIGKDD International Conference
  on Knowledge Discovery \& Data Mining}, pages 2701--2709, 2018.

\end{thebibliography}


\newpage
\onecolumn

\section{Finding Maximum}\label{app:max}
\begin{lemma}(Hoeffding's Inequality) If $X_1, X_2, \cdots, X_n$ are independent random variables with $a_i \leq X_i \leq b_i$ for all $i \in [n]$, then 
\[ \Pr\left[ \left| \sum_{i} X_i - \E[X_i] \right| \geq n\epsilon \right] \leq 2\exp\left( -\frac{2n^2\epsilon^2}{\sum_{i} (b_i - a_i)^2} \right)\]
\label{hoeffding}
\end{lemma}

\subsection{Adversarial Noise}

Let the maximum value among $V$ be denoted by $\vmax$ and the set of records for which the oracle answer  \textit{can} be incorrect is given by 
$$C =\{ u \mid u\in V, u\geq \frac{\vmax}{1+\mu} \}$$ 

\begin{claim}\label{cl:partition_q}
For any partition $V_i$, $\textsc{Tournament}(V_i)$ uses at most $2|V_i|$ oracle queries.
\end{claim}
\begin{proof}
Consider the $i$th round in $\textsc{Tournament}$. We can observe that the number of remaining values is at most $\frac{|V_i|}{2^i}$. So, we make $\frac{|V_i|}{2^{i+1}}$ many oracle queries in this round. Total number of oracle queries made is 
$$
\sum_{i=0}^{\log n} \frac{|V_i|}{2^{i+1}} \leq 2|V_i|
$$ 
\end{proof}

\begin{lemma}\label{lem:count_appendix}
Given a set of values $S$, $\textsc{Count-Max}(S)$ returns a $(1+\mu)^2$ approximation of maximum value of $S$ using $O(|S|^2)$ oracle queries.
\end{lemma}
\begin{proof}
Let $\vmax = \max\{ x \in S \} $.
Consider a value $w \in S$ such that $w < \frac{\vmax}{(1+\mu)^2}$. We compare the \Count~ values for $\vmax$ and $w$ given by, $ \Count(\vmax, S) = \sum_{x \in S} 1\{ \oracle(\vmax, x) == \No{} \}$
and $\Count(w, S) = \sum_{x \in S} 1\{ \oracle(w, x) == \No{} \}$. We argue that $w$ can never be returned by Algorithm \ref{alg:maxCount}, i.e., $\Count(w, S) < \Count(\vmax, S)$.

 \begin{align*}
\Count(\vmax, S) = \sum_{x \in S} 1\{ \oracle(\vmax, x) == \No{} \}  &\geq \sum_{x\in S \setminus \{ \vmax \} } 1 \{  x< \vmax/(1+\mu) \} \\
&= {1}\{ \oracle(\vmax, w) == \No\} \ \ +  \sum_{x\in S\setminus \{\vmax,w\} } 1 \{ x< \vmax/(1+\mu) \} \\
&= 1 + \sum_{x\in S\setminus \{\vmax,w\}  } 1\{ x< \vmax/(1+\mu)\}\\
  \Count(w, S) = \sum_{y \in S} 1\{ \oracle(w, y) == \No{} \} &= \sum_{y\in S\setminus \{w,\vmax\} } {1}\{ \oracle(w, y)== \No \} \\
  &\leq \sum_{y\in S\setminus \{w,\vmax\} } 1 \{  y \leq (1+\mu)w \}\\
&\leq \sum_{y\in S\setminus \{w,\vmax\}  } 1 \{  y \leq {\vmax}/{(1+\mu)} \}
\end{align*}

Combining the two, we have :
  \[\Count(\vmax, S) > \Count(w, S) \]
This shows that the \Count of $\vmax$ is strictly greater than the count of any point $w$ with  
$w < \frac {\vmax}{(1+\mu)^2}$. Therefore, our algorithm would have output $\vmax$ instead of $w$. For calculating the \Count~for all values in $S$, we make at most $|S|^2$ oracle queries as we compare every value with every other value. Finally, we output the maximum value as the value with highest \Count. Hence, the claim.
\end{proof}

\begin{lemma}[Lemma~\ref{lem:tournament} restated]
Suppose $v_{max}$ is the maximum value among the set of records $V$. Algorithm~\ref{alg:maxTournament} outputs a value $u_{max}$ such that $u_{max}\geq \frac{v_{max}}{(1+\mu)^{2 \log_\lambda n}}$ using  $O(n\lambda)$ oracle queries.
\label{lem:tournament_appendix}
\end{lemma}
\begin{proof}
From Lemma~\ref{lem:count_appendix}, we have that we lose a factor of $(1+\mu)^2$ in each level of the tournament tree, we have that after $\log_\lambda n$ levels, the final output will have an approximation guarantee of $(1+\mu)^{2\log_\lambda n}$. The total number of queries used is given by : $\sum_{i=0}^{\log_\lambda n} \frac{|V_i|}{\lambda} \lambda^2 = O(n\lambda)$ where $V_i$ is the number of records at level $i$. 
\end{proof}

\begin{lemma}\label{lem:maxSampling_appendix}
Suppose $|C| > \sqrt{n}/2$. Let $\widetilde{V}$ denote a set of $2\sqrt{n}\log(2/\delta)$ samples obtained by uniform sampling with replacement from $V$. Then, $\widetilde{V}$ contains a $(1+\mu)$ approximation of the maximum value $\vmax$, with probability $1-\delta/2$.
\end{lemma}
\begin{proof}
Consider the first step where we use a uniformly random sample $\widetilde{V}$ of $\sqrt{n}t = 2\sqrt{n}\log(2/\delta)$ values from $V$ (obtained by sampling with replacement).  Given $|C| \geq \frac{\sqrt{n}}{2}$, probability that $\widetilde{V}$ contains a value from $C$ is given by $$
\Pr[\widetilde{V} \cap C \neq \phi] = 1-\left(1- \frac{|C|}{n} \right)^{|\widetilde{V}|} 
> 1-\left(1-\frac{1}{2\sqrt{n}}\right)^{2\sqrt{n}\log(2/\delta)} > 1-\delta/2$$

So, with probability $1- \delta/2$, there exists a value $ u \in C \cap \widetilde{V}$. Hence, the claim.  
\end{proof}

\begin{lemma}\label{lem:maxTournament}
Suppose the partition $V_i$ contains the maximum value $\vmax$ of $V$. If $|C| \leq {\sqrt{n}}/2$, then, $\textsc{Tournament}(V_i)$ returns the $\vmax$ with probability $1/2$.
\end{lemma}
\begin{proof}
Algorithm~\ref{alg:max_adv} uses a modified tournament tree that partitions the set $V$ into $l = \sqrt{n}$ parts of size $\frac{n}{l} = \sqrt{n}$ each and identifies a maximum $p_i$ from each partition $V_i$ using Algorithm~\ref{alg:maxTournament}. If $\vmax \in V_i$, then, 
\[ \E[|C \cap V_i|] = \frac{|C|}{l} = \frac{\sqrt{n}}{2\sqrt{n}} = \frac{1}{2} \]
Using Markov's inequality, the probability that $V_i$ contains a value from $C$ is given by : 
\[ \Pr[|C \cap V_i| \geq 1] \leq \E[|C \cap V_i|] \leq \frac{1}{2}\]
 Therefore, with at least a probability of $\frac{1}{2}$, $\vmax$ will never be compared with any point from $C$ in the partition $V_i$ containing $\vmax$. Hence, $\vmax$ is returned by $\textsc{Tournament}(V_i)$ with probability $1/2$.

\end{proof}


\begin{lemma}[Lemma~\ref{lem:hierquality} restated]
\begin{enumerate}
 \item If $|C| >  {\sqrt{n}}/{2}$, then there exists a value $ v_j \in \widetilde{V}$ satisfying
    $v_j \geq {\vmax}/{(1+\mu)}$ with a probability of $1-\delta/2$.
 \item Suppose $|C| \leq  {\sqrt{n}}/{2}$. Then, $T$ contains  $\vmax$ with a probability at least $1-\delta/2$.
\end{enumerate}
\label{lem:hierquality_appendix}
\end{lemma}
\begin{proof}
Claim (1) follows from Lemma~\ref{lem:maxSampling_appendix}. \\ 

In every iteration $i \leq t$ of Algorithm~\ref{alg:max_adv}, we have that $\vmax \in T_i$ with probability $\frac{1}{2}$ (Using Lemma~\ref{lem:maxTournament}). To increase the success probability, we run this procedure $t$ times and obtain all the outputs. Among the $t = 2\log(2/\delta)$ runs of Algorithm~\ref{alg:maxTournament}, we have that $\vmax$ is never compared with any value of $C$ in atleast one of the iterations with a probability at least $$1-\left( 1-1/2 \right)^{2\log(2/\delta) } \geq 1-\frac{\delta}{2}$$
Hence, $T = \cup_i T_i$ contains $\vmax$ with a probability $1-\frac{\delta}{2}$.
\end{proof}

\begin{theorem}[Theorem~\ref{thm:max_adv} restated]\label{thm:max_adv_appendix}
Given a set of values $V$, Algorithm~\ref{alg:max_adv} returns a $(1+\mu)^3$ approximation of maximum value with probability $1-\delta$ using $O(n \log^2(1/\delta))$ oracle queries.
\end{theorem}
\begin{proof}
In Algorithm~\ref{alg:max_adv}, we first identify an approximate maximum value using $\text{Sampling}$. If $|C| \geq \frac{\sqrt n}{2}$, then, from Lemma~\ref{lem:maxSampling_appendix}, we have that the value returned is a $(1+\mu)$ approximation of the maximum value of $V$. Othwerwise,  from Lemma~\ref{lem:hierquality_appendix}, $T$ contains $\vmax$ with a probability $1-\delta/2$. As we use $\textsc{Count-Max}$ on the set $\widetilde{V} \cup T$, we know that the value returned, i.e.,  $\umax$ is a $(1+\mu)^2$ of the maximum among values from $\widetilde{V} \cup T$. Therefore, $\umax \geq \frac{\vmax}{(1+\mu)^3}$. Using union bound, the total probability of failure is $\delta$.

For query complexity, Algorithm~\ref{alg:max_noise} obtains a set $\widetilde{V}$ of $\sqrt{n} t$ sample values. Along with the set $T$ obtained (where $|T| = \frac{nt}{l}$), we use $\textsc{Count-Max}$ on $\widetilde{V} \cup T$ to output the maximum $\umax$. 
This step requires $O(|\widetilde{V} \cup T |^2 ) = O((\sqrt{n}t + \frac{nt}{l})^2)$ oracle queries. In an iteration $i$, for obtaining $T_i$, we make $O(\sum_{j} |V_j|) = O(n)$ oracle queries (Claim~\ref{cl:partition_q}), and for $t$ iterations, we make $O(n t)$ queries. Using $t = 2\log(2/\delta), l = \sqrt{n}$, in total, we make $O(n t +(\sqrt{n}t + \frac{nt}{l})^2) = O(n \log^2(1/\delta))$ oracle queries. Hence, the theorem.

\end{proof}

\subsection{Probabilistic Noise}
\begin{lemma}
Suppose the maximum value $\umax$ is returned by Algorithm~\ref{alg:maxTournament} with parameters $(V,n)$. Then, 
$\textrm{rank}(\umax, V) = O( \sqrt{n \log(1/\delta)})$ with a probability of $1-\delta$.
\end{lemma}
\begin{proof}
We have for the maximum value $\vmax$, expected count value :
\begin{align*}
    \E[\Count(\vmax, V)] = \sum_{w \in V} \mathbf{1}\{ \oracle(w, \vmax) == w \} = (n-1)(1-p) 
\end{align*}
Using Hoeffding's inequality, with probability $1-\delta/2$ :
\begin{align*}
    \Count(\vmax, V) \ge (n-1)(1-p) - \sqrt{((n-1) \log(2/\delta))/2} 
\end{align*}
Consider a record $u \in V$ with rank at most $5 \sqrt{2n \log(2/\delta)}$. Then, 
\begin{align*}
    \E[\Count(u, V)] &= \sum_{w \in V} \mathbf{1}\{ \oracle(u, \vmax) == w \} = (n-\textrm{rank}(u))(1-p) + (\textrm{rank}(u)-1)p
\end{align*}
Using Hoeffding's inequality, with probability $1-\delta/2$ :
\begin{align*}
    \Count(u, V) &< (n-1)(1-p) - (\textrm{rank}(u) - 1)(1-2p) + \sqrt{0.5 (n-1)\log(2/\delta)}\\
    &< (n-1)(1-p) - (5 \sqrt{2n \log(2/\delta)} - 1)(1-2p) + \sqrt{0.5 (n-1)\log(2/\delta)} \\
    &< \Count(\vmax, V)
\end{align*}
The last inequality is true for a value of $p \leq 0.4$. 
As Algorithm~\ref{alg:maxTournament} returns the record $\umax$ with maximum $\Count$ value, we have that rank$(\umax, V) = O(\sqrt{n \log(1/\delta)})$. Using union bound, for the above conditions to be met, we have the claim.
\end{proof}

To improve the query complexity, we use an early stopping criteria that discards a value $x$ using the $\Count(x, V)$ when it determines that $x$ has no chance of being the maximum. Algorithm~\ref{alg:count_prob} presents the psuedocode for this modified count calculation. We sample $100 \log(n/\delta)$ values randomly, denoted by $S_t$ and compare every non-sampled point with $S_t$. We argue that by doing so, it helps us eliminate the values that are far away from the maximum in the sorted ranking. Using Algorithm~\ref{alg:count_prob}, we compare the $\Count$ scores with respect to $S_t$ of a value $u \in V\setminus S_t$ and if $\Count(u, S_t) \geq 50 \log(n/\delta)$, we make it available for the subsequent iterations. 


\begin{algorithm}[h]
\begin{algorithmic}[1]
\State \textbf{Input} : A set $V$ of $n$ values, failure probability $\delta$.
\State \textbf{Output} : An approximate maximum value of $V$
\State $t \leftarrow 1$
\While{$t < \log(n)$ or $|V| > 100 \log(n/\delta)$}
    \State $S_t$ denote a set of $100 \log(n/\delta)$ values obtained  by sampling uniformly at random from $V$ with replacement.
    \State Set $X \leftarrow \phi$
    \For{$u \in V \setminus S_t$}
    \If{ $\Count(u, S_t) \geq  50 \log(n/\delta)$ }
        \State $X \leftarrow X \cup \{ u \}$
    \EndIf
    \EndFor
    \State $V \leftarrow  X, t \leftarrow t+1$
\EndWhile
\State $\umax \leftarrow \textsc{Count-Max}(V)$
\State \Return $\umax$
\end{algorithmic}
\caption{\textsc{Count-Max-Prob} : Maximum with Probabilistic Noise }\label{alg:count_prob}
\end{algorithm}

As Algorithm~\ref{alg:count_prob} considers each value $u \in V \setminus S_t$ by iteratively comparing it with each value $x \in S_t$ and the error probability is less than $p$, the expected count of $\vmax$ (if it is available) at any iteration $t$ is $(1-p) |S_t|$. Accounting for the deviation around the expected value, we have that the $\Count(\vmax, S_t)$ is at least $50 \log(n/\delta)$ when $p \leq 0.4$\footnote{The constants $50$, $100$ etc. are not optimized and set just to satisfy certain concentration bounds.}. If a particular value $u$ has $\Count(u, S_t) < 50 \log(n/\delta)$ in any iteration, i.e., then it can not be the largest value in $V$ and therefore, we remove it from the set of possible candidates for maximum. Therefore, any value that remains in $V$ after an iteration $t$, must have rank closer to that of $\vmax$. We argue that after every iteration, the number of candidates remaining is at most $1/60$th of the possible candidates.

\begin{lemma}\label{lem:max_opt_appendix}
In an iteration $t$ containing $n_t$ remaining records, using Algorithm~\ref{alg:maxoptimized}, with probability $1-\delta/n$, we discard at least $\frac{59}{60} \cdot n_t$ records.
\end{lemma}
\begin{proof}
Consider an iteration $t$ which has $n_t$ remaining records. Algorithm~\ref{alg:maxoptimized} and a record $u$ with rank $\alpha \cdot n_t$. Now, we have : 
$$\E[\Count(u, S_t)] = ( (1-\alpha)(1-p) + \alpha p )100 \log(n/\delta)$$
For $\alpha = 0$, i.e., we have for maximum value $\vmax$
\[ \E[\Count(\vmax, S_t)] = (1-p) 100 \log(n/\delta)\]
Using  $p \leq 0.4$ and Hoeffding's inequality, with probability $1-\delta/n^2$ , we have :
\begin{align*}
    \Count(\vmax, S_t) \ge (1-p)100 \log(n/\delta) - \sqrt{100} \log(n/\delta) \geq 50 \log(n/\delta)
\end{align*}
For $u$, we calculate the $\Count$ value. Using  $p \leq 0.4$ and Hoeffding's inequality, with probability $1-\delta/n^2$ , we have :
\begin{align*}
    \Count(u, S_t) &<  ( (1-\alpha)(1-p) + \alpha p )100 \log(n/\delta) + \sqrt{100 ((1-\alpha)(1-p) + \alpha p )} \log(n/\delta) \\
                    &< ((1-0.6 \alpha) 100 +  \sqrt{100(1-0.6\alpha)}) \log(n/\delta) < 50 \log(n/\delta) 
\end{align*}
Upon calculation, for $ \alpha > \frac{59}{60}$, we have the above expression. Therefore, using union bound, with probability $1-O(\delta/n)$, all records $u$ with rank at least $\frac{59n_t}{60}$ satisfy :  
$$\Count(u, S_t) < \Count(\vmax, S_t)$$
So, all such values can be removed. Hence, the claim.
\end{proof}

In the previous lemma, we argued that in every iteration, at least $1/60$th fraction is removed and therefore in $\Theta(\log n)$ iterations, the algorithm will terminate. In each iteration, we discard the sampled values $S_t$ to ensure that there is no dependency between the $\Count$ scores, and our guarantees hold.
As we remove at most $O(t \cdot \log(n/\delta)) = O(\log^2(n/\delta))$ sampled points, our final statement of the result is : 
\begin{lemma}\label{lem:max_probappendix}
Query complexity of Algorithm~\ref{alg:maxoptimized} is $O(n \cdot \log^2(n/\delta))$ and $\umax$ satisfies $\textrm{rank}(\umax, V) \leq O(\log^2(n/\delta))$ with probability $1- \delta$.
\end{lemma}
\begin{proof}
From Lemma~\ref{lem:max_opt_appendix}, we have with probability $1-\delta/n$, after iteration $t$, at least $\frac{59 n_t}{60}$ records removed along with the $100 \log(n/\delta)$ records that are sampled. Therefore, we have :
\[ n_{t+1} \leq \frac{n_t}{60} - 100 \log(n/\delta)\]
After $\log(n/\delta)$ iterations, we have $n_{t+1} \leq 1$. As we have removed $\log_{60} n \cdot 100 \log(n/\delta)$ records that were sampled in total, these could records with rank$\leq 100 \log^2(n/\delta)$. So, the rank of $\umax$ output is at most $100 \log^2(n/\delta)$. In an iteration $t$, the number of oracle queries calculating $\Count$ values is $O(n_t \cdot \log(n/\delta))$. In total, Algorithm~\ref{alg:maxoptimized} makes $O(n \log^2(n/\delta))$ oracle queries. Using union bound over $\log(n/\delta)$ iterations, we get a total failure probability of $\delta$.
\end{proof}

\begin{theorem}[Theorem~\ref{thm:max_prob} restated]
There is an algorithm that returns $\umax \in V$ such that $\textrm{rank} (\umax, V) = O(\log^2(n/\delta))$ with probability $1-\delta$ and requires $O(n \log^2(n/\delta))$ oracle queries.
\end{theorem}
\begin{proof}
The proof follows from Lemma~\ref{lem:max_probappendix}
\end{proof}

\section{Farthest and Nearest Neighbor}\label{app:farthest}

\begin{lemma}[Lemma~\ref{lem:pairwise} restated]
Suppose $\max_{v_i \in S} d(u, v_i) \leq \alpha$ and $|S| \geq 6 \log(1/\delta)$. Consider two records $v_i$ and $v_j$ such that $d(u,v_i) < d(u,v_j)-2\alpha$ then  $\FCount(v_i,v_j)\geq 0.3 |S|$  with a probability of $1-\delta$\label{lem:pairwise_appendix}
\end{lemma}
\begin{proof}
Since $d(u,v_i) < d(u,v_j)-2\alpha$, for a point $x\in S$, 
\begin{align*}
d(v_j,x) &\geq d(u,v_j) -d(u,x) \\
&> d(u,v_i) + 2\alpha - d(u,x)\\
&\geq d(v_i,x) -d(u,x) + 2\alpha - d(u,x)\\
&\geq d(v_i,x) + 2\alpha  -2 d(u,x)\\
&\geq d(v_i,x)
\end{align*}
So, $O(v_i,x,v_j,x)$ is $\No$  with a probability $p$.  As $p \leq 0.4$, we have :
\begin{align*}
  \E[\FCount(v_i,v_j)] &= (1-p)|S|\\
  \Pr[\FCount(v_i, v_j) \leq 0.3|S|] &\leq \Pr[\FCount(v_i, v_j) \leq (1-p)|S|/2]
\end{align*}
From Hoeffding's inequality (with binary random variables), we have with a probability $\exp(-\frac{|S|(1-p)^2}{2}) \leq \delta$ (using $|S| \ge  6\log(1/\delta)$, $p < 0.4$) :
$\FCount(v_i,v_j)\leq (1-p)|S|/2$. Therefore, with probability at most $\delta$, we have, $\FCount(v_i, v_j) \leq 0.3|S|$. 
\end{proof}

For the sake of completeness, we restate the $\Count$ definition that is used in Algorithm~\textsc{Count-Max}. For every oracle comparison, we replace it with the \emph{pairwise comparison} query described in Section~\ref{sec:farthest}. Let $u$ be a query point and $S$ denote a set of $\Theta(\log(n/\delta))$ points within a distance of $\alpha$ from $u$. We maintain a $\Count$ score for a given point $v_i \in V$ as :
    \[ \Count(u, v_i, S, V) = \sum_{v_j \in V \setminus \{ v_i \}} 1\{ \textsc{Pairwise-Comp}(u, v_i, v_j, S) == \No{} \}\]

\begin{algorithm}[h]
\begin{algorithmic}[1]
\State \textbf{Input} : A set of points $V$, and query point $u$ and a set $S$.
\State \textbf{Output} : An approximate farthest point from $u$
\For {$v_i \in V$}
\State Calculate $\Count(u, v_i, S, V)$ 
\EndFor
\State $\umax \leftarrow \text{arg max}_{v \in S} \Count(u, v_i, S, V)$
\State \Return $\umax$
\end{algorithmic}
\caption{\textsc{Count-Max}(V) : finds the farthest point by counting in $V$}\label{alg:maxCount_additive}
\end{algorithm}

\begin{lemma}\label{lem:count_additive_appendix}

Given a query vertex $u$ and a set $S$ with $|S| = \Omega(\log(n/\delta))$ such that $\max_{v\in S}d(u,v)\leq \alpha$. Then the farthest identified using Algorithm~\ref{alg:maxCount_additive} (with \textsc{PairwiseComp}), denoted by $\umax$ is within $4 \alpha$ distance from the optimal farthest point, i.e., $d(u, \umax) \geq \max_{v \in V} d(u, v) - 4\alpha$ with a probability of $1-\delta$. Further the query complexity is $O(n^2 \log (n/\delta))$.
\end{lemma}
\begin{proof}
Let $\vmax = \max_{ v \in V } d(u, v)$. Consider a value $w \in V$ such that $d(u, w) < d(u, \vmax) - 4\alpha$. We compare the \Count~ values for $\vmax$ and $w$ given by, $ \Count( u, \vmax, S, V) = \sum_{v_j \in V\setminus \{ \vmax \}} 1\{ \textsc{Pairwise-Comp}(u, \vmax, v_j, S) == \No{} \}$ and $\Count(u, w, S, V) = \sum_{v_j \in V\setminus \{ w \}} 1\{ \textsc{Pairwise-Comp}(u, w, v_j, S) == \No{} \}$. We argue that $w$ can never be returned by Algorithm \ref{alg:maxCount_additive}, i.e., $\Count(u, w, S, V) < \Count( u, \vmax, S, V)$. Using Lemma~\ref{lem:pairwise_appendix} we have : 

 \begin{align*}
\Count( u, \vmax, S, V) &= \sum_{v_j \in V \setminus \{ \vmax \}} 1\{ \textsc{Pairwise-Comp}(u, \vmax, v_j, S) == \No{} \} \\
&\geq \sum_{v_j \in V \setminus \{ \vmax \} } 1 \{ d(u, v_j) < d(u, \vmax) - 2\alpha \} \\
&= {1}\{ d(u, w) < d(u, \vmax) - 2\alpha\} \ \ +  \sum_{v_j \in V \setminus \{ \vmax, w \} } 1 \{ d(u, v_j) < d(u, \vmax) - 2\alpha \}  \\
&= 1 + \sum_{v_j \in V \setminus \{ \vmax, w\} } 1 \{ d(u, v_j) < d(u, \vmax) - 2\alpha \} \\
\Count( u, w, S, V) &= \sum_{v_j \in V\setminus \{ w \}} 1\{ \textsc{Pairwise-Comp}(u, w, v_j, S) == \No{} \} \\
&\leq \sum_{v_j \in V \setminus \{ w \} } 1 \{ d(u, v_j) < d(u, w) + 2\alpha \} \\
&\leq \sum_{v_j \in V \setminus \{ w, \vmax \} } 1 \{ d(u, v_j) < d(u, \vmax) - 2\alpha \} \\
\end{align*}

Combining the two, we have :
  \[\Count(u, \vmax, S, V) > \Count(u, w, S, V) \]
This shows that the \Count of $\vmax$ is strictly greater than the count of any point $w$ when  
$d(u, w) < d(u, \vmax) - 4\alpha$. Therefore, our algorithm would have output $\vmax$ instead of $w$. For calculating the \Count~for all points in $V$, we make at most $|V|^2 \cdot |S|$ oracle queries as we compare every point with every other point using Algorithm~\ref{alg:pairwise}. Finally, we output the point $\umax$ as the value with highest \Count. From Lemma~\ref{lem:pairwise_appendix}, when $|S| = \Omega(\log(n/\delta))$, the answer to any pairwise query is correct with a failure probability of $\delta/n^2$. As there are $n^2$ pairwise comparisons, and each with failure probability of $\delta/n^2$, from union bound, we have that that the total failure probability is $\delta$. Hence, the claim.
\end{proof}

\begin{algorithm}[H]
\begin{algorithmic}[1]
\State \textbf{Input} : Set of values $V$, Degree $\lambda$, query point $u$ and a set $S$.
\State \textbf{Output} : An approximate farthest point from $u$
\State Construct a balanced $\lambda$-ary tree $\mathcal T$ with $|V|$ nodes as leaves.
\State Let $\pi_V$ be a random permutation of $V$ assigned to leaves of $\mathcal T$ 
\For{$i=1$ to $\log_\lambda |V|$}
\For{internal node $w$ at level $\log_\lambda |V|-i$}
\State Let  $U$ denote the children of $w$.
\State Set the internal node $w$ to $\textsc{Count-Max}(u, S, U)$
\EndFor
\EndFor
\State $\umax \leftarrow $ point at root of $\mathcal T$
\State \Return $\umax$
\end{algorithmic}
\caption{\textsc{Tournament} : finds the farthest point using a tournament tree}\label{alg:maxTournament_additive}
\end{algorithm}

\begin{algorithm}[!ht]
\begin{algorithmic}[1]
\State \textbf{Input} : Set of values $V$, number of partitions $l$, query point $u$ and a set $S$.
\State \textbf{Output} : A set of farthest points from each partition

\State Randomly partition $V$ into $l$ equal parts $V_1, V_2, \cdots V_{l}$ 
\For{$i = 1$ to $l$}
\State $p_i\leftarrow \textsc{Tournament}(u, S, V_i,2)$
\State $T\leftarrow T\cup\{p_i\}$
\EndFor
\State \Return $T$
\end{algorithmic}
\caption{\textsc{Tournament-Partition} \label{alg:max_noise_additive}}
\end{algorithm}
\begin{algorithm}[h]
\begin{algorithmic}[1]
\State \textbf{Input} : Set of values $V$, number of iterations $t$, query point $u$ and a set $S$.
\State \textbf{Output} : An approximate farthest point $\umax$
\State $i\leftarrow 1,  T \leftarrow \phi$ 
\State Let $\widetilde{V}$ denote a sample of size $\sqrt{n}t$ selected uniformly at random (with replacement) from $V$.
\For{$i\leq t$}
\State $T_i \leftarrow \textsc{Tournament-Partition}(u, S, V, l)$
\State $T\leftarrow T \cup T_i$
\EndFor
\State  $\umax  \leftarrow \textsc{Count-Max}(u, S, \widetilde{V} \cup T)$
\State \Return $\umax$
\end{algorithmic}
\caption{\textsc{Max-Prob} : {Maximum with Probabilistic Noise} \label{alg:max_prob_additive}}
\end{algorithm}
Let the farthest point from query point $u$ among $V$ be denoted by $\vmax$ and the set of records for which the oracle answer \textit{can} be incorrect is given by 
$$C =\{ w \mid w \in V, d(u, w) \geq  d(u, \vmax) - 2\alpha \}$$ 
\begin{lemma}
\begin{enumerate}
 \item If $|C| >  {\sqrt{n}}/{2}$, then there exists a value $ v_j \in \widetilde{V}$ satisfying
    $d(u, v_j) \geq d(u, {\vmax})- 2\alpha$ with a probability of $1-\delta/2$.
 \item Suppose $|C| \leq  {\sqrt{n}}/{2}$. Then, $T$ contains $\vmax$ with a probability at least $1-\delta/2$.
\end{enumerate}
\label{lem:hierquality_additive_appendix}
\end{lemma}
\begin{proof}
The proof is similar to Lemma~\ref{lem:hierquality_appendix}.
\end{proof}

\begin{theorem}[Theorem~\ref{thm:farthest_nearest_prob} restated]\label{thm:farthest_nearest_prob_appendix}
Given a query vertex $u$ and a set $S$ with $|S| = \Omega(\log(n/\delta))$ such that $\max_{v\in S}d(u,v)\leq \alpha$ then the farthest identified using Algorithm~\ref{alg:max_adv} (with \textsc{PairwiseComp}), denoted by $\umax$ is within $6 \alpha$ distance from the optimal farthest point, i.e., $d(u, \umax) \geq \max_{v \in V} d(u, v) - 6\alpha$ with a probability of $1-\delta$. Further the query complexity is $O(n\log^3(n/\delta))$.
\end{theorem}
\begin{proof}
The proof is similar to Theorem~\ref{thm:max_adv_appendix}. In Algorithm~\ref{alg:max_prob_additive}, we first identify an approximate maximum value using $\text{Sampling}$. If $|C| \geq \frac{\sqrt n}{2}$, then, from Lemma~\ref{lem:hierquality_additive_appendix}, we have that the value returned is a $2\alpha$ additive approximation of the maximum value of $V$. Otherwise,  from Lemma~\ref{lem:hierquality_additive_appendix}, $T$ contains $\vmax$ with a probability $1-\delta/2$. As we use $\textsc{Count-Max}$ on the set $\widetilde{V} \cup T$, we know that the value returned, i.e.,  $\umax$ is a $4\alpha$ of the maximum among values from $\widetilde{V} \cup T$. Therefore, $d(u, \umax) \geq d(u, \vmax) - 6\alpha$. Using union bound over $n \cdot t$ comparisons, the total probability of failure is $\delta$.

For query complexity, Algorithm~\ref{alg:max_noise_additive} obtains a set $\widetilde{V}$ of $\sqrt{n} t$ sample values. Along with the set $T$ obtained (where $|T| = \frac{nt}{l}$), we use $\textsc{Count-Max}$ on $\widetilde{V} \cup T$ to output the maximum $\umax$. 
This step requires $O(|\widetilde{V} \cup T |^2 |S|) = O((\sqrt{n}t + \frac{nt}{l})^2 \log(n/\delta))$ oracle queries. In an iteration $i$, for obtaining $T_i$, we make $O(\sum_{j} |V_j|\log(n/\delta)) = O(n\log(n/\delta))$ oracle queries (Claim~\ref{cl:partition_q}), and for $t$ iterations, we make $O(n t\log(n/\delta))$ queries. Using $t = 2\log(2n/\delta), l = \sqrt{n}$, in total, we make $O(n t\log(n/\delta) +(\sqrt{n}t + \frac{nt}{l})^2\log(n/\delta)) = O(n \log^3(n/\delta))$ oracle queries. Hence, the theorem.
\end{proof}


\section{$k$-center : Adversarial Noise}\label{app:k-center}

 \begin{lemma}
 Suppose in an iteration $t$ of Greedy algorithm, centers are given by $S_t$ and we reassign points using \Assign~which is a $\beta$-approximation to the correct assignment. In iteration $t+1$, using this assignment, if we obtain an $\alpha$-approximate farthest point using \ApproxFarthest, then, after $k$ iterations, Greedy algorithm obtains a $2\alpha\beta^2$-approximation for the $k$-center objective. 
 \label{lem:alphabetalem}
\end{lemma}
 \begin{proof}
 Consider an optimum clustering $C^*$ with centers $u_1, u_2,.., u_k$ respectively: $C^*(u_1), C^*(u_2), \cdots,C^*(u_k)$. Let the centers obtained by Algorithm~\ref{alg:greedykcenter_adv} be denoted by $S$. If $|S \cap C^*(u_i)| = 1$ for all $i$, then, for some point $x \in C^*(u_i)$ assigned to $s_j \in S$ by Algorithm \Assign, we have 
 $$ d(x, S \cap C^*(u_i) ) \leq d(x, u_i ) + d(u_i, S \cap C^*(u_i) ) \leq 2OPT$$ 
 \[ \Longrightarrow d(x, s_j) \leq \beta \ \text{min}_{s_k \in S} \ d(x, s_k) \leq \beta \ d(x, S\cap C^*(u_i)) \leq 2\beta OPT\]

 Therefore, every point in $V$ is at a distance of at most $2 \beta OPT$ from a center assigned in $S$. \\ \\
 Suppose for some $j$ we have $|S \cap C^*(u_j)| \geq 2$. Let $s_1, s_2 \in S \cap C^*(u_j)$ and $s_2$ appeared after $s_1$ in iteration $t+1$. As $s_1 \in S_t$, we have $\min_{w \in S_t} d(w, s_2) \leq d(s_1, s_2)$. In iteration $t$, we know that the farthest point $s_2$  is an $\alpha$-approximation of the farthest point (say $f_t$). Moreover, suppose $s_2$ assigned to cluster with center $s_k$ in iteration $t$ that is a $\beta$-approximation of it's true center. 
 Therefore, $$\frac{1}{\alpha}\ {\min_{w \in S_t} d(w, f_t)} \leq d(s_k, s_2) \leq \beta \min_{w \in S_t} d(w, s_2) \leq \beta d(s_1, s_2)$$ Because $s_1$ and $s_2$ are in the same optimum cluster, from triangle inequality we have $d(s_1, s_2) \leq 2OPT$. Combining all the above we get $\min_{w \in S_t}d(w, f_t) \leq 2\alpha \beta OPT$ which means that farthest point of iteration $t$ is at a distance of $2\alpha \beta OPT$ from $S_t$. In the subsequent iterations, the distance of any point to the final set of centers, given by $S$ only gets smaller. Hence,
 $$\max_{v} \min_{w \in S} d(v,w) \leq \max_v \min_{w \in S_t} d(v,w) = \min_{w \in S_t} d(f_t,w) \leq 2\alpha \beta OPT$$ 
 However, when we output the final clusters and centers, the farthest point after $k$-iterations (say $f_k$) could be assigned to center $v_j \in S$ that is a $\beta$-approximation of the distance to true center. 
 \[ d(f_k, v_j) \leq \beta \ \text{ min}_{w \in S} \ d(f_k, w) \leq 2\alpha \beta^2 \ OPT \]
 Therefore, every point is assigned to a cluster with distance at most $2\alpha \beta^2 \ OPT$. Hence the claim.
 \end{proof}

\begin{lemma}
Given a set $S$ of centers, Algorithm \textsc{Assign} assigns a point $u$ to a cluster $s_j \in S$ such that $d(u, s_j) \leq (1+\mu)^2 \min_{s_t \in S}\{d(u ,s_t)\}$ using $O(nk)$ queries.
\label{lem:assign_adv}
\end{lemma}
\begin{proof}
The proof is essentially the same as Lemma~\ref{lem:count_appendix} and uses $\MCount$ instead of $\Count$.
\end{proof}

\begin{lemma}\label{lem:far_adv}
Given a set of centers $S$, Algorithm~\ref{alg:max_adv} identifies a point $v_j$ with probability $1-\delta/k$, such that $$
\min_{s_j \in S} d(v_j, s_j)  \ge \max_{v_t \in V} \min_{s_t \in S} \frac{\ d(v_t, s_t)}{(1+\mu)^5} $$ 
\end{lemma}
\begin{proof}

Suppose $v_t$ is the farthest point assigned to center $s_t \in S$. Let $v_j$, assigned to $s_j \in S$ be the point returned by Algorithm~\ref{alg:max_adv}. From Theorem~\ref{thm:max_adv_appendix}, we have :


\begin{align*}
 d(v_j,s_j) &\geq \frac{\max_{v_i \in V}  d(v_i, s_i) }{(1+\mu)^3}\\
&\ge \frac{d(v_t, s_t)}{(1+\mu)^3} \geq \frac{\text{min}_{s'_t \in S} \ d(v_t, s'_t)}{(1+\mu)^3}
\end{align*}


Due to error in assignment, using Lemma~\ref{lem:assign_adv}
\[d(v_j,s_j)\leq (1+\mu)^2 \min_{s'_j \in S} d(v_j, s'_j) \]

Combining the above equations we have 
$$\min_{s'_j \in S} d(v_j, s'_j)  \ge \frac{\text{min}_{s'_t \in S} \ d(v_t, s'_t)}{(1+\mu)^5} $$


For \ApproxFarthest, we use $l = \sqrt{n}$ and $t = \log(2k/\delta)$ and $\widetilde{V} = \sqrt{n} t$. So, following the proof in Theorem~\ref{thm:max_adv}, we succeed with probability $1-\delta/k$. Hence, the lemma.
\end{proof}

{
\begin{lemma}
Given a current set of centers $S$,
\begin{enumerate}
    \item \textsc{Assign} assigns a point $u$  to a cluster $C(s_i)$ such that $d(u,s_i) \leq (1+\mu)^2 \min_{s_j \in S}\{d(u,s_j)\}$ {using $O(nk)$ oracle queries additionally}.
    \item \textsc{Approx-Farthest} identifies a point $w$ in cluster $C(s_i)$ such that $\text{ min}_{s_j \in S} \ d(w, s_j) \geq \max_{v_t \in V} \min_{s_t \in S} {\ d(v_t, s_t)}/{(1+\mu)^5}$ with probability $1-\frac{\delta}{k}$ {using $O(n \log^2(k/\delta))$  oracle queries} .
\end{enumerate}
\label{lem:maxAssign_appendix}
\end{lemma}}
\begin{proof}
(1) From Lemma~\ref{lem:assign_adv}, we have the claim.  We assign a point to a cluster based on the scores the cluster center received in comparison to other centers. Except for the newly created center, we have previously queried every center with every other center. Therefore, number of \textit{new}  oracle queries made for every point is $O(k)$; that gives us a total of $O(nk)$ additional new queries used by \textsc{Assign}.

(2) From Lemma~\ref{lem:far_adv}, we have that $\text{ min}_{s_j \in S} \ d(w, s_j) \geq \max_{v_t \in V} \min_{s_t \in S} \frac{\ d(v_t, s_t)}{(1+\mu)^5}$ with probability $1-\delta/k$. As the total number of queries made by Algorithm~\ref{alg:max_adv} is $O(nt + (\frac{nt}{l} + \sqrt{n}t)^2)$. For \ApproxFarthest, we use $l = \sqrt{n}$ and $t = \log(2k/\delta)$ and $\widetilde{V} = \sqrt{n} t$, therefore, the query complexity is $O(n \log^2(k/\delta))$.
\end{proof}

\begin{theorem}[Theorem~\ref{thm:main_adv} restated]
For $\mu < \frac{1}{18}$, 
Algorithm~\ref{alg:greedykcenter_adv} achieves a $(2+O(\mu))$-approximation for the $k$-center objective using $O(nk^2 + n k\cdot \log^2(k/\delta))$ oracle queries with probability $1-\delta$.
\label{thm:main_adv_appendix}
\end{theorem}
\begin{proof}
From the above discussed claim and Lemma~\ref{lem:maxAssign_appendix}, we have that Algorithm~\ref{alg:greedykcenter_adv} achieves a $2(1+\mu)^{9}$ approximation for $k$-center objective. When $\mu < \frac{1}{18}$, we can simplify the approximation factor to $2+18 \mu$, i.e., $2+O(\mu)$.  From Lemma~\ref{lem:maxAssign_appendix}, we have that in each iteration, we succeed with probability $1-\delta/k$. Using union bound, the failure probability is given by $\delta$. For query complexity, as there are $k$ iterations, and in each iteration we use $\Assign$ and $\ApproxFarthest$, using Lemma~\ref{lem:maxAssign_appendix}, we have the theorem.
\end{proof}

\section{$k$-center : Probabilistic Noise}\label{app:k-center-prob}

\subsection{Sampling}


\begin{lemma}\label{lem:sample_appendix}
Consider the sample $\widetilde{V} \subseteq V$ of points obtained by selecting each point with a probability $\frac{450 \log(n/\delta)}{m}$. Then, we have $\frac{400 n \log(n/\delta)}{m}  \leq |\widetilde{V}| \leq \frac{500 n \log(n/\delta)}{m} $ and for every $i \in [k]$, $|C^*(s_i) \cap \widetilde{V}| \geq {400 \log(n/\delta)}$ with probability $1-O(\delta)$ for sufficiently large $\gamma > 0$.
\end{lemma}
\begin{proof} We include every point in $\widetilde{V}$ with a probability $\frac{450 \log(n/\delta)}{m}$ where the size of the smallest cluster is $m$. Using Chernoff bound, with probability $1-O(\delta)$, we have :
\[ \frac{400 n \log(n/\delta)}{m} \leq |\widetilde{V}| \leq \frac{500 n \log(n/\delta)}{m} 
\]
Consider an optimal cluster $C^*(v_i)$ with center $v_i$. As every point is included with probability $\frac{450 \log(n/\delta)}{m}$ :
\[ \E[|C^*(s_i) \cap \widetilde{V}|] = |C^*(s_i)| \cdot \frac{450 \log(n/\delta)}{m} \geq {450 \log(n/\delta)}
\]
Using Chernoff bound, with probability at least $1-\delta/n$, we have
\[ |C^*(s_i) \cap \widetilde{V}| \geq 400 \log(n/\delta)\]
Using union bound for all the $k$ clusters, we have the lemma.
\end{proof}

\subsection{Assignment}
$$\ACount(u,s_i,s_j) = \sum_{x\in R(s_i)} \mathbf{1} \{ \oracle(u,x,u,s_j) == \Yes\}$$

\begin{lemma}\label{lem:assign_prob_appendix}
Consider a point $u$ and $s_j\neq s_i$ such that $d(u,s_i)\leq d(u,s_j)-2\OPT$ and $|R(s_i)|\ge 12\log(n/\delta)$, then,  $\ACount(u,s_i,s_j)\geq 0.3 |R(s_i)|$ with a probability of $1-\frac{\delta}{n^2}$.
\end{lemma}
\begin{proof}
Using triangle inequality, for any $x \in R(s_i)$
\begin{align*}
d(u,x) \leq d(u,s_i) + d(s_i,x) \leq d(u, s_j) - 2\OPT + d(s_i, x) \leq d(u, s_j)  
\end{align*}

So, $\oracle (u,x,u,s_j)$ is \Yes~ with a probability at least $1-p$.  We have: 
$$\E[\ACount(u,s_i,s_j)] =  \sum_{x\in R(s_i)} \E[\mathbf{1} \{ \oracle(u,x,u,s_j) == \Yes\}] \geq (1-p)|R(s_i)|$$

Using Hoeffding's inequality, with a probability of $\exp(-|R(s_i)|(1-p)^2/2) \leq \frac{\delta}{n^2}$ (using $p \leq 0.4$), we have $$\ACount(u,s_i,s_j) \leq (1-p) |R(s_i)|/2$$

We have $\Pr[\ACount(u, s_i, s_j) \leq 0.3 |S|] \leq \Pr[\ACount(u, s_i, s_j) \leq (1-p)|S|/2] $. Therefore, with probability $\frac{\delta}{n^2}$, we have $\ACount(u, s_i, s_j) \leq 0.3 |S|$. Hence, the lemma.
\end{proof}

\begin{lemma}\label{lem:core_appendix}
Suppose $u \in C^*(s_i)$ and for some $s_j \in S$, if $d(s_i, s_j) \geq 6\OPT$, then, Algorithm~\ref{alg:assign_prob} assigns $u$ to center $s_i$ with probability $1-\frac{\delta}{n^2}$.
\end{lemma}
\begin{proof}
As $u \in C^*(s_i)$, we have $d(u, s_i) \leq 2\OPT$. Therefore, 
\begin{align*}
    &d(s_j, u) -  d(s_i, u) \geq d(s_i, s_j) - 2d(s_i, u) \geq 2\OPT \\
    &d(s_j, u) \geq d(s_i, u) + 2\OPT
\end{align*}
From Lemma~\ref{lem:assign_prob_appendix}, we have that if $d(u, s_i) \leq d(u, s_j) - 2\OPT$, then, we will assign $u$ to $s_i$ with probability $1-\frac{\delta}{n^2}$.
\end{proof}




\begin{lemma}\label{lem:assign_approx}
Given a set of centers $S$, every  $u\in V$ is assigned to a cluster $s_i$ such that $d(u,s_i)\leq \min_{s_j\in S} d(u,s_j) + 2\OPT$ with a probability of $1-{1}/{n^2}$.
\end{lemma}
\begin{proof}
From Lemma~\ref{lem:assign_prob_appendix}, we have that a point $u$ is assigned to $s_l$ from $s_m$ if $d(u, s_l) \leq d(u, s_m) - 2\OPT$. If $s_i$ is the final assigned center of $u$, then, for every $s_j$, it must be true that $d(u, s_j) \geq d(u, s_i) - 2\OPT$, which implies $d(u, s_i) \leq \min_{s_j \in S} d(u, s_j) + 2\OPT$. Using union bound over at most $n$ points, we have with a probability of $1-\frac{\delta}{n}$, every point $u$ is assigned as claimed.
\end{proof}



\subsection{Core Calculation}

Consider a cluster $C(s_i)$ with center $s_i$. Let $S_{a}^b$ denote the number of points in the set $| \{x: a\leq d(x,s_i)< b\} |$.

$$\Count(u) = \sum_{x \in C(s_i)} \mathbf{1}\{ \oracle(s_i, x,s_i, u) ==  \No \}$$



\begin{lemma}
Consider any two points $u_1, u_2 \in C(s_i)$ such that $d(u_1,s_i)\leq d(u_2,s_i)$, then $\E[\Count(u_1)] -\E[\Count(u_2)] = (1-2p) S_{d(u_1,s_i)}^{d(u_2,s_i)} $
\end{lemma}
\begin{proof}
For a point $u \in C(s_i)$ 
\begin{align*}
\E[\Count(u)] &= \E \left[\sum_{x \in C(s_i)} \mathbf{1}\{ \mathcal{O}(s_i, x,s_i, u) ==  \No \} \right] \\
              &= S_{0}^{d(u,s_i)} p +  S_{d(u,s_i)}^{\infty} (1-p)
\end{align*}


\begin{align*}
  \E[\Count(u_1)] -\E[\Count(u_2)] &= \Big(S_{0}^{d(u_1, s_i)} p +  S_{d(u_1, s_i)}^{d(u_2, s_i)} (1- p) +  S_{d(u_2, s_i)}^{\infty} (1-p)\Big) - \Big(S_{0}^{d(u_1, s_i)} p +  +  S_{d(u_1, s_i)}^{d(u_2, s_i)} p +  S_{d(u_2, s_i)}^{\infty} (1-p)\Big)\\
    &= (1-2p)S_{d(u_1,s_i)}^{d(u_2,s_i)}
\end{align*}
\end{proof}

\begin{lemma}\label{lem:core_comp}
Consider any two points $u_1, u_2 \in C(s_i)$ such that $d(u_1,s_i)\leq d(u_2,s_i)$ and $|S_{d(u_1,s_i)}^{d(u_2,s_i)}| \geq \sqrt{{100 |C(s_i)|\log(n/\delta)}}$. Then, $\Count(u_1) > \Count(u_2)$ with probability $1-\delta/n^2$.
\end{lemma}
\begin{proof}

Suppose $u_1, u_2 \in C(s_i)$. We have that $\Count(u_1)$ and $\Count(u_2)$ is a sum of $|C(s_i)|$ binary random variables. 

Using Hoeffding's inequality, we have with probability $\exp(-\beta^2/2|C(s_i)|)$ that
\[ \Count(u_1) \leq \E[\Count(u_1)] - \frac{\beta}{2} \]
\[ \Count(u_2) > \E[\Count(u_2)] + \frac{\beta}{2} \]
Using union bound, with probability at least $1-2 \exp(-\beta^2/2|C(s_i)|)$, we can conclude that 
\[ \Count(u_1) - \Count(u_2) > \E[\Count(u_1) - \Count(u_2)] - \beta > (1-2p) S_{d(u_1,s_i)}^{d(u_2,s_i)} - \beta \]

Choosing $\beta = (1-2p)S_{d(u_1,s_i)}^{d(u_2,s_i)}$, we have $\Count(u_1) > \Count(u_2)$ with a probability (for constant $p \leq 0.4$) $$1-2\exp(-(1-2p)^2 \left( S_{d(u_1,s_i)}^{d(u_2,s_i)} \right)^2/2|C(s_i)|) \geq 1-2\exp(-0.02 \left( S_{d(u_1,s_i)}^{d(u_2,s_i)} \right)^2/|C(s_i)|).$$

{Further, simplifying using $ S_{d(u_1,s_i)}^{d(u_2,s_i)} \geq \sqrt{100 |C(s_i)|\log(n/\delta)} $, we get probability of failure is $2\exp({-2\log(n/\delta)}) = O(\delta/n^2)$ }
\end{proof}

\begin{lemma}
If $|C(s_i)| \geq 400\log(n/\delta)$, then, $|R(s_i)| \geq 200\log(n/\delta)$ with probability $1-|C(s_i)|^2 \delta/n^2$.
\end{lemma}
\begin{proof}
From Lemma~\ref{lem:core_comp}, we have that if there are points $u_1, u_2$ with $\sqrt{100|C(s_i)|\log(n/\delta)}$ many points between them, then, we can identify the closer one correctly. When $|C(s_i)| \geq 400 \log(n/\delta)$, we have $\sqrt{100 |C(s_i)|\log(n/\delta)} \geq 200 \log(n/\delta)$ points between every point and the point with the rank $200 \log(n/\delta)$. Therefore, $|R(s_i)| \geq 200 \log(n/\delta)$. Using union bound over all pairs of points in the cluster, we get the claim.
\end{proof}

\begin{lemma}\label{lem:core_approx_appendix}
If $x \in C^*(s_i) $, then, $x \in C(s_i)$ or $x$ is assigned to a cluster $s_j$ such that $d(x, s_j) \leq 8\OPT$.
\end{lemma}
\begin{proof}

If $x \in C^*(s_i)$, we argue that it will be assigned to $C(s_i)$. For the sake of contradiction, suppose $x$ is assigned to a cluster $C(s_j)$ for some $s_j \in S$. We have $d(x, s_i) \leq 2 \OPT$ and  let $d(s_i, s_j) \geq 6\OPT$
\[ d(s_i, s_j) \leq d(s_j, x)+ d(s_i, x)\]
\[ d(s_j, x) \geq 4 \OPT \]
However, we know that $d(s_j, x) \leq d(s_i, x) + 2 \OPT \leq 4\OPT$ from Lemma~\ref{lem:assign_prob_appendix}. We have a contradiction. Therefore, $x$ is assigned to $s_i$. If $d(s_i, s_j) \leq 6\OPT$, we have $d(x, s_j) \leq d(x, s_i) + 2\OPT \leq 8\OPT$. Hence, the lemma.
\end{proof}

\subsection{Farthest point computation}

Let $R(s_i)$ represent the core of the cluster $C(s_i)$ and contains $\Theta(\log(n/\delta))$ points.  We define $\FCount$ for comparing two points $v_i, v_j$ from their centers $s_i, s_j$ respectively. If $s_i \neq s_j$, we let :
$$ \FCount(v_i, v_j) = \sum_{x \in \widetilde{R}(s_i), y\in \widetilde{R}(s_j)} \mathbf{1} \{\oracle(v_i,x,v_j,y) == \Yes \} $$ 
Otherwise, we let $\FCount(v_i, v_j) = \sum_{x \in {R}(s_i)} \mathbf{1} \{\oracle(v_i,x,v_j,x) == \Yes \}$. First, we observe that each of the summation is over $|R(s_i)|$ many terms, because $|\widetilde{R}(s_i)| = \sqrt{|R(s_i)|}$.\\

\begin{lemma} Consider two records $v_i$, $v_j$ in different clusters $C(s_i)$, $C(s_j)$ respectively such that $d(s_i,v_i) < d(s_j,v_j)-4\OPT$ then  $\FCount(v_i,v_j)\geq 0.3 |\widetilde{R}(s_i)| |\widetilde{R}(s_j)|$  with a probability of $1-\frac{\delta}{n^2}$. \label{lem:clustercomp}
\end{lemma}
\begin{proof}
We know $\max_{v_i \in \widetilde{R}(s_i)} d(u, v_i) \leq 2\OPT$ and  $\max_{v_j \in \widetilde{R}(s_j)} d(v_j, s_j) \leq 2\OPT$. 

For a point $x\in R(s_i)$, $y \in R(s_j)$
\begin{align*}
d(v_j,y) &\geq d(s_j,v_j) - d(s_j,y) \\
&> d(v_i, s_i) + 4\OPT - d(s_j, y)\\
&> d(v_i, x) - d(x, s_i) + 4\OPT - d(s_j, y)\\
&> d(v_i,x)
\end{align*}
So, $O(v_i,x,v_j,y)$ is $\No$  with a probability $p$.  As $p \leq 0.4$, we have :
\begin{align*}
  \E[\FCount(v_i,v_j)] &= (1-p)|\widetilde{R}(s_i)| |\widetilde{R}(s_j)|\\
  \Pr[\FCount(v_i, v_j) \leq 0.3|\widetilde{R}(s_i)| |\widetilde{R}(s_j)|] &\leq \Pr[\FCount(v_i, v_j) \leq (1-p)|\widetilde{R}(s_i)| |\widetilde{R}(s_j)|/2]
\end{align*}
From Hoeffding's inequality (with binary random variables), we have with a probability $\exp(-\frac{|\widetilde{R}(s_i)| |\widetilde{R}(s_j)|(1-p)^2}{2}) \leq \frac{\delta}{n^2}$ (using $|\widetilde{R}(s_i)| |\widetilde{R}(s_j)| \ge  12 \log(n/\delta)$, $p < 0.4$) :
$\FCount(v_i,v_j)\leq (1-p)|\widetilde{R}(s_i)| |\widetilde{R}(s_j)|/2$. Therefore, with probability at most $\delta/n^2$, we have, $\FCount(v_i, v_j) \leq 0.3|\widetilde{R}(s_i)| |\widetilde{R}(s_j)|$. 

\end{proof}


In order to calculate the farthest point, we use the ideas discussed in Section~\ref{sec:finding_max} to identify the point that has the maximum distance from its assigned center. As noted in Section~\ref{sec:farthest}, our approximation guarantees dependend on the maximum distance of points in the core from the center. In the next lemma, we show that assuming a maximum distance of a point in the core (See Lemma~\ref{lem:core_approx_appendix}), we can obtain a good approximation for the farthest point.
\begin{lemma}\label{lem:approx_farthest_prob}
Let $\max_{s_j \in S, u \in R(s_j)} d(u, s_j) \leq \alpha$. In every iteration, if the farthest point is at a distance more than $(6\alpha + 3\OPT)$, then, \ApproxFarthest outputs a $(6\alpha/\OPT + 3)$-approximation. Otherwise, the point output is at most $(6\alpha + 3\OPT)$ away.
\end{lemma}
\begin{proof}
The farthest point output \ApproxFarthest is a $6 \alpha$ additive approximation. However, the assignment of points to the cluster also introduces another additive approximation of $2\OPT$, resulting in a total $6 \alpha+2\OPT$ approximation. Suppose in the current iteration, the distance of the farthest point is $\beta \OPT$, then the point output by \ApproxFarthest is at least $\beta \OPT-(6\alpha+2\OPT) $ away. So, the approximation ratio is $\frac{\beta}{\beta-(6\alpha+2\OPT)}$. If $\beta \OPT \geq 6\alpha+3\OPT$, we have $ \frac{\beta \OPT}{\beta \OPT-(6\alpha+2\OPT)} \leq \beta$. As we are trying to minimize the approximation ratio, we set $\beta \OPT = 6\alpha + 3\OPT$ and get the claimed guarantee.
\end{proof}

\subsection{Final Guarantees}
Throughout this section, we assume that $m = \Omega\left(\frac{\log^3(n/\delta)}{\delta} \right)$ for a given failure probability $\delta > 0$.

\begin{lemma}
 Given a current set of centers $S$, and $\max_{v_j \in S, u \in R(v_j)} d(u, v_j) \leq \alpha$, we have : 
\begin{enumerate}
    \item Every point $u$ is assigned to a cluster $C(s_i)$ such that $d(u,s_i) \leq \min_{s_j \in S} d(u,s_j) + 2\OPT$ {using $O(n k\log(n/\delta))$ oracle queries} with probability $1-O(\delta)$.
    \item \textsc{Approx-Farthest} identifies a point $w$ in cluster $C(s_i)$ such that $ \min_{v_j \in S} d(w, v_j) \geq {\max_{v_j \in V}\min_{s_j \in S}d(v_j, s_j)}/(6\alpha/\OPT + 3)$ with probability $1-O(\delta/k)$ {using $O(|\widetilde{V}| \log^3(n/\delta))$  oracle queries.}
\end{enumerate}
\label{lem:maxAssign_prob_appendix}
\end{lemma}
\begin{proof}
(1)  First, we argue that cores are calculated correctly. From Lemma~\ref{lem:core_appendix}, we have that a point $u \in C^*(s_i)$ is assigned to the center correctly $s_i$. Therefore, all the points from $\widetilde{V} \cap C^*(S_i)$ move to $C(S_i)$. As the size of $|C(S_i)| \geq |\widetilde{V} \cap C^*(S_i)| \geq 400 \log(n/\delta)$, we have $|R(s_i)| \geq 200\log(n/\delta)$ with a probability $1-|C(s_i)|^2\delta/n^2$(From Lemma~\ref{lem:core_comp}). Using union bound, we have that all the cores are calculated correctly with a failure probability of $\sum_i |C(s_i)|^2/n^2 = \delta$.\\

For every point, we compare the distance with every cluster center by maintaining a center that is the current closest. From Lemma~\ref{lem:assign_prob_appendix}, we have that the query will fail with a probability of $\delta/n^2$. Using union bound, we have that the failure probability is $O(kn\delta/n^2) = \delta$. From Lemma~\ref{lem:assign_prob_appendix}, we have the approximation guarantee.\\

(2) From Lemma~\ref{lem:approx_farthest_prob}, we have our claim regarding the approximation guarantees. For \ApproxFarthest, we use the parameters $t = 2\log(2k/\delta)$, $l = \sqrt{|\widetilde{V}|}$. As we make $O(|\widetilde{V}| \log^2(k/\delta))$ cluster comparisons using Algorithm~\textsc{ClusterComp} (for \ApproxFarthest), we have that the total number of oracle queries is $O(|\widetilde{V}| \log(n/\delta)\log^2(k/\delta)) = O(|\widetilde{V}| \log^3(n/\delta))$. Using union bound, we have that the failure probability is $O(\delta/k + |\widetilde{V}| \log^2(k/\delta)/n^2) = O(\delta/k)$.
\end{proof}

\begin{theorem}\label{thm:kcenterprob_appendix}[Theorem~\ref{thm:kcenterprob} restated]
Given $p \leq 0.4$, a failure probability $\delta$, and $m = \Omega\left(\frac{\log^3(n/\delta)}{\delta} \right)$. Then, Algorithm~\ref{alg:greedykcenter_prob} achieves a $O(1)$-approximation for the $k$-center objective using $O(nk\log(n/\delta)+\frac{n^2}{m^2}k \log^2(n/\delta))$ oracle queries with probability $1-O(\delta)$.
\end{theorem}
\begin{proof}
Using similar proof as Lemma~\ref{lem:alphabetalem}, we have that the approximation ratio of Algorithm~\ref{alg:greedykcenter_prob} is $4(6 \alpha/\OPT + 3) + 2$. Using $\alpha = 8\OPT$ from  Lemma~\ref{lem:core_approx_appendix}, we have that the approximation factor is $206$. For the first stage, from Lemma~\ref{lem:maxAssign_prob_appendix}, we have that for all the $k$ iterations, the number of oracle queries is $O(|\widetilde{V}|k \log^3(n/\delta))$. Using union bound over $k$ iterations, success probability is $1-O(\delta)$. For the calculation of \emph{core}, the query complexity is $O(|\widetilde{V}|^2 k)$. For assignment, the query complexity is $O(nk \log(n/\delta))$. Therefore, total query complexity is $O(nk\log(n/\delta) + \frac{n}{m} k \log^4(n/\delta) + \frac{n^2}{m^2}k \log^2(n/\delta)) = O(nk\log(n/\delta)+\frac{n^2}{m^2}k \log^2(n/\delta))$.  
\end{proof}

\section{Hierarchical Clustering}\label{app:hierarchical}

\begin{lemma}[Lemma~\ref{lem:single_linkage} restated]
Given a collection of clusters $\mathcal{C}=\{C_1,\ldots,C_r\}$, our algorithm to calculate the closest pair (using Algorithm~\ref{alg:max_adv}) identifies $C_1$ and $C_2$ to merge according to single  linkage objective if $d_{SL}(C_2,C_2) \leq (1+\mu)^3 \min_{C_i,C_j\in \mathcal{C}} d(C_i,C_j)$ with $1-\delta$ probability and requires  $O(r^2\log^2(n/\delta))$ queries.\label{lem:single_linkage_appendix}
\end{lemma}
\begin{proof}
In each iteration, our algorithm considers a list of ${r\choose 2}$ distance values and calculates the closest using Algorithm~\ref{alg:max_adv}. The claim follows from the proof of Theorem~\ref{thm:max_adv}
\end{proof}

Using the same analysis, we get the following result for complete linkage.
\begin{lemma}
Given a collection of clusters $\mathcal{C}=\{C_1,\ldots,C_r\}$, our algorithm to calculate the closest pair (using Algorithm~\ref{alg:max_adv}) identifies $C_1$ and $C_2$ to merge according to complete linkage objective if $d_{SL}(C_2,C_2) \leq (1+\mu)^3 \min_{C_i,C_j\in \mathcal{C}} d(C_i,C_j)$ with $1-\delta$ probability and requires  $O(r^2\log^2(n/\delta))$ queries.
\end{lemma}

\begin{theorem}[Theorem~\ref{thm:optimizedsingle} restated]
In any iteration, suppose the distance between a cluster $C_j\in \mathcal{C}$ and its identified nearest neighbor $\widetilde{C_j}$ is $\alpha$-approximation of its distance from the optimal nearest neighbor, then the distance between pair of clusters merged by  Algorithm~\ref{alg:greedy_hierarchical} is $\alpha (1+\mu)^3$ approximation of the optimal distance between the closest pair of clusters in $\mathcal{C}$ with a probability of $1-\delta$ using $O(n \log^2(n/\delta))$ oracle queries.
\end{theorem}
\begin{proof}
Algorithm~\ref{alg:greedy_hierarchical} iterates over the list of pairs $(C_i,\widetilde{C_i}), \forall C_i\in \mathcal{C}$ and identifies the closest pair using Algorithm~\ref{alg:max_adv}.
The claim follows from the proof of Theorem~\ref{thm:max_adv}
\end{proof}


\end{document}